\documentclass{article}
\usepackage{amsmath,amssymb,theorem}
\usepackage{pstricks}
\usepackage{graphicx}
\usepackage{verbatim}
\graphicspath{{semfinal/}}

\usepackage{color}
\usepackage{makeidx}

\theorembodyfont{\upshape}

 \makeatletter
 \@addtoreset{equation}{section}
 \makeatother

 \newtheorem{ittheorem}{Theorem}
 \newtheorem{itlemma}{Lemma}
 \newtheorem{itproposition}{Proposition}
 \newtheorem{itdefinition}{Definition}

 \newtheorem{itremark}{Remark}
 \newtheorem{itclaim}{Claim}
 \newtheorem{itcorollary}{\bf Corollary}

 \newenvironment{theorem}{\addtocounter{equation}{1}
 \begin{ittheorem}}{\end{ittheorem}}

 \newenvironment{lemma}{\addtocounter{equation}{1}
 \begin{itlemma}}{\end{itlemma}}

 \newenvironment{proposition}{\addtocounter{equation}{1}
 \begin{itproposition}}{\end{itproposition}}

 \newenvironment{definition}{\addtocounter{equation}{1}
 \begin{itdefinition}}{\end{itdefinition}}

 \newenvironment{remark}{\addtocounter{equation}{1}
 \begin{itremark}}{\end{itremark}}

 \newenvironment{claim}{\addtocounter{equation}{1}
 \begin{itclaim}}{\end{itclaim}}

 \newenvironment{proof}{\noindent {\bf Proof.\,}
 }{\hspace*{\fill}$\qed$\medskip}

 \newenvironment{corollary}{\addtocounter{equation}{1}
 \begin{itcorollary}}{\end{itcorollary}}

 \newcommand{\be}[1]{\begin{eqnarray*}\label{#1}}
 \newcommand{\ee}{\end{eqnarray*}}

 \newcommand{\bl}[1]{\begin{lemma}\label{#1}}
 \newcommand{\el}{\end{lemma}}

 \newcommand{\br}[1]{\begin{remark}\label{#1}}
 \newcommand{\er}{\end{remark}}

 \newcommand{\bt}[1]{\begin{theorem}\label{#1}}
 \newcommand{\et}{\end{theorem}}

 \newcommand{\bd}[1]{\begin{definition}\label{#1}}
 \newcommand{\ed}{\end{definition}}

 \newcommand{\bcl}[1]{\begin{claim}\label{#1}}
 \newcommand{\ecl}{\end{claim}}

 \newcommand{\bp}[1]{\begin{proposition}\label{#1}}
 \newcommand{\ep}{\end{proposition}}

 \newcommand{\bc}[1]{\begin{corollary}\label{#1}}
 \newcommand{\ec}{\end{corollary}}

 \newcommand{\bpr}{\begin{proof}}
 \newcommand{\epr}{\end{proof}}

 \newcommand{\bi}{\begin{itemize}}
 \newcommand{\ei}{\end{itemize}}

 \newcommand{\ben}{\begin{enumerate}}
 \newcommand{\een}{\end{enumerate}}

\def\Al{{\mathcal A}^\ell}
\def\Alm{\big({\mathcal A}^\ell\big)^m}
\def\alm{({\mathcal A}^\ell)^m}
\def\zl{\{\,0,\dots,\ell\,\}}
\def\ul{\{\,1,\dots,\ell\,\}}
\def\zm{\{\,0,\dots,m\,\}}
\def\zlm{\{\,0,\dots,\ell\,\}^m}

\def\wtau{\widetilde{\tau}}
\def\oti{o^\theta_{\text{enter}}}
\def\oto{o^\theta_{\text{exit}}}
\def\ota{o^\ell_{\text{exit}}}
\def\otb{o^1_{\text{exit}}}
\def\lk{\ell(1-1/\kappa)}
\def\lK{\ell(1-\frac{1}{\kappa})}
\def\lk{\ell_\kappa}
\def\lK{\ell_\kappa}

\def\um{\{\,1,\dots,m\,\}}
\def\umu{\{\,1,\dots,m-1\,\}}
\def\pml{{\mathcal P}(m,\ell +1)}
\def\pml{{\mathcal P}^m_{\ell +1}}
\def\var{\mathop{\rm Var}} 
\def\ltq{\ell^{3/4}} 
\def\lq{\ell^{2}} 

\def\lkep{\lk(1-\ve')}

\def\bp{{\overline{p}}}

\def\bp{{\overline{p}}}
\def\exa{\exp(-a)}
\def\exa{e^{-a}}

 \def \ba {\begin{array}}
 \def \ea {\end{array}}

 \def \qed {{\heartsuit\hfill}}
 
 \def \R {{\mathbb R}}

 \def \cN {{\cal N}}
 \def \cNm {{\cal N}}
 \def \cS {{\cal S}}
 \def \cE {{\cal E}}

 \def \cR {{\cal R}}
\def\ve{\varepsilon}

 \def \cB {{\cal B}}
 \def \cM {{\cal M}}
 \def \cW {{\cal W}}

\def \qed {{\square\hfill}}

\let\F=E     
\let\O=\Omega
 \def\cB{{\cal B}}  
\def\cE{{\cal E}}   
\def\cI{{\cal I}}   
\def\cM{{\cal M}} \def\cN{{\cal N}} \def\cO{{\cal O}} 
 \def\cR{{\cal R}} \def\cS{{\cal S}} \def\cT{{\cal T}}

  \def\cW{{\cal W}} 
 
 \def \cW {{\cal W}^*}
 \def \cMH {{\cal M}_H}
 \def \cSH {{\cal S}_H}

\def \qed {{\square\hfill}}

\def\R{{\mathbb R}}

\def\eqref#1{(\ref{#1})}

\def\card{\text{card}\,}
\def\anc{\text{ancestor}}

\def\wild{\text{Master}}
\def\vwild{\text{Variance}}

\makeindex

\begin{document}
%%%%%%%%%%%%%

%\color{blue}
%\title{Critical population size and error threshold 
%\\ on the
%sharp peak landscape}
%\title{Critical population size and error threshold on the
%sharp peak landscape for the Moran model}
\title{Critical population and error threshold \\ on the
sharp peak landscape \\ for a Moran model}
%\title{Critical population size and error threshold \\ on the
%sharp peak landscape \\ for a stochastic Eigen model}
%\title{Critical population size for quasispecies emergence}
%\title{Critical population for quasispecies emergence}

 \author{
Rapha\"el Cerf
%\footnote{\texttt{rcerf@math.u-psud.fr}}
\\
Universit\'e Paris Sud and IUF
%Math\'ematique, B\^atiment~$425$\\
%\small
%91405 Orsay Cedex--France\\
%\small
}

\maketitle

%\centerline{Universit\'e Paris Sud and Institut Universitaire de France}

%\renewcommand{\abstractname}{}

\begin{abstract}
\noindent
The goal of this work is to propose
a finite population counterpart to
Eigen's model, which incorporates stochastic effects. 
We consider a Moran model describing the evolution
of a population of size~$m$ of chromosomes of length~$\ell$
over an alphabet of cardinality $\kappa$.
%$\{\,A,T,G,C\,\}$. 
The mutation probability per locus is $q$.
We deal only with the sharp peak landscape: the replication rate is
$\sigma>1$ for the master sequence and $1$ for the other sequences.
We study the equilibrium distribution of the process in the regime where
$$\displaylines{
\ell\to +\infty\,,\qquad m\to +\infty\,,\qquad q\to 0\,,\cr
{\ell q} \to a\in ]0,+\infty[\,,
\qquad\frac{m}{\ell}\to\alpha\in [0,+\infty]\,.}$$
We obtain an equation 
$\alpha\,\phi(a)=\ln\kappa$ 
in the parameter space $(a,\alpha)$
separating the regime where the equilibrium population is totally random from the
regime where a quasispecies is formed. We observe the existence of a critical
population size necessary for a quasispecies to emerge and
we recover the finite population counterpart of the error threshold.
Moreover, in the limit of very small mutations, we obtain a lower bound on
the population size allowing
the emergence
of a quasispecies:
if $\alpha< \ln\kappa/\ln\sigma$ then
the equilibrium population is totally random, and
a quasispecies can be formed
only when 
$\alpha\geq \ln\kappa/\ln\sigma$. 
Finally, in the limit of very large populations, we recover
an error catastrophe reminiscent of
Eigen's model:
if $\sigma\exa\leq 1$
then the equilibrium population is totally random, and
a quasispecies can be formed
only when $\sigma\exa>1$.
These results are supported by computer simulations. 
\end{abstract}

%\part{Models and results}

%We propose here a version of the famous Eigen model for finite populations.
%Eigen's model is originally formulated for an infinite population and
%several authors have proposed versions of this model in finite populations.

\section{Introduction.}
In his famous paper \cite{EI1}, Eigen introduced a model for the evolution
of a population of macromolecules. In this model, the macromolecules
replicate themselves, yet the replication mechanism is subject to errors
caused by mutations. These two basic mechanisms are described by a family
of chemical reactions. The replication rate of a macromolecule is governed
by its fitness.
A fundamental discovery of Eigen is the existence of an error threshold on
the sharp peak landscape.
If the mutation rate exceeds a critical value, called the error threshold,
then, at equilibrium, the population is completely random.
If the mutation rate is below the error threshold, 
then, at equilibrium, the population contains a positive fraction of the
master sequence (the most fit macromolecule) and a cloud of mutants which
are quite close to the master sequence.
This specific distribution of individuals is called a quasispecies. This notion
has been further investigated by
Eigen, McCaskill and Schuster \cite{ECS1} and it had a profound impact
on the understanding of molecular evolution \cite{ESTE}.
It has been argued that, at the population level, evolutionary processes 
select quasispecies rather than single individuals. Even more importantly, 
this theory is supported by experimental studies \cite{DBMJ}.
Specifically, it seems
that some RNA viruses evolve with a rather high mutation rate, which is
adjusted to be close to an error threshold.
It has been suggested that
this is the case for the HIV virus \cite{TBVD}. 
Some promising antiviral strategies consist in using mutagenic drugs
that induce an error catastrophe \cite{ADL,CCA}.
A similar error catastrophe
could also play
a role in the development of some cancers \cite{SODE}.

Eigen's model was initially designed to understand a population of
macromolecules governed by a family of chemical reactions. In this setting,
the number of molecules is huge, and there is a finite number of types
of molecules. From the start, this model is formulated for an infinite
population and the evolution is deterministic (mathematically, it is a
family of differential equations governing the time 
evolution of the densities of each type of macromolecule). The error
threshold appears when the number of types goes to $\infty$.
This creates a major obstacle if one wishes to extend the notions of 
quasispecies
and error threshold to genetics. Biological populations are finite,
and even if they are large so that they might be considered infinite
in some approximate scheme, it is not coherent to consider situations
where the size of the population is much larger than the number of possible
genotypes. Moreover, it has long been recognized that random effects play
a major role in the genetic evolution of populations \cite{Kimu}, 
yet they are ruled
out from the start in a deterministic infinite population model.
Therefore, it is crucial to develop a finite population counterpart to
Eigen's model, which incorporates stochastic effects. 
This problem is already discussed by
Eigen, McCaskill and Schuster \cite{ECS1} and more recently by Wilke \cite{Wilke}.
Numerous works have attacked this issue:
Demetrius, Schuster and Sigmund \cite{Deme},
McCaskill \cite{Cas1},
Gillespie \cite {GI},
Weinberger \cite{WE}.
Nowak and Schuster \cite{NS} constructed a birth and death model to approximate
Eigen's model. This birth and death model plays a key role in our analysis,
as we shall see later.
Alves and Fontanari \cite{AF} study how the error threshold depends on the
population in a simplified model.
More recently, Musso \cite{MUS} and
Dixit, Srivastava, Vishnoi \cite{DSV}
considered finite population models which approximate Eigen's
model when the population size goes to $\infty$. These models are
variants of the classical Wright--Fisher model of population genetics.
Although this is an interesting approach, it is already a delicate matter
to prove the convergence of these models towards Eigen's model.
We adopt here a different strategy. Instead of trying to prove that
some finite population model converges in some sense to Eigen's model, 
we try to prove directly in the finite model an error threshold 
phenomenon. To this end, we look for the simplest possible model, and we
end up with a Moran model.
The model we choose here is not particularly original, the contribution
of this work is rather to show a way to analyze this kind of finite population
models.

We consider a population of size~$m$ of chromosomes of length~$\ell$
over the alphabet
$\{\,A,T,G,C\,\}$. The evolution of the population is governed by two 
antagonistic effects, namely mutation and replication.
Mutations occur randomly and independently at each locus with 
probability~$q$. The replication rate of a chromosome is given by its fitness.
We consider only the sharp peak landscape: there is one specific sequence,
called the master sequence, whose fitness is $\sigma>1$, and all the
other sequences have fitness equal to~$1$.
The mutations drive the population towards a totally random 
state, while the replication favors the master sequence.
These two effects interact in a complicated way in the dynamics
and it is extremely difficult to analyze precisely the time evolution
of such a model.
Let us focus on the equilibrium distribution of the process.
A fundamental problem is to determine the law of the number of
copies of the master sequence present in the population at equilibrium.
If we keep the parameters $m,\ell,q$ fixed, there is little hope to
get useful results.
In order to simplify the picture, we consider an adequate asymptotic regime.
In Eigen's model, the population size is infinite from the start.
The error threshold appears when $\ell$ goes to $\infty$ and $q$ goes
to $0$ in a regime where $\ell q=a$ is kept constant.
We wish to understand the influence of the population size~$m$, thus we
use a different approach and we consider the following regime.
We send simultaneously 
$m,\ell$ to $\infty$ and $q$ to $0$
and we try to understand the respective influence of each parameter on the
equilibrium law of the master sequence.
By the ergodic theorem, the average number 
of copies of the master sequence at equilibrium is equal to the limit,
as the time goes to $\infty$, of the time average
of the number 
of copies of the master sequence present through the whole evolution
of the process.
In the finite population model, the number of copies of the master sequence
fluctuates with time.
%We analyze these fluctuations with the help of the following heuristics.
Our analysis of these fluctuations relies on the following heuristics.
Suppose that the process starts with a population of size $m$
containing exactly one master sequence.
The master sequence is likely to invade the whole population and
become dominant. Then the master sequence will be present in the population
for a very long
time without interruption.
We call this time the {\bf persistence} time of the master sequence.
The destruction of all the master sequences of the population is quite unlikely,
nevertheless it will happen and the process will eventually
land in the neutral region
consisting of the populations devoid of master sequences.
The process will wander randomly throughout this region for a very long time.
We call this time the
{\bf discovery} time of the master sequence.
Because the cardinality of the possible
genotypes is enormous, the master sequence
is difficult to discover, nevertheless the mutations will eventually 
succeed 
%in discovering it 
and the process will start again with
a population 
containing exactly one master sequence.
If, on average, the discovery time is much larger than the persistence time, then the 
equilibrium state will be totally random, while a quasispecies will be formed if
the persistence time is much larger than the discovery time.
Let us illustrate this idea in a very simple model.
%\vskip-20pt
%\begin{comment}
\begin{figure}[h]
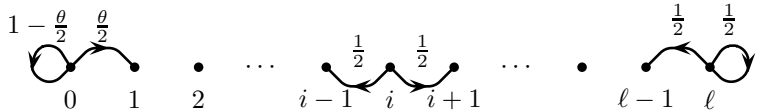

\centering
{
\psset{unit=0.85}
\pspicture(-1,-1)(11,1)
\psdots[dotscale=1](0,0)(1,0)(2,0)(4,0)(5,0)(6,0)(8,0)(9,0)(10,0)
\rput(0,-0.5){$0$}
\rput(1,-0.5){$1$}
\rput(2,-0.5){$2$}
\rput(3,0){$\cdots$}
\rput(4,-0.5){$i-1$}
\rput(5,-0.5){$i$}
\rput(6,-0.5){$i+1$}
\rput(7,0){$\cdots$}
%\rput(8,-0.5){$\ell-2$}
\rput(9,-0.5){$\ell-1$}
\rput(10,-0.5){$\ell$}
\psset{arrowscale=1.5}
%\psline[linewidth=2pt,linearc=0]{->}(0,0)(0.5,1)(1,0)
%\psline[linewidth=1pt,linearc=.25]{->}(0,0)(0.5,0.5)(1,0)
%\psline[linewidth=1pt,linearc=.25]{->}(1,0)(1.5,-0.5)(2,0)
\psline[linewidth=1pt,linearc=.25]{->}(0,0)(0.3,0.3)(0.5,0.3)(0.6,0.3)
\psline[linewidth=1pt,linearc=.25]{-}(0.5,0.3)(0.7,0.3)(1,0)
%\psline[linewidth=1pt,linearc=.25]{->}(1,0)(0.7,-0.3)(0.5,-0.3)(0.4,-0.3)
%\psline[linewidth=1pt,linearc=.25]{-}(0.5,-0.3)(0.3,-0.3)(0,0)
%\psline[linewidth=1pt,linearc=.25]{->}(1,0)(1.3,-0.3)(1.5,-0.3)(1.6,-0.3)
%\psline[linewidth=1pt,linearc=.25]{-}(1.5,-0.3)(1.7,-0.3)(2,0)
%\rput(0.5,-0.75){$\frac{1}{2}$}
%\rput(1.5,-0.75){$\frac{1}{2}$}
\rput(0.5,0.65){$\frac{\theta}{2}$}
\rput(-0.5,0.65){$1-\frac{\theta}{2}$}
%\rput(0.5,-0.75){$0.5$}
\psline[linewidth=1pt,linearc=.25]{->}(5,0)(4.7,-0.3)(4.5,-0.3)(4.4,-0.3)
\psline[linewidth=1pt,linearc=.25]{-}(4.5,-0.3)(4.3,-0.3)(4,0)
\psline[linewidth=1pt,linearc=.25]{->}(5,0)(5.3,-0.3)(5.5,-0.3)(5.6,-0.3)
\psline[linewidth=1pt,linearc=.25]{-}(5.5,-0.3)(5.7,-0.3)(6,0)
\rput(4.5,0.25){$\frac{1}{2}$}
\rput(5.5,0.25){$\frac{1}{2}$}
%\psline[linewidth=1pt,linearc=.25]{->}(9,0)(8.7,-0.3)(8.5,-0.3)(8.4,-0.3)
%\psline[linewidth=1pt,linearc=.25]{-}(8.5,-0.3)(8.3,-0.3)(8,0)
%\psline[linewidth=1pt,linearc=.25]{->}(9,0)(9.3,-0.3)(9.5,-0.3)(9.6,-0.3)
%\psline[linewidth=1pt,linearc=.25]{-}(9.5,-0.3)(9.7,-0.3)(10,0)
%\rput(8.5,-0.75){$\frac{1}{2}$}
%\rput(9.5,-0.75){$\frac{1}{2}$}
\psline[linewidth=1pt,linearc=0.25]{->}(0,0)(-.25,.4)(-0.55,0.275)(-0.61,-0.06)
\psline[linewidth=1pt,linearc=0.25]{-}(-0.605,-0.)(-0.6,-0.1)(-.3,-.3)(0,0)
\psline[linewidth=1pt,linearc=0.25]{->}(10,0)(10.2,.3)(10.35,.35)(10.6,0.28)(10.6,-0.08)
\psline[linewidth=1pt,linearc=0.25]{-}(10.6,-0.)(10.6,-0.1)(10.3,-.3)(10,0)
\psline[linewidth=1pt,linearc=.25]{->}(10,0)(9.7,0.3)(9.5,0.3)(9.4,0.3)
\psline[linewidth=1pt,linearc=.25]{-}(9.5,0.3)(9.3,0.3)(9,0)
\rput(10.3,0.75){$\frac{1}{2}$}
\rput(9.5,0.75){$\frac{1}{2}$}
%\psline[linewidth=1pt,linearc=.25]{->}(1,0)(1.3,-0.5)(1.5,-0.5)(1.6,-0.5)
%\psline[linewidth=1pt,linearc=.25]{->}(1.5,-0.5)(1.7,-0.5)(2,0)
%  \psaxes[Ox=1,dy=1,Dy=2]{->}(0,0)(-1,-1)(6,5)
\endpspicture
}
\caption{Random walk example}
%\label{Random walk} 
\end{figure} 
%\end{comment}

\noindent
We consider the random walk on $\zl$ with the transition 
probabilities depending on a parameter $\theta$ given by:
$$\displaylines{p(0,1)=\frac{\theta}{2}\,,\quad
p(0,0)=1-\frac{\theta}{2}\,,\quad
p(\ell,\ell-1)= p(\ell,\ell)=\frac{1}{2}\,,\cr
p(i,i-1)= p(i,i+1)=\frac{1}{2}\,,\quad 1\leq i\leq \ell-1\,.
}$$
The integer $\ell$ is large and the parameter $\theta$ is small. Hence
the walker spends its time either wandering in $\ul$
or being trapped in $0$. The state $0$ plays the role of the quasispecies
while the set $\ul$ plays the role of the neutral region.
With this analogy in mind, the persistence time is the expected time of exit from $0$,
it is equal to $2/\theta$. The discovery time is the expected time needed to
discover $0$ starting for instance from $1$, it is equal to $2\ell$. 
The equilibrium law of the walker is the probability measure $\mu$ given by
$$\mu(0)\,=\,\frac{1}{1+\theta \ell}\,,\qquad
\mu(1)\,=\,\cdots\,=\,\mu(\ell)\,=\,\frac{\theta}{1+\theta \ell}\,.
%\mu(i)\,=\,\frac{\theta}{1+\theta \ell}\,,\quad 1\leq i\leq \ell\,.
$$
We send
$\ell$ to $\infty$ and $\theta$ to $0$ simultaneously.
If $\theta\ell$ goes to $\infty$, the entropy factors wins and $\mu$
becomes totally random.
If $\theta\ell$ goes to $0$, the selection drift wins and $\mu$
converges to the Dirac mass at $0$.

In order to implement the previous heuristics, we have
to estimate the persistence time and the discovery time of the master sequence
in the Moran model.
For the persistence time, we rely on a classical computation from mathematical
genetics.
Suppose we start with a population containing $m-1$ copies of the master sequence
and another non master sequence. The non master sequence is very unlikely
to invade the whole population, yet it has a small probability to do so, called
the fixation probability. If we neglect the mutations, standard computations yield that, in a population of
size~$m$, if the master sequence has a selective advantage of $\sigma$, the 
fixation probability of the non master sequence is roughly of order
$1/{\sigma}^m$ (see for instance \cite{NO1}, section 6.3).
%$$\Big(\frac{1}{\sigma}\Big)^m\,.$$
Now the persistence time can be viewed as 
the time needed for non master sequences to invade the population.
This time is approximately equal to the inverse of the 
fixation probability of the non master sequence, that is of order
$\sigma^m$.
% is roughly of order of the fix
%the fixation time of non master sequences,
%thus the persistence time should be of order $\sigma^m$.
For the discovery time, there is no miracle: before discovering the
master sequence, the process is likely 
to explore a significant portion of the
genotype space, hence the discovery time should be of order
$$\card\{\,A,T,G,C\,\}^\ell\,=\,4^\ell\,.$$
These simple heuristics indicate that the persistence time depends
on the selection drift, while the discovery time depends on the
spatial entropy.
Suppose that we send 
$m,\ell$ to $\infty$
simultaneously.
If the discovery time is much larger than the persistence time,
then the population will be neutral most of the time
and the fraction of the master sequence at equilibrium will be null.
If the persistence time is much larger than the discovery time,
then the population will be invaded by the master sequence
most of the time
and the fraction of the master sequence at equilibrium will be positive.
Thus the master sequence vanishes in the regime
$$m,\ell\to+\infty\,,\qquad \frac{m}{\ell}\to 0\,,$$
while a quasispecies might be formed in the regime
$$m,\ell\to+\infty\,,\qquad \frac{m}{\ell}\to +\infty\,.$$
This leads to an interesting feature, namely the existence of a critical
population size for the emergence of a quasispecies.
For chromosomes of length~$\ell$, a quasispecies can be formed only if the
population size $m$ is such that
ratio $m/\ell$ is large enough.
In order to go further, we must put the heuristics on a firmer ground and we
should take the mutations into account when estimating the persistence time.
The main problem is to obtain
finer estimates on the persistence and discovery times.
%Indeed, the previous considerations do not 
We cannot compute explicitly the laws of these random times, so we will
%We cannot compute explicitly the laws of the persistence time
%and the discovery time, so we will
compare the Moran model with simpler processes.
%\begin{comment}
\begin{figure}[h]
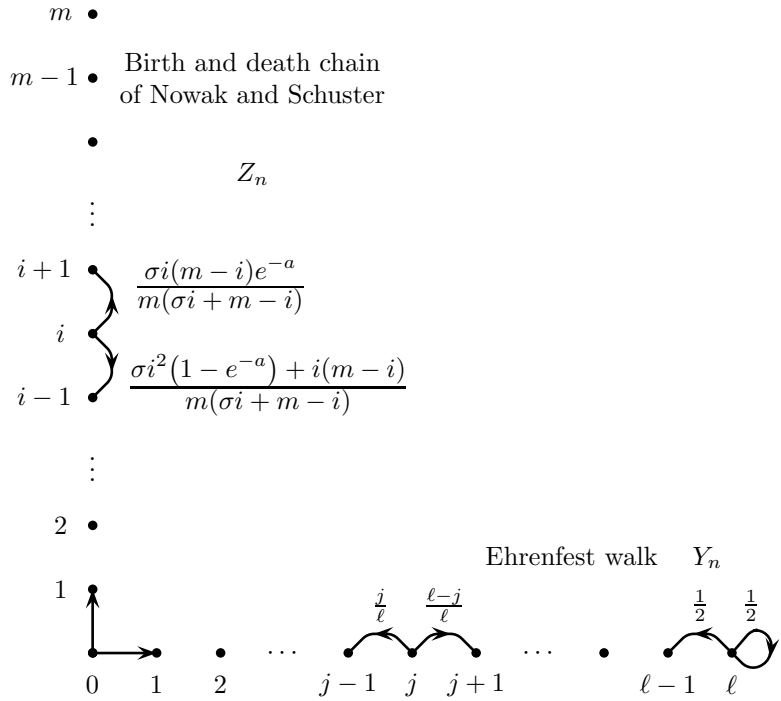

\centering
\psset{unit=0.85}
\psset{arrowscale=1.5}
\pspicture(-1,-1)(11,11)
\psdots[dotscale=1](0,0)(1,0)(2,0)(4,0)(5,0)(6,0)(8,0)(9,0)(10,0)
\psdots[dotscale=1](0,0)(0,1)(0,2)(0,4)(0,5)(0,6)(0,8)(0,9)(0,10)
\rput(0,-0.5){$0$}
\rput(1,-0.5){$1$}
\rput(2,-0.5){$2$}
\rput(3,0){$\cdots$}
\rput(4,-0.5){$j-1$}
\rput(5,-0.5){$j$}
\rput(6,-0.5){$j+1$}
\rput(7,0){$\cdots$}
%\rput(8,-0.5){$\ell-2$}
\rput(9,-0.5){$\ell-1$}
\rput(10,-0.5){$\ell$}
%\psline[linewidth=2pt,linearc=0]{->}(0,0)(0.5,1)(1,0)
%\psline[linewidth=1pt,linearc=.25]{->}(0,0)(0.5,0.5)(1,0)
%\psline[linewidth=1pt,linearc=.25]{->}(1,0)(1.5,-0.5)(2,0)
%\psline[linewidth=1pt,linearc=.25]{->}(0,0)(0.3,0.3)(0.5,0.3)(0.6,0.3)
%\psline[linewidth=1pt,linearc=.25]{-}(0.5,0.3)(0.7,0.3)(1,0)
%\psline[linewidth=1pt,linearc=.25]{->}(1,0)(0.7,-0.3)(0.5,-0.3)(0.4,-0.3)
%\psline[linewidth=1pt,linearc=.25]{-}(0.5,-0.3)(0.3,-0.3)(0,0)
%\psline[linewidth=1pt,linearc=.25]{->}(1,0)(1.3,-0.3)(1.5,-0.3)(1.6,-0.3)
%\psline[linewidth=1pt,linearc=.25]{-}(1.5,-0.3)(1.7,-0.3)(2,0)
%\rput(0.5,-0.75){$\frac{1}{2}$}
%\rput(1.5,-0.75){$\frac{1}{2}$}
%\rput(0.5,0.65){$\theta$}
%\rput(-0.5,0.65){$1-\theta$}
%\rput(0.5,-0.75){$0.5$}
\psline[linewidth=1pt,linearc=.25]{->}(5,0)(4.7,0.3)(4.5,0.3)(4.4,0.3)
\psline[linewidth=1pt,linearc=.25]{-}(4.5,0.3)(4.3,0.3)(4,0)
\psline[linewidth=1pt,linearc=.25]{->}(5,0)(5.3,0.3)(5.5,0.3)(5.6,0.3)
\psline[linewidth=1pt,linearc=.25]{-}(5.5,0.3)(5.7,0.3)(6,0)
\rput(5.5,0.75){$\frac{\ell -j}{\ell}$}
\rput(4.5,0.75){$\frac{j}{\ell}$}
\rput(8,1.5){Ehrenfest walk $\quad Y_n$}
\rput(2.5,9.250){Birth and death chain}
\rput(2.5,8.75){of Nowak and Schuster}
\rput(2.5,7.5){$Z_n$}
%\psline[linewidth=1pt,linearc=.25]{->}(9,0)(8.7,-0.3)(8.5,-0.3)(8.4,-0.3)
%\psline[linewidth=1pt,linearc=.25]{-}(8.5,-0.3)(8.3,-0.3)(8,0)
%\psline[linewidth=1pt,linearc=.25]{->}(9,0)(9.3,-0.3)(9.5,-0.3)(9.6,-0.3)
%\psline[linewidth=1pt,linearc=.25]{-}(9.5,-0.3)(9.7,-0.3)(10,0)
%\rput(8.5,-0.75){$\frac{1}{2}$}
%\rput(9.5,-0.75){$\frac{1}{2}$}
%\psline[linewidth=1pt,linearc=0.25]{->}(0,0)(-.3,.3)(-0.6,0.2)(-0.6,-0.05)
%\psline[linewidth=1pt,linearc=0.25]{-}(-0.6,-0.)(-0.6,-0.1)(-.3,-.3)(0,0)
\psline[linewidth=1pt,linearc=0.25]{->}(10,0)(10.3,.35)(10.65,0.39)(10.6,-0.07)
\psline[linewidth=1pt,linearc=0.25]{-}(10.6,-0.)(10.6,-0.1)(10.3,-.3)(10,0)
\psline[linewidth=1pt,linearc=.25]{->}(10,0)(9.7,0.3)(9.5,0.3)(9.4,0.3)
\psline[linewidth=1pt,linearc=.25]{-}(9.5,0.3)(9.3,0.3)(9,0)
\rput(10.3,0.75){$\frac{1}{2}$}
\rput(9.5,0.75){$\frac{1}{2}$}
\rput(-.5,1){$1$}
\rput(-.5,2){$2$}
\rput(0,3){$\vdots$}
\rput(0,7){$\vdots$}
\rput(-.75,4){$i-1$}
\rput(-.5,5){$i$}
\rput(-.75,6){$i+1$}
\rput(-.75,9){$m-1$}
\rput(-.5,10){$m$}
\psline[linewidth=1pt,linearc=.25]{->}(0,5)(0.3,4.7)(0.3,4.5)(0.3,4.4)
\psline[linewidth=1pt,linearc=.25]{-}(0.3,4.5)(0.3,4.3)(0,4)
\psline[linewidth=1pt,linearc=.25]{->}(0,5)(0.3,5.3)(0.3,5.5)(0.3,5.6)
\psline[linewidth=1pt,linearc=.25]{-}(0.3,5.5)(0.3,5.7)(0,6)
\rput(2.,5.75){$
\frac{\displaystyle\sigma i(m-i) e^{-a} }
{\displaystyle m(\sigma i + m -i)}$}
\rput(2.75,4.2){$
\frac{\displaystyle \sigma i^2 
\big(1-e^{-a} \big)+
%\big(1-\exp-a \big)+
i(m-i)
}{\displaystyle m(\sigma i + m -i)}
$}
\psline[linewidth=1pt,linearc=0.25]{->}(0,0)(0,1)
\psline[linewidth=1pt,linearc=0.25]{->}(0,0)(1,0)
%\psline[linewidth=1pt,linearc=.25]{->}(1,0)(1.3,-0.5)(1.5,-0.5)(1.6,-0.5)
%\psline[linewidth=1pt,linearc=.25]{->}(1.5,-0.5)(1.7,-0.5)(2,0)
%  \psaxes[Ox=1,dy=1,Dy=2]{->}(0,0)(-1,-1)(6,5)
\endpspicture
\caption{Approximating process}
\end{figure} 
%\end{comment}

\noindent
In the non neutral populations, we shall compare the process with a
birth and death process 
$(Z_n)_{n\geq 0}$ on $\zm$, which is precisely the one 
introduced by Nowak and Schuster \cite{NS}.
The value $Z_n$ approximates the number of copies of the master sequence
present in the population.
For birth and death processes, explicit formula are available and we obtain
that, if 
$\ell,m\to +\infty,\,q\to 0,\,
{\ell q} \to a\in]0,+\infty[$, then
$$\text{\bf persistence time}\,
\,\mathop{\sim}\,
\exp\big(m\,\phi(a)
\big)\,,$$
%\mathop{\sim}\limits_{\genfrac{}{}{0pt}{1}{\ell,m\to\infty}
%{q\to 0,\,
%{\ell q} \to a,\,
%%\frac{\scriptstyle m}{\scriptstyle \ell}\to\alpha
%}}
where 
$$\phi(a)\,=\,
\frac
{ \displaystyle \sigma(1-e^{-a})
\ln\frac{\displaystyle\sigma(1-e^{-a})}{\displaystyle\sigma-1}
+\ln(\sigma e^{-a})}
{ \displaystyle (1-\sigma(1-e^{-a})) }
\,.
$$
In the neutral populations, we shall replace the process by a
random walk on
$\{\,A,T,G,C\,\}^\ell\,=\,4^\ell$. The lumped version of this random
walk behaves like
an
Ehrenfest process
$(Y_n)_{n\geq 0}$ on $\zl$ (see \cite{BI} for a nice review).
The value $Y_n$ represents the distance of the walker to the
master sequence.
A celebrated theorem of Kac from 1947 \cite{Kac}, which helped to resolve
a famous paradox of statistical mechanics, yields that, 
when $\ell\to\infty$,
$$\text{\bf discovery time}\,\sim\, 4^\ell
\,.$$
Thus the Moran process is approximated by the process on
$$\Big(\zl\times \{\,0\,\}\Big)\,\cup
\Big(\{\,0\,\} \times \zm\Big)$$ 
described loosely as follows.
On 
$\zl\times \{\,0\,\}$, the process follows the dynamics of the Ehrenfest urn.
On 
$\{\,0\,\} \times \zm$, the process follows the dynamics of the 
birth and death process of Nowak and Schuster \cite{NS}.
When in $(0,0)$, the process can jump to either axis.
With this simple heuristic
picture, we recover all the features of our main result.
We suppose that 
%$\ell$ goes to $\infty$, $m$ goes to $\infty$ and $q$ goes to $0$ 
$$\ell\to +\infty\,,\qquad m\to +\infty\,,\qquad q\to 0\,,$$
in such a way that
$${\ell q} \to a\in ]0,+\infty[\,,
\qquad\frac{m}{\ell}\to\alpha\in [0,+\infty]\,.$$
The critical curve is then defined by the equation
$$
\text{\bf discovery time}\,\sim\,
\text{\bf persistence time}$$
which can be rewritten as
%$\alpha\,\phi(a)=\ln 4$.
$$
\alpha\,\phi(a)
%\frac
%{ \displaystyle \sigma(1-e^{-a})
%%\ln\frac{\displaystyle\sigma(1-e^{-a})}{\displaystyle\sigma-1}
%\big(\ln({\sigma(1-e^{-a})})-\ln({\sigma-1})\big)
%+\ln(\sigma e^{-a})}
%{
%\displaystyle (1-\sigma(1-e^{-a}))
 %}
\,=\,
 %\ell\ln 4
\ln 4
\,.
$$
%We obtain the equation of a critical curve 
This way we obtain an equation 
%$\alpha=\phi(a)$
in the parameter space $(a,\alpha)$
separating the regime where the equilibrium population is totally random from the
regime where a quasispecies is formed. We observe the existence of a critical
population size necessary for a quasispecies to emerge and
we recover the finite population counterpart of the error threshold.
Moreover, in the regime of very small mutations, we obtain a lower bound on
the population size allowing
the emergence
of a quasispecies:
if $\alpha< \ln 4/\ln\sigma$ then
the equilibrium population is totally random, and
a quasispecies can be formed
only when 
$\alpha\geq \ln 4/\ln\sigma$. 
Finally, in the limit of very large populations, we recover
an error catastrophe reminiscent of
Eigen's model:
if $\sigma\exa\leq 1$
then the equilibrium population is totally random, and
a quasispecies can be formed
only when $\sigma\exa>1$.
These results are supported by computer simulations. 
%The programs are written in $C$
%with the help of the GNU scientific library. The graphical output is generated
%with the help of the Gnuplot program.
The good news is that,
already for small values of $\ell$, the simulations are very conclusive.
\begin{figure}[!h]
\centering
\hbox{
\kern-33pt
\includegraphics[scale=1]{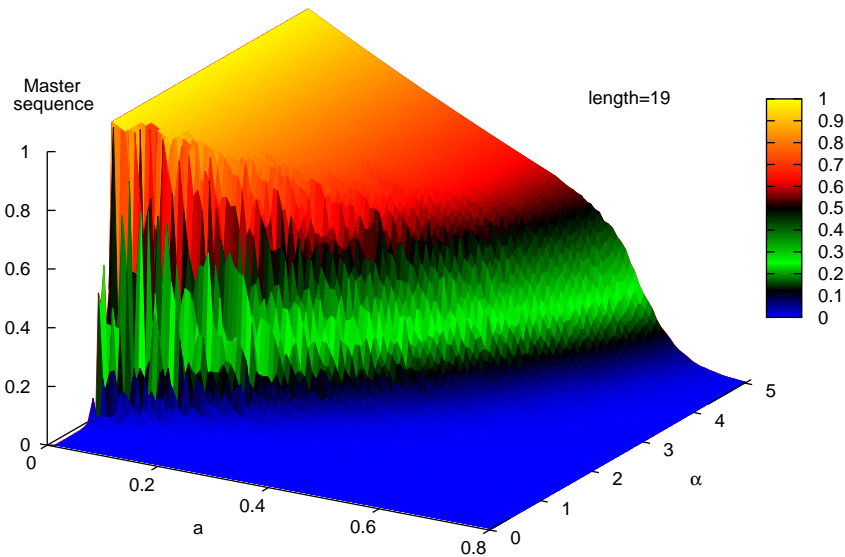} 
}
\vskip-20pt
\caption{Simulation of the equilibrium density of the
Master sequence}
\end{figure} 
\bigskip

\noindent
It is certainly well known that the population dynamics depends on the population
size (see the discussion of Wilke \cite{Wilke}). 
In a theoretical study \cite{NCH},
Van Nimwegen, Crutchfield and Huynen developed a model for the evolution of populations 
on neutral networks and they show that an important parameter
is the product of the population size and the mutation rate. 
The nature
of the dynamics changes radically depending on whether this product is 
small or large. 
Sumedha, Martin and Peliti \cite{SMP}
analyze further the influence of this parameter.
In \cite{NC}, 
Van Nimwegen and Crutchfield derived analytical expressions for the waiting times needed
to increase the fitness, starting from a local optimum. Their scaling relations 
involve the population size and show the existence of two different barriers,
a fitness barrier and an entropy barrier. Although they pursue a different goal 
than ours, most of the heuristic ingredients explained previously
are present in their work,
and much more; they observe and discuss also
the transition from the quasispecies regime for large populations
to the disordered regime for small populations.
The dependence on the population size and genome length has been
investigated numerically by
Elena, Wilke, Ofria and Lenski \cite{EWOL}.
Here we show rigorously the existence of a critical population size for the sharp
peak landscape in a specific asymptotic regime.
The existence of a critical population size for the emergence of a quasispecies is a
pleasing result: it shows that, even under the action of selection forces, a form
of cooperation is necessary to create a quasispecies. Moreover the critical
population size is much smaller than the cardinality of the possible genotypes.
In conclusion, even in the very simple framework of the Moran model on
the sharp peak landscape, 
cooperation is necessary to achieve the survival of the master sequence.
%This theme has been explored in more sophisticated models by Nowak \cite{}.

As emphasized by Eigen in \cite{EI2}, 
the error threshold phenomenon is similar to
a phase transition in statistical mechanics. 
Leuth{\"a}usser established a formal correspondence between
Eigen's model and an anisotropic Ising model \cite{LEU}.
Several researchers have employed tools from statistical mechanics
to analyze models of biological evolution, and more specifically
the error threshold: see the nice review written by Baake and Gabriel
\cite{BG}.
Baake investigated the so--called Onsager landscape in \cite{BAA1}. This way 
she could transfer to a biological model the famous computation
of Onsager for the two dimensional Ising model.
Saakian, Deem and Hu \cite{SAA1} 
compute the variance of the mean fitness in a finite population model
in order to control how it approximates the infinite population model.
Deem, Mu\~noz and Park \cite{PEM}
use a field theoretic representation in order to derive analytical
results.
%Derrida completely analyzed the evolution in a flat landscape \cite{}.

We were also very much inspired by ideas from statistical mechanics, but with
a different flavor. We do not use exact computations, rather we rely on
softer tools, namely coupling techniques and correlation inequalities. These
are the basic tools to prove the existence of a phase transition in
classical models, like the Ising model or percolation.
We seek large deviation estimates rather than precise scaling relations
in our asymptotic regime.
Of course the outcome of these techniques is very rough compared to exact
computations, yet they are much more robust and their range of applicability
is much wider. 
The model is presented in the next section and the main results in 
section~\ref{mainres}. The remaining sections are devoted to the proofs.
In the appendix we recall several classical results of the theory
of finite Markov chains.
\vfill\eject
\section{The model.}
\label{secmodel}
This section is devoted to the presentation of the model.
Let $\cal A$ be a finite alphabet and let
\index{$\cal A$}
\index{$\kappa$}
$\kappa=\card\cal A$ be its cardinality.
Let $\ell\geq 1
\index{$\ell$}$ be an integer. We consider the space
${\cal A}^\ell$ of sequences of length $\ell$ over the
alphabet $\cal A$.
Elements of this space represent the chromosome of an haploid
individual, or equivalently its genotype.
In our model, all the genes have the same set of alleles and each 
letter of the alphabet $\cal A$ is a possible allele.
Typical examples are
%choices for the set $\cal A$ are
${\cal A}=\{\,A,T,G,C\,\}$ to model standard DNA, or
${\cal A}=\{\,0,1\,\}$ to deal with binary sequences.
Generic elements of
${\cal A}^\ell$ will be denoted by the letters $u,v,w\index{$u,v,w$}$.
We shall study a simple model for the evolution of a finite population
of chromosomes on the space
${\cal A}^\ell$.
An essential feature of the model we consider is that the size of the
population is constant throughout the evolution. We denote by
$m$ the size of the population.
A population is an $m$--tuple of elements of
${\cal A}^\ell$.
Generic populations will be denoted by the letters 
$x,y,z\index{$x,y,z$}$.
Thus a population $x$ is a vector
$$x\,=\,
\left(
\begin{matrix}
x(1)\\
\vdots\\
x(m)
\end{matrix}
\right)
$$
whose components are chromosomes.
%, i.e., elements of ${\cal A}^\ell$.
For $i\in\{\,1,\dots,m\,\}$, we denote by
$$x(i,1),\dots,x(i,\ell)$$
the letters of the sequence $x(i)$. This way a population $x$
can be represented as an array
$$x\,=\,
\left(
\begin{matrix}
x(1,1)&\cdots&x(1,\ell)\\
\vdots& & \vdots\\
x(m,1)&\cdots&x(m,\ell)\\
\end{matrix}
\right)
$$
of size $m\times\ell$ of elements of $\cal A$, the
$i$--th line being the $i$--th chromosome.
The evolution of the population will be random and it will be driven
by two antagonistic forces: mutation and replication.
\medskip

\noindent
{\bf Mutation.} We assume that the mutation mechanism is the same
for all the loci, and that mutations occur independently.
Moreover we choose the most symmetric mutation scheme.
We denote by $q\in]0,1-1/\kappa[\index{$q$}$ 
the probability of the occurrence of a mutation
at one particular locus. If a mutation occurs, then the letter is replaced
randomly by another letter, chosen uniformly over the 
$\kappa-1$ remaining letters. We encode this  mechanism in a mutation matrix
$$M(u,v)\,,\quad u,v\in
{\cal A}^\ell
\index{$M(\cdot,\cdot)$}
$$
where $M(u,v)$ is the probability that the chromosome $u$ is transformed
by mutation into the chromosome $v$.
The analytical formula for 
$M(u,v)$ is then
$$
M(u,v)\,=\,
\prod_{j=1}^\ell
\left((1-q){1}_{u(j)=v(j)}
+\frac{q}{\kappa-1}
{1}_{u(j)\neq v(j)}
\right)\,.
$$
%\noindent
%{\bf Acknowledgements:} 
\medskip

\noindent
{\bf Replication.} The replication favors the development of fit chromosomes.
The fitness of a chromosome is encoded in a fitness function
$$A:{\cal A}^\ell\to[0,+\infty[\,.
\index{$A$}
$$
The fitness of a chromosome can be interpreted as its reproduction rate.
A chromosome $u$ gives birth at random times and the mean time interval
between two consecutive births is $1/A(u)$.
In the context of Eigen's model, the quantity $A(u)$ is the kinetic constant
associated to the chemical reaction for the replication of a 
macromolecule of
type $u$.
\medskip

\noindent
%{\bf Notation.}
{\bf Authorized changes.}
In our model, the only authorized changes in the population
consist
in replacing one chromosome of the population by a new one. The new
chromosome is obtained 
%either by performing mutations on the old one, or 
by replicating  another chromosome,
possibly with errors.
We introduce a specific notation corresponding to these changes.
%If
%$x\in
%\smash{\left({\cal A}^\ell\right)^m}$ is a population and
For a population
$x\in
\smash{\left({\cal A}^\ell\right)^m}$,
$j\in\{\,1,\dots,m\,\}$, $u\in {\cal A}^\ell$, we denote by
$x(j\leftarrow u)$ the population $x$ in which the $j$--th chromosome
$x(j)$ has been replaced by $u$:
$$
\index{$x(j\leftarrow u)$}
x(j\leftarrow u)
\,=\,
\left(
\begin{matrix}
x(1)\\
\vdots\\
x(j-1)\\
u\\
x(j+1)\\
\vdots\\
x(m)
\end{matrix}
\right)
$$
%whose components are chromosomes, i.e., elements of
%$\smash{{\cal A}^\ell}$.
We make this modeling choice in order to build a very simple model.
This type of model is in fact classical in population dynamics, 
they are called Moran models \cite{EW}. 
\medskip

\noindent
%{\bf Coupled versus decoupled mutations.}
{\bf The mutation--replication scheme.}
Several further choices have to be done to define the model precisely.
We have to decide how to combine the mutation and the replication processes.
There exist two main schemes in the literature. 
In the first scheme, mutations occur at any time of the life cycle
and they are caused by radiations or thermal fluctuations. This leads
to a decoupled Moran model.
In the second scheme, mutations
occur at the same time as births and they are caused by replication errors.
This is the case of the famous Eigen model and it leads 
to the Moran model we study here.
%\subsection*{Mathematical definition}
This Moran model can be described loosely as follows.
Births occur at random times. 
The rates of birth are given by the
fitness function~$A$.
There is at most one birth at each instant.
When an individual gives birth, it produces an offspring through a replication
process. Errors in the replication process induce mutations. The offspring
replaces an individual chosen randomly in the population
(with the uniform probability). 
\medskip

We build next a mathematical model for the evolution of a finite population
of size $m$ on the space
${\cal A}^\ell$, driven by mutation and replication as described above.
We will end up with a stochastic process on the population space
$\smash{\left({\cal A}^\ell\right)^m}$.
%$$\left({\cal A}^\ell\right)^m\,.$$
Since the genetic composition of a population contains all the necessary
information to describe its future evolution, our process will be Markovian.
\medskip

\noindent
{\bf Discrete versus continuous time.} 
We can either build a discrete time Markov chain 
or a continuous time Markov process.
%We can either build a discrete time Markov chain $(X_n)_{n\geq 0}$
%or a continuous time Markov process $(X_t)_{t\geq 0}$.
Although the mathematical construction of a discrete time Markov chain is
simpler, a continuous time process seems more adequate as a model of evolution
for a population: births, deaths and mutations can occur at any time. 
In addition, the continuous time model is mathematically more appealing.
%For these reasons, 
%we will state the main results for continuous time processes. 
We will build both types of models, in continuous and discrete time.
Continuous time models are conveniently defined by their infinitesimal
generators, while discrete time models are defined by their
transition matrices (see the appendix).
It should be noted, however, that the discrete time and the continuous
time processes are linked through a standard stochastization procedure 
and they have the same stationary distribution.
Therefore the asymptotic results we present here hold in both frameworks.
%\subsection*{Continuous time}
%\subsubsection*{Infinitesimal generator}
\medskip

\noindent
{\bf Infinitesimal generator.}
The continuous time Moran model
is the Mar\-kov process $(X_t)_{t\in{\mathbb R}^+}
\index{$X_t$}$
%whose infinitesimal generator $L$ 
%of the 
%continuous time Moran model
%{Stochastic Eigen model}
%is defined as follows:
having the following infinitesimal generator:
for $\phi$ a function from
$\smash{\left({\cal A}^\ell\right)^m}$
to $\mathbb R$ and
for any
$x\in
\smash{\left({\cal A}^\ell\right)^m}$, 
%$$L\phi(x)\,=\,
\begin{multline*}
\lim_{t\to 0}\,\frac{1}{t}\Big(E\big(\phi(X_t)|X_0=x\big) -\phi(x)\Big)\,=\,\cr
%\lim_{t\to 0}\,\frac{ E\big(\phi(X_t)|X_0=x\big) -\phi(x) }{t}
%\,=\,
\sum_{1\leq i,j\leq m}
\sum_{u\in {\cal A}^\ell}
%\frac{1}{m}
A(x(i))M(x(i),u)
\Big(\phi\big(x(j\leftarrow u)\big)-\phi(x)\Big)\,.\hfil
%$$
\end{multline*}
%\medskip
%
%\noindent
%{\bf Discrete time.}
%\subsubsection*{Loose description of the continuous time dynamics}
%\subsection*{Discrete time}
%\subsubsection*{Transition matrix}

\noindent
{\bf Transition matrix.}
The discrete time Moran model
is 
the Markov chain $(X_n)_{n\in\mathbb N}\index{$X_n$}$ whose
transition matrix is given by
%The transition matrix of the
%discrete time Moran model
%$(X_n)_{n\geq 0}$ 
%is given by
\begin{multline*}
\forall n\in{\mathbb N}\quad
\forall x\in
\smash{\left({\cal A}^\ell\right)^m} \quad
\forall j\in\{\,1,\dots,\ell\,\}\quad
\forall u \in\smash{{\cal A}^\ell} \setminus\{\,x(j)\,\}\cr
%\forall u\neq x(j)\cr
%\forall u \in\smash{{\cal A}^\ell} \setminus\{\,x(j)\,\}\cr
%p\big(x,x(j\leftarrow u)\big)\,=\,
P\big(X_{n+1}=x(j\leftarrow u)\,|\,X_n=x\big)\,=\,
\frac{1}{m^2\lambda}
\sum_{1\leq i\leq m}
{A(x(i))} M(x(i),u)\,,
\hfil
%p\big(x,x\big)\,=\,
%1-\frac{1}{m
%{\lambda} }
%\sum_{1\leq i\leq m}
%{A(x_i)} \,,
%\hfil
\end{multline*}
where $\lambda>0$ is a constant such that
$$\lambda\,\geq\,\max\,\big\{\,A(u):
{u \in\smash{{\cal A}^\ell}}
\,\big\}\,.
\index{$\lambda$}
$$
The other non diagonal coefficients of the transition matrix are zero.
The diagonal terms are chosen so that the sum of each line is equal to one.
%$$\forall x\in
%\smash{\left({\cal A}^\ell\right)^m} \qquad
%p\big(x,x\big)\,=\,
%1-
%\sum_{1\leq j\leq m}
%\sum_{
%\genfrac{}{}{0pt}{1}{u \in\smash{{\cal A}^\ell}}
%{u\neq x(j)}
%}
%p\big(x,x(j\leftarrow u)\big)\,.
%$$
%p\big(x,x\big)\,=\,
%1-\frac{1}{m
%{\lambda} }
%\sum_{1\leq i\leq m}
%{A(x_i)} \,,
%\hfil
%\subsubsection*{Loose description of the discrete time dynamics}
%\subsubsection*{Loose description of the dynamics}
\noindent
Notice that the continuous time formulation is more concise and
elegant: it does not require
the knowledge of the maximum of the fitness function $A$ in its definition.
\medskip

\noindent
{\bf 
Loose description of the dynamics.}
We explain first the discrete time dynamics of 
the Markov chain $(X_n)_{n\in\mathbb N}$.
Suppose that $X_n=x$ for some $n\in\mathbb N$ and let 
us describe loosely the transition mechanism to $X_{n+1}=y$.
An index $i$ in
$\{\,1,\dots,m\,\}$
is selected
randomly with the uniform probability.
%Select randomly with uniform probability an index $i$ in
%$\{\,1,\dots,m\,\}$.
With probability
$1-{A(x(i)) }/
\lambda$, nothing happens and $y=x$.
With probability
${A(x(i)) }/
{\lambda}$, the chromosome $x(i)$ enters the
replication process and it produces an offspring $u$ according
to the law $M(x(i),\cdot)$ given by the mutation matrix.
Another index $j$ is
selected randomly with uniform probability in
$\{\,1,\dots,m\,\}$. The population $y$ is obtained by replacing
the chromosome $x(j)$ in the population $x$ by a chromosome $u$.

We consider next the continuous time dynamics of 
the Markov process $(X_t)_{t\in{\mathbb R}^+}$.
The dynamics is governed by a clock that rings randomly.
The time interval $\tau$
between each of the clock ringing is exponentially distributed
with parameter
$m^2{\lambda}$:
$$\forall t\in{\mathbb R}^+\qquad P(\tau>t)\,=\,\exp\big(- 
m^2{\lambda} t\big)\,.$$
Suppose that the clock rings at time $t$ and that the process was in
state $x$ just before the time $t$.
The population $x$ is transformed into the population $y$ following
the same scheme as for the discrete time Markov chain 
$(X_n)_{n\in\mathbb N}$
described previously. At time $t$, the process jumps to the state $y$.
\vfill\eject
\section{Main results.}
\label{mainres}
This section is devoted to the presentation of the main results.
%We present the main results in this section. 
%The central result 
%of the whole study is theorem~\ref{mainth}.
%, which identifies a critical curve separating the 
\medskip

\noindent
{\bf Convention.}
%We denote by
%$(X_t)_{t\geq 0}$ 
%the Moran model 
%and by
%the Moran model. 
The results hold for both the discrete time and the
continuous time models, so we do not make separate statements. The time
variable is denoted by $t$ throughout this section, 
it is either discrete with values
in $\mathbb N$ or continuous with values in $\mathbb R^+$.
\medskip

\noindent
{\bf Sharp peak landscape.} 
We will consider only the sharp peak landscape defined as follows.
We fix a specific sequence, denoted by $w^*$, called the wild type or the
master sequence.
Let $\sigma>1\index{$\sigma$}$ be a fixed real number.
The fitness function $A$ is given by
$$\forall u\in
{\cal A}^\ell\qquad
A(u)\,=\,
\begin{cases}
1 &\text{if }u\neq w^*\\
\sigma &\text{if }u=w^*\\
\end{cases}
$$
\medskip

\noindent
{\bf Density of the master sequence.}
We denote by $N(x)\index{$N(x)$}$ 
the number of copies of the master sequence~$w^*$
present in the population $x$:
$$N(x)\,=\,\card\big\{\,i: 1\leq i\leq m, \, x(i)=w^*\,\big\}\,.$$
We are interested in the expected
density of the master sequence in the steady state
distribution of the process, that is,
$$\wild(\sigma,\ell,m,q)\,=\,
\lim_{t\to\infty} 
E\Big(
\frac{1}{m}
N(X_t)\Big)
\index{$\wild$}
\,,$$
%\,=\,
%\lim_{n\to\infty} 
%E\Big(
%\frac{1}{m}
%N(X_n)\Big)
as well as the variance
$$\vwild(\sigma,\ell,m,q)\,=\,
\lim_{t\to\infty} 
E\bigg(\Big(
\frac{1}{m}
N(X_t)
-\wild(\sigma,\ell,m,q)
\Big)^2
\bigg)
\index{$\vwild$}
\,.$$
The limits exist because the transition mechanism of the Markov process 
$(X_t)_{t\geq 0}$ is irreducible (and aperiodic for the discrete time case)
as soon as the mutation probability is strictly between $0$ and $1$.
Since the state space is finite, the Markov process 
$(X_t)_{t\geq 0}$ 
admits a unique
invariant probability measure, which describes the steady state of
the process. The ergodic theorem for Markov chains implies that the
law of  
$(X_t)_{t\geq 0}$ converges towards this invariant probability measure,
hence the above expectations converge. 
The limits depend on the parameters of the model, that is
$\sigma,\ell,m,q$.
Our choices for the infinitesimal
generator and the matrix transition imply that the discrete time version
and the continuous time version have exactly the same invariant probability
measure.
In order to exhibit a sharp transition phenomenon, we send
$\ell,m$ to $\infty$ and $q$ to $0$. 
%\medskip
%
%\noindent
%{\bf Asymptotics.}
%There exists a critical curve $\gamma$ in the plane ${\mathbb R}^+\times
%{\mathbb R}^+$
%separating two regions $Q$ and $U$
%such that
%%$(c,\alpha)$ of equation 
%%$\gamma(c,\alpha)=0$ such that:
Let $\phi:{\mathbb R}^+\to
{\mathbb R}^+\cup\{\,+\infty\,\}$ be the function defined by
%$\phi(a)=+\infty$
%if $a\geq\ln\sigma$ and
$$
\forall a<\ln\sigma\qquad
\phi(a)\,=\,
\frac
{ \displaystyle \sigma(1-e^{-a})
\ln\frac{\displaystyle\sigma(1-e^{-a})}{\displaystyle\sigma-1}
%\big(\ln({\sigma(1-e^{-a})})-\ln({\sigma-1})\big)
+\ln(\sigma e^{-a})}
{
\displaystyle (1-\sigma(1-e^{-a}))
}
\index{$\phi(a)$}
$$
%\begin{cases}
%\displaystyle\frac{(1-\sigma(1-e^{-a}))
%\,\ln\kappa
%}
%{
%\sigma(1-e^{-a})
%\ln\frac{\displaystyle\sigma(1-e^{-a})}{\displaystyle\sigma-1}
%%\big(\ln({\sigma(1-e^{-a})})-\ln({\sigma-1})\big)
%+\ln(\sigma e^{-a})}
%&\text{if $\sigma\exa>1$} \\
%+\infty&\text{if $\sigma\exa\leq 1$} \\
%\end{cases}
%%\,\ln\kappa
%%\,.
%if $a<\ln\sigma$ and
and $\phi(a)=0$
if $a\geq\ln\sigma$. 
\begin{figure}[ht]
\centering
\hbox{
\kern-20pt
\includegraphics[scale=0.97]{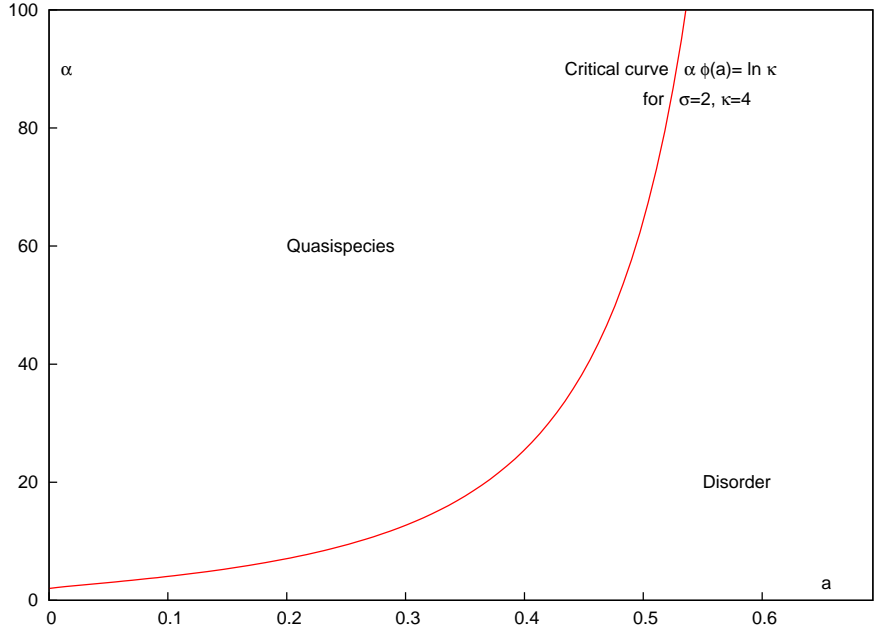} 
}
\vskip-10pt
\caption{Critical curve}
\label{crit_curve} 
\end{figure} 
\begin{theorem}\label{mainth}
We suppose that 
%$\ell$ goes to $\infty$, $m$ goes to $\infty$ and $q$ goes to $0$ 
$$\ell\to +\infty\,,\qquad m\to +\infty\,,\qquad q\to 0\,,$$
in such a way that
$${\ell q} \to a\in ]0,+\infty[\,,
\qquad\frac{m}{\ell}\to\alpha\in [0,+\infty]\,.
\index{$a,\alpha$}
$$
We have the following dichotomy:
%$$
%\lim_{\ell\to\infty}\wild\big(\sigma,\ell,\alpha \ell,a/\ell\big)\,=\,
%\begin{cases}
%0&\text{if $\alpha<\phi(a)$ } \\
%\frac{\displaystyle\sigma\exa-1}{\displaystyle\sigma-1}
%%&\text{if $\sigma\exa>1$} \\
%&\text{if $\alpha>\phi(a)$ } \\
%\end{cases}
%\,.
%$$
\medskip

\noindent
$\bullet\quad$ If $\alpha\,\phi(a)<\ln\kappa$ then
%$\quad
%\lim_{\ell\to\infty}\vwild\big(\sigma,\ell,\alpha \ell,a/\ell\big)\,=\,0$
%and
%\lim_{\ell\to\infty}\wild\big(\sigma,\ell,\alpha \ell,a/\ell\big)\,=\,0\,.$
%$$\lim_{\ell\to\infty}\wild\big(\sigma,\ell,\alpha \ell,a/\ell\big)\,=\,0
%\,,\qquad
$\wild\big(\sigma,\ell,m,q\big)\,\to\,0$.
%\lim_{\ell\to\infty}\vwild\big(\sigma,\ell,\alpha \ell,a/\ell\big)\,=\,0
%\,.$$
\smallskip

\noindent
$\bullet\quad$ If $\alpha\,\phi(a)>\ln\kappa$ then
%$\quad
%\lim_{\ell\to\infty}\vwild\big(\sigma,\ell,\alpha \ell,a/\ell\big)\,=\,0$
%\lim_{\ell\to\infty}\wild\big(\sigma,\ell,\alpha \ell,a/\ell\big)\,
%=
%\frac{\displaystyle\sigma\exa-1}{\displaystyle\sigma-1}
%\,.$
%and
$\wild\big(\sigma,\ell,m,q\big)\,\to\,
\frac{\displaystyle\sigma\exa-1}{\displaystyle\sigma-1}$.
%$$\lim_{\ell\to\infty}\wild\big(\sigma,\ell,\alpha \ell,a/\ell\big)\,=\,
%\frac{\displaystyle\sigma\exa-1}{\displaystyle\sigma-1}\,.$$
%\,,\qquad\cr
%\lim_{\ell\to\infty}\vwild\big(\sigma,\ell,\alpha \ell,a/\ell\big)\,&=\,0
%\,.
%\end{align*}
\smallskip

\noindent
In both cases, we have
$\vwild\big(\sigma,\ell,m,q\big)\,\to\,0$.
\end{theorem}
\begin{figure}[ht]
\centering
%\includegraphics[scale=0.89]{3d777_sem} 
%\vskip-20pt
\hbox{
\kern-38pt
\includegraphics[scale=1.11]{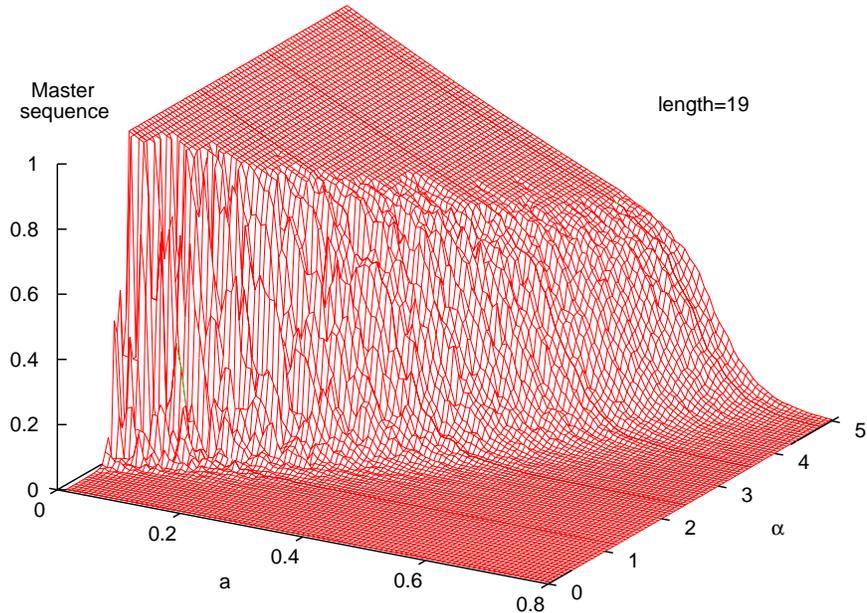} 
}
\vskip-20pt
\caption{Master sequence at equilibrium} 
\label{equimas}
\end{figure} 
\begin{figure} 
%\centering
\hbox{
\kern-20pt
\includegraphics[scale=1]{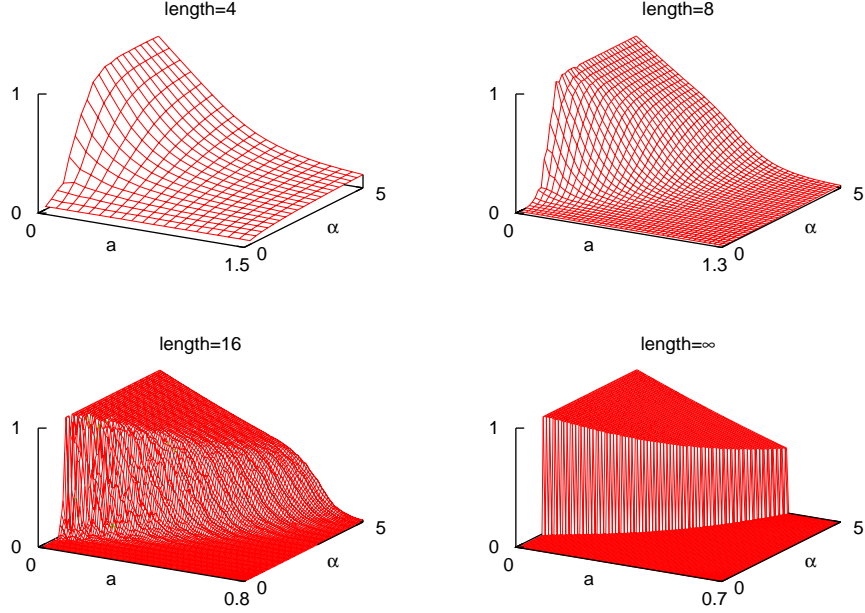} 
}
\vskip-20pt
\caption{Varying the length $\ell$}
\label{finite} 
%\vskip-20pt
\end{figure} 
These results are supported by computer simulations (see figure~\ref{equimas}). 
On the simulations, which are of course done for small values of $\ell$,
the transition associated to the critical population size seems even sharper than
the transition associated to the error threshold.
The programs are written in
$C$ with the help of the GNU scientific library and the graphical output is generated
with the help of the Gnuplot program.
To increase the efficiency of the simulations, we simulated the 
occupancy process
obtained by lumping
the original Moran model. The number of generations in a simulation
run was adjusted empirically in order to stabilize the output within a reasonable
amount of time.
%Typically, in a simulation of the model with parameters $\ell,m$, the number of
%generations is taken to be 
%$$100\,000\times 2^{\max(\ell, m)}$$
%multiplied by a factor between $1$ and $10$.
Twenty years ago, Nowak and Schuster could perform simulations with
$\ell=10$ and $m=100$ for $20\,000$ generations \cite{NS}. 
Today's computer powers allow to simulate easily models with $\ell=20$
and $m=100$ for $10\,000\,000\,000$ generations.
% within a fortnight.
The good news is that,
already for small values of $\ell$, 
the simulations are very conclusive.
Figure~\ref{finite} presents three pictures corresponding to simulations
with $\ell=4,8,16$, as well as the theoretical shape for $\ell=\infty$
in the last picture.
Notice that the statement of the theorem holds also in the case where $\alpha$ is null
or infinite. This yields the following results:
\smallskip

\noindent
{\bf Small populations.}
%$\bullet$ 
If $\ell,m\to +\infty,\,q\to 0,\,
{\ell q} \to a\in ]0,+\infty[,\,
\frac{\textstyle m}{\textstyle \ell}\to  0$, then
$\wild\big(\sigma,\ell,m, q)\to 0$.
%\smallskip
%
%\noindent
%$\bullet$ 
%If $\ell,m\to +\infty,\,q\to 0,\,
%{\ell q} \to \infty$, then
%$\wild\big(\sigma,\ell,m, q)\to 0$.
\smallskip

\noindent
{\bf Large populations.} Suppose that
$$\ell,m\to +\infty,\quad q\to 0,\quad
{\ell q} \to a\in ]0,+\infty[,\quad
\frac{\textstyle m}{\textstyle \ell}\to  +\infty\,.$$
If $a\geq \ln\sigma$, then
$\wild\big(\sigma,\ell,m, q)\to 0$.
If $a<\ln\sigma$, then
$$\wild\big(\sigma,\ell,m,q\big)\,\to\,
\frac{\displaystyle\sigma\exa-1}{\displaystyle\sigma-1}\,.$$
%\smallskip
%
%\noindent
%$\bullet$ 
%If $(\ell,m,q)$ is sent to $(\infty,\infty,0)$ in such a way that
%$m/\ell$ goes to $0$ or $\ell q$ goes to $\infty$, then
%$\wild\big(\sigma,\ell,m, q)$ goes to $0$
%and
%$\vwild\big(\sigma,\ell,m, q)$ goes to $0$.
%\smallskip
%
%\noindent
%$\bullet$ 
%If $(\ell,m,q)$ is sent to $(\infty,\infty,0)$ in such a way that
%$m/\ell$ goes to $\infty$ and $\ell q$ goes to $0$, then
%$\wild\big(\sigma,\ell,m, q)$ goes to $1$
%and
%$\vwild\big(\sigma,\ell,m, q)$ goes to $0$.
Interestingly, the large population regime
is reminiscent of Eigen's model.
A slightly more restrictive formulation consists in sending
$\ell$ to $\infty$, 
$m$ to $\infty$ 
and
$q$ to $0$ 
in such a way that $m/\ell$ and $\ell q$ are kept constant.
We might then take $q$ and $m$ as functions of $\ell$.
Let $a,\alpha\in ]0,+\infty[$. We take $q=a/\ell$ and $m=\alpha \ell$
%We introduce the variables
%$$a\,=\,\ell\, q,\quad\alpha\,=\,\frac{m}{\ell}\,.$$
and we have
$$\lim_{\ell\to\infty}\wild\big(\sigma,\ell,\alpha \ell,a/\ell\big)\,=\,
\begin{cases}
\quad\phantom{aaa} 0\phantom{\frac{1}{2}}&\text{if $\alpha\,\phi(a)<\ln\kappa$ } \\
\quad \displaystyle\frac{\displaystyle\sigma\exa-1}{\displaystyle\sigma-1}\quad
%&\text{if $\sigma\exa>1$} \\
&\text{if $\alpha\,\phi(a)>\ln\kappa$ } \\
\end{cases}
%\,.
$$
\noindent
Notice that 
$\alpha\,\phi(a)>\ln\kappa$ implies that $a<\ln\sigma$ and $\sigma\exa>1$.
%\begin{cases}
%\frac{\displaystyle\sigma\exa-1}{\displaystyle\sigma-1}
%&\text{if $\sigma\exa>1$} \\
%0&\text{if $\sigma\exa\leq 1$} \\
%\end{cases}
%%\begin{theorem}[Critical curve]\label{mainth}
%\begin{theorem}\label{mainthold}
%%There exists a critical curve $\gamma$ in the plane ${\mathbb R}^+\times
%%{\mathbb R}^+$
%%separating two regions $Q$ and $U$
%%such that
%%%$(c,\alpha)$ of equation 
%%%$\gamma(c,\alpha)=0$ such that:
%There exists a function $\phi:{\mathbb R}^+\to
%{\mathbb R}^+\cup\{\,+\infty\,\}$ such that
%\smallskip
%
%\noindent
%$\bullet\quad$ If $\alpha<\phi(a)$ then
%$\quad
%\lim_{\ell\to\infty}\wild\big(\sigma,\ell,\alpha \ell,a/\ell\big)\,=\,0\,.$
%\smallskip
%
%\noindent
%$\bullet\quad$ If $\alpha>\phi(a)$ then
%$\quad
%\lim_{\ell\to\infty}\wild\big(\sigma,\ell,\alpha \ell,a/\ell\big)\,>\,0\,.$
%\end{theorem}
The critical curve
$$\big\{\,(a,\alpha)\in
{\mathbb R}^+\times {\mathbb R}^+:\alpha\,\phi(a)=\ln\kappa\,\big\}$$
corresponds to parameters $(a,\alpha)$
which are exactly at the error threshold and the critical
population size.
We are able to compute explicitly the critical curve and the limiting
density 
because we consider a toy model.
We did not examine here what happens on the critical curve. It is expected
that the limiting density of the master sequence still fluctuates 
so that
$\vwild\big(\sigma,\ell,\alpha \ell,a/\ell\big)$ does not converge to $0$
whenever $\alpha\,\phi(a)=\ln\kappa$.
An important observation is that the critical scaling should be the same for
similar Moran models. In contrast, the critical curve seems to depend
strongly on the specific dynamics of the model.
However,
in the limit where $a$ goes to $0$, the function
$\phi(a)$ 
converges towards $\ln\sigma$. This yields
the minimal
population size allowing
the emergence
of a quasispecies.
\begin{figure}[ht]
\centering
\hbox{
\kern-20pt
\includegraphics[scale=0.97]{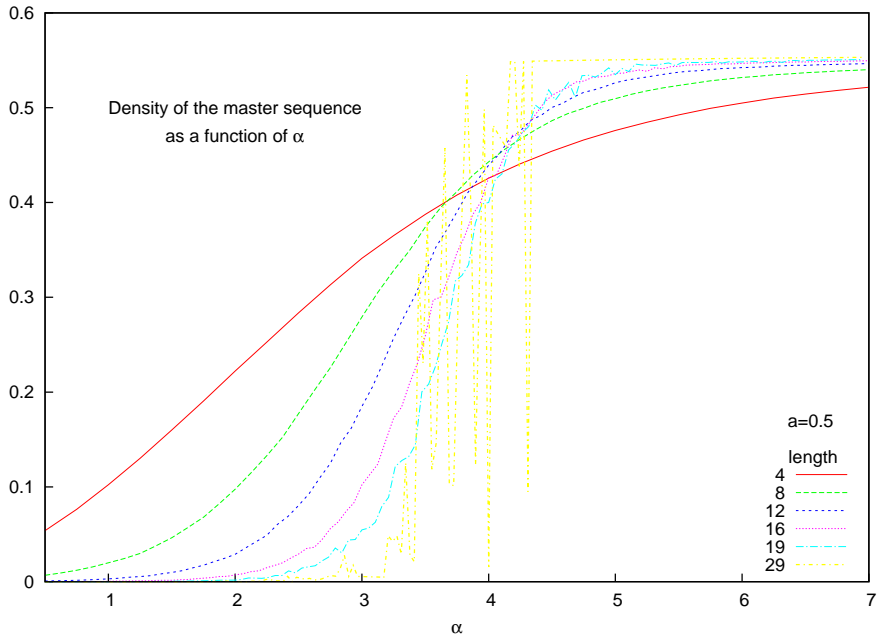} 
}
\caption{Critical population size}
\label{crit_size} 
%\kern-20pt
\end{figure} 

\begin{figure}[ht]
\centering
\hbox{
\kern-20pt
\includegraphics[scale=0.97]{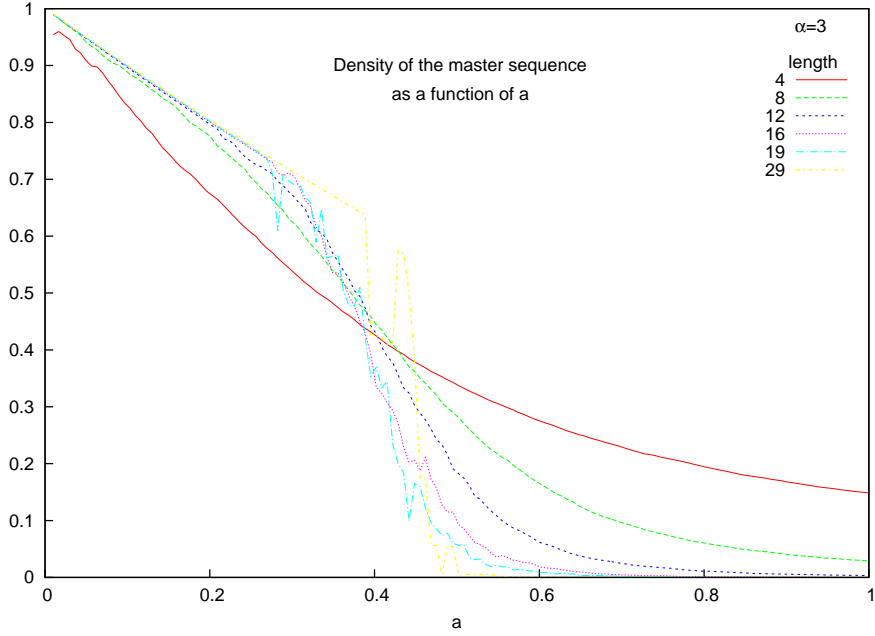} 
}
\caption{Error threshold}
\label{err_threshold}
%\kern-20pt
\end{figure} 
\begin{corollary}\label{minpop}
If $\alpha< \ln\kappa/\ln\sigma$ then
$$
\forall a>0\qquad 
\lim_{\ell\to\infty}\wild\big(\sigma,\ell,\alpha \ell,a/\ell\big)\,=\,0\,.$$
If $\alpha> \ln\kappa/\ln\sigma$ then
$$
\exists a>0\qquad 
\lim_{\ell\to\infty}\wild\big(\sigma,\ell,\alpha \ell,a/\ell\big)\,>\,0\,.$$
\end{corollary}
We can also compute the 
maximal mutation rate 
permitting the emergence
of a quasispecies. 
Interestingly, this maximal mutation rate is reminiscent of
the error catastrophe 
in Eigen's model.
\begin{corollary}
%[Eigen's error threshold]
\label{eigenet}
If $a> \ln\sigma$ then
$$
\forall \alpha>0\qquad 
\lim_{\ell\to\infty}\wild\big(\sigma,\ell,\alpha \ell,a/\ell\big)\,=\,0\,.$$
If $a< \ln\sigma$ then
$$
\exists \alpha>0\qquad 
\lim_{\ell\to\infty}\wild\big(\sigma,\ell,\alpha \ell,a/\ell\big)\,>\,0\,.$$
\end{corollary}
In conclusion, on the sharp peak landscape, a quasispecies can emerge only if
$$
m\,>\, \frac{\ln\kappa}{\ln\sigma}\,\ell\,,\qquad
q\,<\, \frac{\ln\sigma}{\ell} \,.
$$
The heuristic ideas behind theorem~\ref{mainth} were explained in the introduction.
These heuristics are quite simple, however,
the corresponding proofs are rather delicate and technical.
There is very little hope to do a proof entirely based on
exact computations. Our strategy
consists in comparing the original Moran process with simpler processes in order
to derive adequate lower and upper bounds. To this end, we couple the
various processes starting with different initial conditions 
(section~\ref{couplsec}). 
Unfortunately, the natural coupling for the Moran model
we wish to study is not monotone. Therefore we consider an almost equivalent
model, which we call the normalized Moran model.
This model is obtained by
normalizing the reproduction rates so that the total reproduction
rate of any population is one (section~\ref{secnorm}). 
%Unfortunately, the natural coupling for the Moran model
We first observe that the Moran model is
exchangeable (section~\ref{secexc}).
However, 
the initial state
space of the Moran process has no order structure and it is
huge. 
We use a classical technique, called lumping, in order to reduce the
state space (section~\ref{seclum}). 
This way we obtain two lumped processes: the distance
process 
$(D_t)_{t\geq 0}$ which records the Hamming distances between the 
chromosomes of the population and the Master sequence
and the occupancy process
$(O_t)_{t\geq 0}$ which records the distribution of these
Hamming distances.
The distance process is monotone in the neutral case $\sigma=1$,
while the occupancy process is monotone for any value $\sigma\geq 1$
(section~\ref{secmono}).
Therefore we construct lower and upper processes to bound the occupancy
process 
(section~\ref{secbounds}).
These processes have the same dynamics as the original process
in the neutral region and they evolve 
as a birth and death process as soon as the population contains a master
sequence. We use then the ergodic theorem for Markov chains and a renewal
argument to estimate the invariant probability measures of these processes.
The behavior of the lower and upper bounds depends mainly 
 on the persistence time and the discovery time of the master sequence.
We rely on the explicit formulas available for 
birth and death processes to estimate the persistence time 
(section~\ref{bide}).
To estimate the discovery time, we rely on rough estimates
for the mutation dynamics and correlation inequalities
(section~\ref{disc}).
The mutation dynamics is quite similar to the Ehrenfest urn, however 
it is more complicated 
because several mutations can occur
simultaneously and
exact formulas are not available.
The proof is concluded 
in section~\ref{secsyn}.

\noindent
{\bf Warning.} From section~\ref{secexc} onwards,
we work with the normalized Moran model defined
in section~\ref{secnorm}. This model is denoted by 
$(X_t)_{t\geq 0}$ and its transition matrix by $p$, like the initial Moran model.
We deal only with discrete time processes in the proofs. The time
is denoted by $t$ or $n$.
\vfill\eject
\section{Coupling}\label{couplsec}
The definition of the processes through infinitesimal generator
is not very intuitive at first sight. 
We will provide here a direct construction
of the processes, which does not make appeal to a general existence
result. This construction is standard and it is the formal counterpart
of the loose description of the dynamics given in section~\ref{secmodel}.
Moreover it provides a useful coupling of the processes with different
initial conditions and different control parameters $\sigma,q$.
All the processes will be built on a single large probability space.
We consider a probability space $(\Omega,{\mathcal F}, P)$
containing the following collection of independent random variables:
\medskip

\noindent
$\bullet$ a Poisson process 
$(\tau(t))_{t\geq 0}\index{$\tau(t)$}$
with intensity $m^2\lambda$.
\medskip

\noindent
$\bullet$ two sequences
of random variables 
$I_{n},J_n,\,n\geq 1
\index{$I_n$}
\index{$J_n$}
$,
with uniform law on 
the index set 
$\{\,1,\dots,\ell\,\}$.
\medskip

\noindent
$\bullet$ a family 
of random variables 
$U_{n,l},\,n\geq 1,\,1\leq l\leq\ell
\index{$U_{n,l}$}$,
with uniform law on 
the interval 
$[0,1]$.
\medskip

\noindent
$\bullet$ a sequence
of random variables 
$S_{n},\,n\geq 1 \index{$S_n$}$, 
with uniform law on the interval $[0,1]$.
\medskip

\noindent
%All these variables are independent.
We denote by 
%$(\tau_n, n\geq 1)$ 
$\tau_n
\index{$\tau_n$}$
the $n$--th arrival time of
the Poisson process
$(\tau(t))_{t\geq 0}$, i.e.,
$$\forall
n\geq 1\qquad
\tau_n \,=\, \inf\,\{\,t\geq 0:\tau(t)=n\,\}\,.$$
The random variables
$I_n$, $J_n$, $U_{n,l},\,\,1\leq l\leq\ell$,
and $S_{n}$ will be used to decide which move occurs at time $\tau_n$. 
To build the coupling, it is more convenient to replace the mutation
probability $q$ by the parameter $p$ given by
$$p\,=\,\frac{\kappa}{\kappa-1}\,q\,.
\index{$p$}
$$
We define a Markov chain 
$(X_n)_{n\in\mathbb N}$ with the help of the previous random ingredients,
whose law is the law of the Moran model.
The process starts at time $0$ from an arbitrary population $x_0$.
Let $n\geq 1$,
suppose that the process has been defined up to time $n-1$
%\in\mathbb N$ 
and 
that $X_{n-1}=x$.  
We explain how to build $X_{n}=y$.
Let us set $i=I_n$.
If $S_n>
{A(x(i)) }/
{\lambda}$, then $y=x$. 
Suppose next that $S_n\leq 
{A(x(i)) }/
{\lambda}$. We define $y$ as follows.
We index the elements of the alphabet $\mathcal A$ in an arbitrary way:
$${\mathcal A}\,=\,
\big\{\,a_1,\dots,a_\kappa\,\big\}\,.
\index{$a_r$}
$$
Let $j=J_n$. We set
$$\forall l\in\{\,1,\dots,\ell\,\}\qquad
y(j,l)\,=\,
\begin{cases}
\quad a_1&\text{if }\quad U_{n,l}<
\genfrac{}{}{}{0}{p}{\kappa} \\
\quad \vdots\\
\quad a_r&\text{if }\quad (r-1)\genfrac{}{}{}{0}{p}{\kappa}<U_{n,l}<
r\genfrac{}{}{}{0}{p}{\kappa} \\
\quad \vdots\\
\quad a_\kappa&\text{if }\quad (\kappa-1)\genfrac{}{}{}{0}{p}{\kappa}<U_{n,l}<
p\\
\quad x(i,l) &\text{if }\quad U_{n,l}\geq p\\
\end{cases}
$$
For $k\neq j$ we set $y(k)=x(k)$.
Finally we define $X_{n}=y$.

We define also a Markov process
$(X_t)_{t\in\mathbb R^+}$ with right continuous trajectories.
The process starts at time $0$ from an arbitrary population $x_0$ and it
moves only when there is an arrival in the Poisson process
$(\tau(t))_{t\geq 0}$.
Let $t>0$ and suppose that
$\tau_n=t$ for some $n\geq 1$.
Suppose that just before $t$ the process was in state $x$:
$$\lim_{\genfrac{}{}{0pt}{1}{s\to t}{s<t}} X_s\,=\,x\,.$$
We proceed as in the construction of the discrete time process at step $n$
to build the new population $y$ starting from $x$ and we set
$X_t=y$. Therefore we have
$$\forall n\geq 0\quad
\forall t\in[\tau_n,\tau_{n+1}[\qquad
X_t\,=\,X_n\,.$$
\vfill\eject
\section{Normalized model}\label{secnorm}
The Moran model defined previously
is difficult
to analyze for several reasons.
%, because it is not monotone. 
A major problem is that the natural coupling constructed in
section~\ref{couplsec} is not monotone.
We define next a related Moran
model which is simpler to study. This model is obtained by
normalizing the reproduction rates so that the total reproduction
rate of any population is one.
The continuous time normalized Moran model
is the Markov process $(X_t)_{t\in{\mathbb R}^+}$ 
whose infinitesimal generator $L$ 
%{Stochastic Eigen model}
is defined as follows:
for $\phi$ a function from
$\smash{\left({\cal A}^\ell\right)^m}$
to $\mathbb R$ and
for any
$x\in
\smash{\left({\cal A}^\ell\right)^m}$, 
$$L\phi(x)\,=\,
\sum_{1\leq i,j\leq m}
\sum_{u\in {\cal A}^\ell}
%\frac{1}{m}
\frac{A(x(i))M(x(i),u)}
{A(x(1))+\cdots+A(x(m))}
\Big(\phi\big(x(j\leftarrow u)\big)-\phi(x)\Big)\,.$$
%\medskip
%
%\noindent
%{\bf Discrete time.}
%\subsubsection*{Loose description of the continuous time dynamics}
%\subsection*{Discrete time}
%\subsubsection*{Transition matrix}
The discrete time normalized Moran model
is the 
Markov chain 
$(X_n)_{n\in\mathbb N}$
with transition matrix $p$ given by
\begin{multline*}
\forall x\in
\smash{\left({\cal A}^\ell\right)^m} \quad
\forall j\in\{\,1,\dots,\ell\,\}\quad
%\forall u \in\smash{{\cal A}^\ell} \cr
\forall u \in\smash{{\cal A}^\ell} \setminus\{\,x(j)\,\}\cr
%\forall u\neq x(j)\cr
%\forall u \in\smash{{\cal A}^\ell} \setminus\{\,x(j)\,\}\cr
p\big(x,x(j\leftarrow u)\big)\,=\,
\frac{1}{m}
\sum_{1\leq i\leq m}
\frac{A(x(i))
 M(x(i),u)
}
{A(x(1))+\cdots+A(x(m))}
 \,.
%p\big(x,x\big)\,=\,
%1-\frac{1}{m
%{\lambda} }
%\sum_{1\leq i\leq m}
%{A(x_i)} \,,
%\hfil
\end{multline*}
The other non diagonal coefficients of the transition matrix are zero.
%Until the final section~\ref{secfinal}, 
In the remaining of the paper,
we shall work with this Markov chain
$(X_n)_{n\in\mathbb N}$
and the transition matrix $p$. We shall prove 
the main theorem~\ref{mainth}
%~\ref{sub},\ref{super},\ref{mainth}
of section~\ref{mainres}
for this process.
In fact, we shall even prove the following stronger result.
%Let $\mu$ be the invariant probability 
%measure of the Markov chain
%$(X_n)_{n\geq 0}$ and
Let $\nu$ be the image of 
the invariant probability 
measure of %the Markov chain
$(X_n)_{n\geq 0}$ 
%and $\mu$
through the map
$$x\in\Alm\mapsto
\frac{1}{m}
N(x)\in[0,1]\,.$$
The probability measure $\nu\index{$\nu$}$ is a measure on the interval $[0,1]$
describing the equilibrium density of the master sequence in the population.
Indeed,
$$\forall i\in\zm\qquad \nu\Big(\frac{i}{m}\Big)\,=\,
\lim_{n\to\infty} \,
P\big(N(X_n)=i\big)\,.$$
The probability $\nu$ depends on the parameters 
$\sigma,\ell,m,q$ of the model.
%We set
%$$a\,=\,\ell\, q,\quad\alpha\,=\,\frac{m}{\ell}\,.$$
%Let $\phi:{\mathbb R}^+\to
%{\mathbb R}^+\cup\{\,+\infty\,\}$ be the function defined by
%$\phi(a)=+\infty$
%if $a\geq\ln\sigma$ and
Let $\phi(a)$ be the function defined before theorem~\ref{mainth}, i.e.,
$$
\forall a<\ln\sigma\qquad
\phi(a)\,=\,
\frac
{ \displaystyle \sigma(1-e^{-a})
\ln\frac{\displaystyle\sigma(1-e^{-a})}{\displaystyle\sigma-1}
%\big(\ln({\sigma(1-e^{-a})})-\ln({\sigma-1})\big)
+\ln(\sigma e^{-a})}
{
\displaystyle (1-\sigma(1-e^{-a}))
}
\index{$\phi(a)$}
$$
and $\phi(a)=0$
if $a\geq\ln\sigma$. 
Let
$$\rho^*\,=\,
\frac{\displaystyle\sigma\exa-1}{\displaystyle\sigma-1}\,.
\index{$\rho^*$}
$$
%$\bullet$ 
\begin{theorem}\label{stronger}
We suppose that 
%$\ell$ goes to $\infty$, $m$ goes to $\infty$ and $q$ goes to $0$ 
$$\ell\to +\infty\,,\qquad m\to +\infty\,,\qquad q\to 0\,,$$
in such a way that
$${\ell q} \to a\in ]0,+\infty[\,,
\qquad\frac{m}{\ell}\to\alpha\in [0,+\infty]\,.$$
We have the following dichotomy:
\medskip

\noindent
$\bullet\quad$ 
If $\alpha\,\phi(a)<\ln\kappa$ then
$\nu$ converges towards the Dirac mass at $0$:
% as $\ell$ goes to $\infty$:
$$\forall\ve>0\qquad
\nu([0,\ve])\,\to\,1\,.$$
%\lim_{\ell\to\infty}\nu([0,\ve])\,=\,1\,.$$
%\smallskip
%
%\medskip
%
\noindent
$\bullet\quad$ 
If $\alpha\,\phi(a)>\ln\kappa$ then
$\nu$ converges towards the Dirac mass at $\rho^*$:
$$\forall\ve>0\qquad
%\lim_{\ell\to\infty}\nu([\rho^*-\ve,\rho^*+\ve])\,=\,1\,.$$
\nu([\rho^*-\ve,\rho^*+\ve])\,\to\,1\,.$$
%=\nu\big(\Big[
%\frac{\displaystyle\sigma\exa-1}{\displaystyle\sigma-1}
%-\ve,
%\frac{\displaystyle\sigma\exa-1}{\displaystyle\sigma-1}
%+\ve
%\Big]\Big)\,=\,1
\end{theorem}
We shall prove this theorem for the normalized Moran model
$(X_n)_{n\in\mathbb N}$. 
%Let us show how to transfer it to
%the initial model.
Let us show how this implies theorem~\ref{mainth} for
the initial model.
In the remainder of this argument, we denote by
$(X'_n)_{n\in\mathbb N}$ the Moran model described in section~\ref{secmodel}
and by $p'$ its transition matrix.
%To alleviate the notation, we suppress the upper bar in
%$\bX_n$ and $\bp$.
%The diagonal terms are chosen so that the sum of each line is equal to one:
%$$\forall x\in
%\smash{\left({\cal A}^\ell\right)^m} \qquad
%p\big(x,x\big)\,=\,
%1-
%\sum_{1\leq j\leq m}
%\sum_{
%\genfrac{}{}{0pt}{1}{u \in\smash{{\cal A}^\ell}}
%{u\neq x(j)}
%}
%p\big(x,x(j\leftarrow u)\big)\,.
%$$
%p\big(x,x\big)\,=\,
%1-\frac{1}{m
%{\lambda} }
%\sum_{1\leq i\leq m}
The transition matrices $p$ and $p'$ are related by the simple relation
$$\forall x,y\in\Alm\,,\quad
x\neq y\,,\qquad
p'(x,y)\,=\,\beta(x)\,
p(x,y)$$
where 
$$\forall x\in\Alm\qquad
\beta(x)\,=\,
\frac{1}{m\lambda}\big({A(x(1))+\cdots+A(x(m))\big)}\,.
\index{$\beta(x)$}
$$
%Let us denote by $\mu$ (respectively $\mu'$) the invariant probability 
%measure of the process
%$(X_t)_{t\geq 0}$ (respectively 
%$(X'_t)_{t\geq 0}$). 
%Let us denote by $\mu'$ the invariant probability 
Let $\mu$ and $\mu'$ be the invariant probability 
measures of the processes
$(X_t)_{t\geq 0}$
and
$(X'_t)_{t\geq 0}$.
The probability $\mu\index{$\mu$}$ 
is the unique solution of the system of equations
$$\forall x\in\Alm\qquad
\mu(x)\,=\,
\sum_{y\in\alm}{\mu(y)}\,p(y,x)\,.
$$
We rewrite these equations as:
$$\forall x\in\Alm\qquad
\mu(x)
\sum_{
\genfrac{}{}{0pt}{1}{y\in\alm}
{y\neq x}}
\,p(x,y)
\,=\,
\sum_{
\genfrac{}{}{0pt}{1}{y\in\alm}
{y\neq x}}
{\mu(y)}\,p(y,x)\,.
$$
Replacing $p$ by $p'$, we get
$$\forall x\in\Alm\qquad
\frac{\mu(x)}{\beta(x)}
%{\beta(x)} {\mu'(x)}
\sum_{
\genfrac{}{}{0pt}{1}{y\in\alm}
{y\neq x}}
\,p'(x,y)
\,=\,
\sum_{
\genfrac{}{}{0pt}{1}{y\in\alm}
{y\neq x}}
\frac{\mu(y)}{\beta(y)}
%{\beta(y)}
%{\mu(y)}
\,p'(y,x)\,.
$$
Using the uniqueness of the invariant probability
measure associated to $p'$, we conclude that
%The relation on the transition matrices yields that
$$\forall x\in\Alm\qquad
\mu'(x)\,=\,
\frac{\displaystyle
\frac{\mu(x)}{\beta(x)}}
{\displaystyle
\sum_{y\in\alm}\frac{\mu(y)}{\beta(y)}}\,.
$$
In the case of the
sharp peak landscape, the function $\beta(x)$ can be rewritten as
$$\forall x\in\Alm\qquad
\beta(x)\,=\,
\frac{1}{m\lambda}\big((\sigma-1)N(x)+m\big)\,.$$
Let us denote by $\nu$ and $\nu'$ the images of $\mu$ and
$\mu'$ through the map
$$x\in\Alm\mapsto
\frac{1}{m}
N(x)\in[0,1]\,.$$
%With this notation, 
We can thus rewrite
$$
\sum_{y\in\alm}\frac{\mu(y)}{\beta(y)}
\,=\,
\sum_{y\in\alm}
\frac{\displaystyle\lambda
{\mu(y)}
}{\displaystyle(\sigma-1)\frac{N(y)}{m}+1}
\,=\,
\int_{[0,1]}
\frac{\displaystyle\lambda
\,d\nu(t)
}{\displaystyle(\sigma-1)t+1}
\,.
$$
For any function $f:[0,1]\to\R$,
we have then
\begin{multline*}
%\wild'(\sigma,\ell,m,p)\,=\,
\int_{[0,1]} f\,d\nu'\,=\,
\lim_{t\to\infty} 
E\bigg(f\Big(
\frac{1}{m}
N(X'_t)\Big)\bigg)
\,=\,
\sum_{x\in\alm}
%\int_{\textstyle\Alm}
f\Big(\frac{1}{m}
N(x)\Big)\,\mu'(x)\cr
\,=\,
\frac{\displaystyle
\sum_{x\in\alm}
%\int_{\textstyle\Alm}
f\Big(
\frac{
N(x)
}{m}\Big)
\frac{\mu(x)}{\beta(x)}}
{\displaystyle
\sum_{y\in\alm}\frac{\mu(y)}{\beta(y)}}
\,=\,
\frac{\displaystyle\int_{[0,1]}
\frac{\displaystyle\lambda f(t)
\,d\nu(t)
}{\displaystyle(\sigma-1)t+1}}{
\displaystyle
\int_{[0,1]}
\frac{\displaystyle\lambda
\,d\nu(t)
}{\displaystyle(\sigma-1)t+1}}\,.
%\int_{[0,1]}
%x\,d\nu'(x)\cr
\end{multline*}
We suppose that 
%$\ell$ goes to $\infty$, $m$ goes to $\infty$ and $q$ goes to $0$ 
$$\ell\to +\infty\,,\qquad m\to +\infty\,,\qquad q\to 0\,,$$
in such a way that
$${\ell q} \to a\in ]0,+\infty[\,,
\qquad\frac{m}{\ell}\to\alpha\in [0,+\infty]\,.$$
By theorem~\ref{stronger}, away from the critical curve $\alpha\,\phi(a)=\ln\kappa$,
%as $\ell$ goes to $\infty$, 
the probability $\nu$ converges towards a Dirac mass. If $\nu$ converges towards
a Dirac mass at $\rho$, then we conclude from the above formula that
$\nu'$ converges towards the same Dirac mass and
\begin{align*}
%\lim_{\ell\to\infty}
%\wild\big(\sigma,\ell,\alpha \ell,a/\ell\big)\,&=\,\rho
\wild\big(\sigma,\ell,m,q\big)\,&\to\,\rho
\,,\qquad\cr
%\lim_{\ell\to\infty}\vwild\big(\sigma,\ell,\alpha \ell,a/\ell\big)\,&=\,0
\vwild\big(\sigma,\ell,m,q\big)\,&\to\,0
\,.
\end{align*}
%In both cases, we have
%The same holds whenever $\nu$ converges towards the Dirac mass at~$0$.
This way we obtain the statements of theorem~\ref{mainth}. 
From now onwards, in the proofs,
we work exclusively with the normalized Moran process,
and we denote it by~$(X_t)_{t\geq 0}$.
\vfill\eject
\section{Exchangeability}\label{secexc}
\def\Sm{\mathfrak{S}_m} 
The symmetric group $\Sm$ 
of the permutations of
$\um$ acts in a natural way on the populations through the following
group operation:
$$\forall x\in\Alm\quad\forall \rho\in\Sm\quad
\forall j\in\um\qquad
(\rho\cdot x)(j)\,=\,x(\rho(j))\,.$$
A probability measure $\mu$ on $\Alm$ is exchangeable if it is 
invariant under the action of $\Sm$:
$$\forall \rho\in\Sm\quad \forall x\in\Alm\quad
\mu(\rho\cdot x)\,=\,\mu(x)\,.$$
A process
$(X_t)_{t\geq 0}$ with values in $\Alm$ is exchangeable if and only if,
for any $t \geq 0$, the law of $X_t$ is exchangeable.
\begin{lemma}\label{matinv}
The transition matrix $p$ 
%of the model
is invariant under the action of $\Sm$: 
\begin{multline*}
\forall x\in\Alm\quad\forall \rho\in\Sm\quad
\forall j\in\um\quad
\forall u \in\smash{{\cal A}^\ell} \setminus\{\,x(j)\,\}\cr
p\big(\rho\cdot x,
\rho\cdot( x(j\leftarrow u))\big)
%\rho\cdot x(\rho^{-1}(j)\leftarrow u)\big)
 \,=\,
p\big(x,
x(j\leftarrow u)\big)
\,.
\end{multline*}
\end{lemma}
\begin{proof}
Let $x,\rho,j,u$ be as in the statement of the lemma.
We have
\begin{multline*}
p\big(\rho\cdot x,
\rho\cdot( x(j\leftarrow u))\big)
 \,=\,
p\big(\rho\cdot x,
(\rho\cdot x)(\rho^{-1}(j)\leftarrow u)\big)
\cr
 \,=\,
\frac{1}{m}
\sum_{1\leq i\leq m}
\frac
{A((\rho\cdot x)(i)) 
M((\rho\cdot x)(i),u)}
{A((\rho\cdot x)(1))+\cdots+
A((\rho\cdot x)(m))} 
\cr
 %\,=\,
%\frac{1}{m}
%\sum_{1\leq i\leq m}
%\frac
%{A(x(\rho(i))) 
%M(x(\rho(i)),u)}
%{A(x(\rho(1)))+\cdots+
%A(x(\rho(m)))} 
%\cr
 \,=\,
\frac{1}{m}
\sum_{1\leq i\leq m}
\frac{A(x(i))
 M(x(i),u)
}
{A(x(1))+\cdots+A(x(m))}
 \,=\,
p\big(x,
x(j\leftarrow u)\big)\,.
\end{multline*}
Thus the matrix $p$ satisfies the required invariance property.
\end{proof}
\begin{corollary}\label{exchang}
Let $\mu$ be an exchangeable probability distribution on the population space
$\Alm$. The Moran model
$(X_t)_{t\geq 0}$ 
starting 
with $\mu$ as the initial distribution is exchangeable.
\end{corollary}
\begin{proof}
Let $\rho\in\Sm$ and let $f$ be a function from $\Alm$ to $\R$. Using
the exchangeability of $\mu$ and lemma~\ref{matinv}, we have,
for any $t\geq 1$,
\begin{multline*}
E\big(f(\rho\cdot X_t)\big)\,=\,
\sum_{x_0,\cdots,x_t\in\alm}
\mu(x_0)\,p(x_0,x_1)\cdots
p(x_{t-1},x_t)\,
f(\rho\cdot X_t)\cr
\sum_{x_0,\cdots,x_t\in\alm}
\mu(\rho\cdot x_0)\,p(\rho\cdot x_0,\rho\cdot x_1)\cdots
p(\rho\cdot x_{t-1},\rho\cdot x_t)\,
f(\rho\cdot X_t)\cr
\,=\,
\sum_{x_0,\cdots,x_t\in\alm}
\mu(x_0)\,p(x_0,x_1)\cdots
p(x_{t-1},x_t)\,
f(X_t)
\,=\,
E\big(f(X_t)\big)\,.
\end{multline*}
Thus the process
$(X_t)_{t\geq 0}$ is exchangeable.
\end{proof}
\vfill\eject
\section{Lumping}\label{seclum}
The state space of the process 
$(X_t)_{t\geq 0}$ 
is huge, it has cardinality
$\kappa^{\ell m}$. We will rely on
a classical technique to reduce the state space called lumping
(see the appendix). 
We consider here only the sharp peak landscape. In this situation,
the fitness of a chromosome is a function of its distance to the master sequence.
A close look at the mutation mechanism reveals that chromosomes which are
at the same distance from the Master sequence are equivalent for the
dynamics, hence they can be lumped together in order to build a 
simpler process on a reduced space.
For simplicity, we consider only the discrete time process. However
similar results hold in continuous time.
\subsection{Distance process}
We denote by $d_H$ the Hamming distance between two chromosomes:
$$\forall u,v\in
{\mathcal A}^\ell\qquad
d_H(u,v)\,=\,\card\,\big\{\,j:
1\leq j\leq\ell,\,u(j)\neq v(j)\,\big\}\,.
\index{$d_H$}
$$
%We recall that $w^*$ is the master sequence. 
We will keep track of
the distances of the chromosomes to 
the master sequence 
$w^*$.
%We define two functions
We define a function
$
{H}:\smash{{\cal A}^\ell}\to
\{\,0,\dots,\ell\,\}
\index{$H$}$
by setting
$$\forall u\in
{\cal A}^\ell\qquad
H(u)\,=\,
d_H\big(u,w^*\big)\,.$$
The map $H$ induces a partition of
${\cal A}^\ell$ into Hamming classes
$$H^{-1}(\{\,b\,\})\,,\qquad
b\in
\{\,0,\dots,\ell\,\}\,.$$
We prove first that the mutation matrix is lumpable with respect 
to the function $H$.
\begin{lemma}[Lumped mutation matrix]\label{mhlump}
%For any $u,v\in\Al$ such that $H(u)=H(v)$, we have
%$$\forall c\in\zl\qquad
%\sum_{
%\genfrac{}{}{0pt}{1}{w\in {{\cal A}^\ell}}{H(w)=c}
%}M(u,w)\,=\,
%\sum_{
%\genfrac{}{}{0pt}{1}{w\in {{\cal A}^\ell}}{H(w)=c}
%}M(v,w)\,.$$
%For any $u\in\Al$ and
%$c\in\zl$ the sum
%\par\noindent
Let $b,c\in\zl$ and let $u\in\Al$ such that $H(u)=b$.
The sum
$$\sum_{
\genfrac{}{}{0pt}{1}{w\in {{\cal A}^\ell}}{H(w)=c}
}M(u,w)$$
does not depend on $u$ in $H^{-1}(\{\,b\,\})$, 
it is a function of $b$ and $c$ only, which we denote by 
$M_H(b,c) \index{$M_H$}$.
%For
%$b,c\in
%\{\,0,\dots,\ell\,\}$, 
The coefficient
$M_H(b,c)$ is equal to
$$
\sum_{
\genfrac{}{}{0pt}{1}{0\leq k\leq\ell-b}{
\genfrac{}{}{0pt}{1}
 {0\leq l\leq b}{k-l=c-b}
}
}
{ \binom{\ell-b}{k}}
{\binom{b}{l}}
\Big(p\Big(1-\frac{1}{\kappa}\Big)\Big)^k
\Big(1-p\Big(1-\frac{1}{\kappa}\Big)\Big)^{\ell-b-k}
\Big(\frac{p}{\kappa}\Big)^l
\Big(1-\frac{p}{\kappa}\Big)^{b-l}\,.
$$
\end{lemma}
\begin{proof}
Let $b,c\in\zl$ and let $u\in\Al$ such that $H(u)=b$.
We will compute the law of $H(w)$ whenever $w$ follows
the law $M(u,\cdot)$ given by the line of $M$ associated to~$u$.
For any $w\in\Al$, we have
\begin{multline*}
H(w)\,=\,\sum_{1\leq l\leq\ell}1_{w(l)\neq w^*(l)}\cr
\,=\,\sum_{1\leq l\leq\ell}
\Big(1_{w(l)\neq w^*(l), u(l)= w^*(l)}
+
1_{w(l)\neq w^*(l), u(l)\neq w^*(l)}\Big)
\cr
\,=\,H(u)+
\sum_{1\leq l\leq\ell}
\Big(1_{w(l)\neq w^*(l), u(l)= w^*(l)}
-1_{w(l)= w^*(l), u(l)\neq w^*(l)}\Big)\,.
\end{multline*}
According to the mutation kernel $M$, for indices $l$ such that
$u(l)=w^*(l)$, the variable
$\smash{1_{w(l)\neq w^*(l)}}$
is Bernoulli with parameter $p(1-1/\kappa)$, while for indices
$l$ such that
$u(l)\neq w^*(l)$, the variable
$1_{w(l)= w^*(l)}$
is Bernoulli with parameter $p/\kappa$. Moreover these Bernoulli variables
are independent. Thus the law 
of $H(w)$ under the kernel 
$M(u,w)$ is given by
$$H(u)+\text{Binomial}\big(\ell-H(u),p(1-1/\kappa)\big)
-\text{Binomial}\big(H(u),p/\kappa\big)$$
where the two binomial random variables are independent.
This law depends only on $H(u)$, therefore the sum
$$\sum_{
\genfrac{}{}{0pt}{1}{w\in {{\cal A}^\ell}}{H(w)=c}
}M(u,w)$$
is a function of $b=H(u)$ and $c=H(w)$ only,
which we denote by $M_H(b,c)$.
The 
formula for
the lumped matrix $M_H$ is obtained by computing the law of
the difference of the two independent binomial laws appearing above.
%Since $H(u)=H(v)$, 
%the sums are the same for $u$ and $v$.
%We denote it by $M_H(b,c)$.
%and this implies the claim of the lemma.
\end{proof}

\noindent
The fitness function $A$ of the
sharp peak landscape can be factorized through $H$. If we define
$$\forall 
b\in
\{\,0,\dots,\ell\,\}
\qquad
A_H(b)\,=\,
\begin{cases}
\sigma &\text{if }b=0\\
1 &\text{if }b\geq 1\\
\end{cases}
\index{$A_H$}
$$
then we have
$$\forall u\in
{\cal A}^\ell\qquad
A(u)\,=\,A_H(H(u))\,.$$
We define further a 
vector function
${\mathbb H}:{\left({\cal A}^\ell\right)^m}\to
\{\,0,\dots,\ell\,\}^m
\index{$\mathbb H$}
$
by setting
$$
\forall x\,=\,
\left(
\begin{matrix}
x(1)\\
\vdots\\
x(m)
\end{matrix}
\right)
\in
\smash{\left({\cal A}^\ell\right)^m}
\qquad
{\mathbb H}(x)\,=\,
\left(
\begin{matrix}
H\big(x(1)\big)\\
\vdots\\
H\big(x(m)\big)
\end{matrix}
\right)\,.
$$
The partition of
$\smash{\left({\cal A}^\ell\right)^m}$ 
induced by the map $\mathbb H$ 
is
$${\mathbb H}^{-1}(\{\,d\,\})\,,\qquad
d\in
\{\,0,\dots,\ell\,\}^m\,.$$
We define finally the distance process
$(D_t)_{t\geq 0}
\index{$D_t$}
$
by 
$$\forall t\geq 0\qquad D_t\,=\,{\mathbb H}\big(X_t\big)\,.$$
Our next goal is to prove 
that the process 
$(X_t)_{t\geq 0}$ is lumpable with respect to the partition of
$\smash{\left({\cal A}^\ell\right)^m}$ induced by the map $\mathbb H$,
so that the distance process 
$(D_t)_{t\geq 0}$
is a genuine Markov process.
%
%\noindent
%From the previous argument,
%it is possible to provide an explicit formula for
%the lumped matrix $M_H$.
\begin{proposition}[${\mathbb H}$ Lumpability]\label{hlump}
%For any $e\in
%\{\,0,\dots,\ell\,\}^m$, any $x,y\in
%{\left({\cal A}^\ell\right)^m}$, 
Let $p$ be the transition matrix of the Moran model.
%model
We have
\begin{multline*}
\forall e\in
\{\,0,\dots,\ell\,\}^m \quad
\forall 
x,y\in
{\left({\cal A}^\ell\right)^m}\,,\cr
{\mathbb H}(x)=
{\mathbb H}(y)\quad\Longrightarrow\quad
\sum_{
\genfrac{}{}{0pt}{1}{z\in {\left({\cal A}^\ell\right)^m}}
{{\mathbb H}(z)=e}
}
p(x,z)\,=\,
\sum_{
\genfrac{}{}{0pt}{1}{z\in {\left({\cal A}^\ell\right)^m}}
{{\mathbb H}(z)=e}
}
p(y,z)\,.
\end{multline*}
\end{proposition}
\begin{proof}
For the process
$(X_t)_{t\geq 0}$, the only transitions having positive probability
are the transitions of the form
$$x\quad\longrightarrow\quad x(j\leftarrow u)\,,\qquad
1\leq j\leq m,\quad u\in{\mathcal A}^\ell\,.$$
Let $e\in \{\,0,\dots,\ell\,\}^m$
and let
$x,y\in
{\left({\cal A}^\ell\right)^m}$ be such that
${\mathbb H}(x)=
{\mathbb H}(y)$. We set
$d={\mathbb H}(x)=
{\mathbb H}(y)$.
If the vectors $d,e$ differ for more than two components, then the sums 
appearing in the statement of the proposition are equal to zero.
Suppose first
that the vectors $d,e$ differ in exactly one component, 
so that there exist
$j\in\um$ and $c\in\zl$ such that
$e=d(j\leftarrow c)$ and $d(j)\neq c$. Naturally,
$d(j\leftarrow c)$ is the vector $d$ in which the $j$--th component
$d(j)$ has been replaced by $c$:
$$
\index{$d(j\leftarrow c)$}
d(j\leftarrow c)
\,=\,
\left(
\begin{matrix}
d(1)\\
\vdots\\
d(j-1)\\
c\\
d(j+1)\\
\vdots\\
d(m)
\end{matrix}
\right)
$$
We have then
$$
\sum_{
\genfrac{}{}{0pt}{1}{z\in {\left({\cal A}^\ell\right)^m}}{{\mathbb H}(z)=e}
}
p(x,z)\,=\,
\sum_{
\genfrac{}{}{0pt}{1}{w\in {{\cal A}^\ell}}{H(w)=c}
}
p\big(x,x(j\leftarrow w)\big)\,.
$$
Using
lemma~\ref{mhlump},
we have
\begin{multline*}
\sum_{
\genfrac{}{}{0pt}{1}{w\in {{\cal A}^\ell}}{H(w)=c}
}
p\big(x,x(j\leftarrow w)\big)
\,=\,\sum_{
\genfrac{}{}{0pt}{1}{w\in {{\cal A}^\ell}}{H(w)=c}
}
\frac{1}{m}
\sum_{1\leq i\leq m}
\frac{A(x(i))
M(x(i),w)
}
{A(x(1))+\cdots+A(x(m))}
\cr
\,=\,
\frac{1}{m}
\sum_{1\leq i\leq m}
\frac{
{{A_H\big(H(x(i))\big)} M_H(H(x(i)),c)}
}
{A_H\big(H(x(1))\big)
+\cdots+ A_H\big(H(x(m))\big)}\,.
%\cr
%\,=\,
%\frac{1}{m
%{\lambda} }
%\sum_{1\leq i\leq m}
%{A_H\big(d(i)\big)} M_H(d(i),c)\,\cr
\end{multline*}
This sum is a function of ${\mathbb H}(x)$ and $c$ only.
Since
${\mathbb H}(x)=
{\mathbb H}(y)$, 
the sums are the same for $x$ and $y$.
Suppose next that $d=e$. Then
\begin{multline*}
\sum_{
\genfrac{}{}{0pt}{1}{z\in {\left({\cal A}^\ell\right)^m}}{{\mathbb H}(z)=e}
}
p(x,z)\,=\,
p(x,x)\,+\,
\sum_{1\leq j\leq m}
\sum_{
\genfrac{}{}{0pt}{1}{w\in {{\cal A}^\ell\setminus\{\,x(j)\,\}}}{H(w)=H(x(j))}
}
p\big(x,x(j\leftarrow w)\big)\cr
\,=\,1-
%\sum_{1\leq j\leq m}
\kern-7pt
\sum_{
\genfrac{}{}{0pt}{1}{1\leq j\leq m}{w\in {{\cal A}^\ell\setminus\{\,x(j)\,\}}}
}
\kern-7pt
p\big(x,x(j\leftarrow w)\big)+
\sum_{1\leq j\leq m}
\sum_{
\genfrac{}{}{0pt}{1}{w\in {{\cal A}^\ell\setminus\{\,x(j)\,\}}}{H(w)=H(x(j))}
}
p\big(x,x(j\leftarrow w)\big)\cr
\,=\,1-
\sum_{1\leq j\leq m}
\sum_{
\genfrac{}{}{0pt}{1}{w\in {{\cal A}^\ell\setminus\{\,x(j)\,\}}}{H(w)\neq H(x(j))}
}
p\big(x,x(j\leftarrow w)\big)\cr
\,=\,1-
\sum_{1\leq j\leq m}
\sum_{
\genfrac{}{}{0pt}{1}{c\in \zl}{c\neq H(x(j))}
}
\sum_{
\genfrac{}{}{0pt}{1}{w\in {{\cal A}^\ell}}{H(w)=c}
}
p\big(x,x(j\leftarrow w)\big)
\,.
\end{multline*}
We have seen in the previous case that the last sum
is a function of ${\mathbb H}(x)$ and $c$ only. The second sum as well
depends only on 
${\mathbb H}(x)$. Therefore the above quantity is the same for $x$ and $y$.
\end{proof}

\par\noindent
We apply the classical lumping result (see theorem~\ref{lumpt})
to conclude that the 
distance process 
$(D_t)_{t\geq 0}$
is a Markov process. From the previous computations, we see that
its
transition matrix $p_H$ is given by
\begin{multline*}
\forall d\in\zlm\quad\forall j\in\um
\quad
\forall c \in\zl \setminus\{\,d(j)\,\}\cr
%\forall u\neq x(j)\cr
%\forall u \in\smash{{\cal A}^\ell} \setminus\{\,x(j)\,\}\cr
p_H\big(d,d(j\leftarrow c)\big)\,=\,
\frac{1}{m}
\sum_{1\leq i\leq m}
\frac{
{{A_H(d(i))} M_H(d(i),c)}
}
{A_H(d(1))
+\cdots+ A_H(d(m))}\,.
\end{multline*}
%p\big(x,x\big)\,=\,
%1-\frac{1}{m
%{\lambda} }
%\sum_{1\leq i\leq m}
\subsection{Occupancy process}
We denote by $\pml \index{$\pml$}$ the set of
the ordered partitions of the integer $m$ in at most 
$\ell+1$ parts:
$$
\pml\,=\,
\big\{\,(o(0),\dots,o(\ell))\in{\mathbb N}^{\ell+1}:
o(0)+\cdots+o(\ell)=m\,\big\}\,.$$
These partitions are interpreted as occupancy distributions. The 
partition
$(o(0),\dots,o(\ell))$ 
corresponds to a population in which
$o(l)$ chromosomes are at Hamming distance $l$ from the
master sequence,
for any $l\in\{\,0,\dots,\ell\,\}$.
Let $\mathcal O \index{$\cO$}$ be the map which associates to each population $x$
its occupancy distribution ${\mathcal O}(x)= (o(x,0),\dots,o(x,\ell))$, 
defined by:
$$\forall l\in\{\,0,\dots,\ell\,\}\qquad
o(x,l)\,=\,\card\big\{\,i:
1\leq i\leq m,\,d_H(x(i),w^*)=l\,\big\}\,.$$
The map $\cO$
can be factorized through ${\mathbb H}$. 
For $d\in\zlm$, we set
$$%\forall l\in\{\,0,\dots,\ell\,\}\qquad
o_H(d,l)\,=\,\card\big\{\,i:
1\leq i\leq m,\,d(i)=l\,\big\}$$
and we define a map
$\cO_H:\zlm\to\pml 
\index{$\cO_H$}$
by setting
$$%\forall d\in\zlm \qquad
\cO_H(d)\,=\, (o_H(d,0),\dots,o_H(d,\ell))\,.$$
We have then
$$\forall x\in\Alm\qquad
\cO(x)\,=\,\cO_H\big({\mathbb H}(x)\big)\,.$$
The map $\cO$ lumps together populations which are permutations
of each other:
$$\forall x\in\Alm\quad\forall \rho\in\Sm\qquad
\cO(x)\,=\,
\cO(\rho\cdot x)\,.$$
We define the occupancy process
$(O_t)_{t\geq 0} \index{$\O_t$}$ by setting
$$\forall t\geq 0\qquad 
O_t\,=\,{\mathcal O}(X_t)
\,=\,{\mathcal O}_H(D_t)
\,.$$
For the process
$(D_t)_{t\geq 0}$, the only transitions having positive probability
are the transitions of the form
$$d\quad\longrightarrow\quad d(j\leftarrow c)\,,\qquad
1\leq j\leq m,\quad c\in\zl\,.$$
Therefore the only possible transitions for the process
$(O_t)_{t\geq 0}$ are
$$o\quad\longrightarrow\quad o(k\rightarrow l)\,,\qquad 0\leq k,l\leq\ell\,,$$
where
$o(k\rightarrow l) \index{$o(k\rightarrow l)$}$
 is the partition obtained by moving a chromosome from
the class $k$ to the class $l$, i.e.,
$$\forall h\in\{\,0,\dots,\ell\,\}\qquad
o(k\rightarrow l)(h)\,=\,
\begin{cases}
o(h) &\text{if }h\neq k,l\\
o(k)-1&\text{if }h=k\\
o(l)+1&\text{if }h=l\\
\end{cases}
$$
\begin{proposition}[${\mathcal O}$ Lumpability]\label{olump}
%For any $e\in
%\{\,0,\dots,\ell\,\}^m$, any $x,y\in
%{\left({\cal A}^\ell\right)^m}$, 
Let $p_H$ be the transition matrix of the distance process.
We have
\begin{multline*}
\forall o\in\pml
\quad
\forall 
d,e\in
\zlm\,,\cr
%{\left({\cal A}^\ell\right)^m}\,,\cr
{\mathcal O}_H(d)=
{\mathcal O}_H(e)\quad\Longrightarrow\quad
\sum_{
\genfrac{}{}{0pt}{1}{f\in\zlm }
{{\mathcal O}_H(f)=o}
}
p_H(d,f)\,=\,
\sum_{
\genfrac{}{}{0pt}{1}{f\in \zlm}
{{\mathcal O}_H(f)=o}
}
p_H(e,f)\,.
\end{multline*}
\end{proposition}
\begin{proof}
Let
$o\in\pml$ and
$d,e\in \zlm$ such that
${{\mathcal O}_H(d)=
{\mathcal O}_H(e)}$.
%Suppose first that
%${{\mathcal O}_H(d)= o}$.
%Then
%$$\sum_{
%\genfrac{}{}{0pt}{1}{f\in\zlm }
%{{\mathcal O}_H(f)=o}
%}
%p_H(d,f)\,=\,p_H(d,d)\,.$$
%${\mathcal O}_H(d)=
%{\mathcal O}_H(e)$. We suppose also that
Since ${\mathcal O}_H(d)=
{\mathcal O}_H(e)$, then
there exists a permutation
$\rho\in\Sm$ such that
$\rho\cdot d=e$.
%Using the exchangeability of the 
By lemma~\ref{matinv}, the transition matrices $p$ and $p_H$
are invariant under the action of $\Sm$, therefore
\begin{multline*}
\sum_{
\genfrac{}{}{0pt}{1}{f\in\zlm }
{{\mathcal O}_H(f)=o}
}
p_H(d,f)\,=\,
\sum_{
\genfrac{}{}{0pt}{1}{f\in\zlm }
{{\mathcal O}_H(f)=o}
}
p_H(\rho\cdot d,\rho\cdot f)\cr
\,=\,
\sum_{
\genfrac{}{}{0pt}{1}{f\in\zlm }
{{\mathcal O}_H(\rho^{-1}\cdot f)=o}
}
p_H(e,f)
\,=\,
\sum_{
\genfrac{}{}{0pt}{1}{f\in\zlm }
{{\mathcal O}_H(f)=o}
}
p_H(e,f)
\end{multline*}
as requested.
\end{proof}

\noindent
We apply the classical lumping result (see theorem~\ref{lumpt})
to conclude that the 
occupancy process
$(O_t)_{t\geq 0}$
is a Markov process. 
Let us compute its transition probabilities.
Let
$o\in\pml$ and
$d\in \zlm$ be such that
${{\mathcal O}_H(d)\neq o}$.
Let us consider the sum
$$\sum_{
\genfrac{}{}{0pt}{1}{f\in\zlm }
{{\mathcal O}_H(f)=o}
}
p_H(d,f)\,.$$
The terms in the sum vanish unless
$$\exists\,j\in\um\quad
\exists\,c\in\zl\,,\quad c\neq d(j)\,,\qquad
f=d(j\leftarrow c)\,.$$
%{\left({\cal A}^\ell\right)^m}\,,\cr
Suppose that it is the case. If in addition $f$ is such that
${{\mathcal O}_H(f)=o}$, then
$$o\,=\,
{{\mathcal O}_H(d)}(d(j)\rightarrow c)\,.$$
Setting $k=d(j)$ and $l=c$, we conclude that
$$\exists\,k,l\in\zl\qquad o\,=\,
{{\mathcal O}_H(d)}(k\rightarrow l)\,.$$
The two indices $k,l$ satisfying the above condition are distinct and
unique.
We have then
\begin{multline*}
\sum_{
\genfrac{}{}{0pt}{1}{f\in\zlm }
{{\mathcal O}_H(f)=o}
}
p_H(d,f)\,=\,
\sum_{
\genfrac{}{}{0pt}{1}{j\in\um }
{d(j)=k}
}
p_H(d,d(j\leftarrow l))\cr
\,=\,
\sum_{
\genfrac{}{}{0pt}{1}{j\in\um }
{d(j)=k}
}
\!\!
\frac{1}{m}
\sum_{1\leq i\leq m}
\frac{
{{A_H(d(i))} M_H(d(i),l)}
}
{A_H(d(1))
+\cdots+ A_H(d(m))}\,\cr
\,=\,
%\frac{1}{m}{{\mathcal O}_H(d)}(k)
%\frac{
%{{\mathcal O}_H(d)}(k)
%}{m}
\frac{\displaystyle
{{\mathcal O}_H(d)}(k)
%\sum_{h=0}^\ell
\sum_{0\leq h\leq\ell}
{{\mathcal O}_H(d)}(h)\,
{A_H(h)}\, M_H(h,l)}
{\displaystyle
m
%\sum_{h=0}^\ell
\sum_{0\leq h\leq\ell}
{{\mathcal O}_H(d)}(h)\,
{A_H(h)}}
\,.
%\cr
%\,=\,
%\frac{1}{m
%{\lambda} }
%\sum_{1\leq i\leq m}
%{A_H\big(d(i)\big)} M_H(d(i),c)\,\cr
\end{multline*}
This fraction is a function of 
${{\mathcal O}_H(d)}$, $k$ and $l$, thus it
depends only on 
${{\mathcal O}_H(d)}$ and $o$ as requested.
%Since
%${{\mathcal O}_H(d)}=
%{{\mathcal O}_H(e)}$,
%the sums are the same for $d$ and $e$.
We conclude 
that the transition matrix of the occupancy process 
%$(X_n)_{n\geq 0}$ 
is given by
\begin{multline*}
\forall o\in\pml\quad\forall k,l\in\zl\,,
\quad k\neq l\,,\cr
%\forall l \in\zl \setminus\{\,k\,\}\cr
%\forall u\neq x(j)\cr
%\forall u \in\smash{{\cal A}^\ell} \setminus\{\,x(j)\,\}\cr
p_{O}\big(o,o(k\rightarrow l)\big)\,=\,
%\frac{1}{m
%{\lambda} }\,
%\sum_{1\leq i\leq m}
%f(k)\,{A_H(k)}\, M_H(k,l)\,.\cr
%{{\mathcal O}_H(d)}(k)
%o(k)
%\sum_{h=0}^\ell
%%{{\mathcal O}_H(d)}(h)\,
%o(h)\,
%{A_H(h)}\, M_H(h,l)\,.
\frac{\displaystyle
o(k)
\sum_{h=0}^\ell
o(h)\,
{A_H(h)}\, M_H(h,l)}
{\displaystyle
m\sum_{h=0}^\ell
o(h)\,
{A_H(h)}}\,.
\end{multline*}
%p\big(x,x\big)\,=\,
%1-\frac{1}{m
%{\lambda} }
%\sum_{1\leq i\leq m}
\subsection{Invariant probability measures}\label{inmeas}
There are several advantages
in working with the lumped processes.
The main advantage is that the state space
is considerably smaller.
For the process
$(X_t)_{t\geq 0}$, the cardinality of the state space is
$$\card\Alm\,=\,\kappa^{\ell m}\,.$$
For the distance process
$(D_t)_{t\geq 0}$, it becomes
$$\card\zlm\,=\,(\ell+1)^{m}\,.$$
Finally for the occupancy process, the cardinality is 
the number of ordered partitions of $m$ into at most
$\ell +1$ parts. This number is quite complicated to compute, 
but in any case
$$\card\pml\,\leq\, (\ell+1)^{m}\,.$$
Our goal is to estimate the law $\nu$ of the 
fraction of the master sequence in the population
at equilibrium.
The probability measure $\nu$ is the probability measure on the interval $[0,1]$
%describing the equilibrium density of the master sequence in the population.
satisfying the following identities. 
For any function $f:[0,1]\to\R$,
$$\int_{[0,1]}f\,d\nu
\,=\,
\lim_{t\to\infty} 
E\Big(
\,f\Big(
\frac{1}{m}
N(X_t)\Big)
\Big)
\,=\,
%\frac{1}{m}
\int_{\textstyle\Alm}
f\Big(
\frac{1}{m}
N(x)\Big)
\,d\mu(x)$$
where $\mu$ is the invariant probability measure of the process
$(X_t)_{t\geq 0}$ and
$N(x)$ is the number of copies of the master sequence~$w^*$
present in the population $x$:
$$N(x)\,=\,\card\big\{\,i: 1\leq i\leq m, \, x(i)=w^*\,\big\}\,.$$
In fact, the probability measure $\nu$ 
is the image of $\mu$
through the map
$$x\in\Alm\mapsto
\frac{1}{m}
N(x)\in[0,1]\,.$$
Yet $N(x)$ is also lumpable with respect to $\mathbb H$, i.e., it can be
written as a function of ${\mathbb H}(x)$:
$$\forall x\in\Alm\qquad
N(x)\,=\,
N_H({\mathbb H}(x))\,,$$
%\card\big\{\,i: 1\leq i\leq m, \, 
%{\mathbb H}(x)(i)=0\,\big\}\,.$$
where $N_H$ is the lumped function defined by
$$\forall d\in\zlm\qquad
N_H(d)
\,=\,\card\big\{\,i: 1\leq i\leq m, \, 
d(i)=0\,\big\}\,.$$
Let $\mu_H$ be the invariant probability measure of
the process
$(D_t)_{t\geq 0}$. 
For $d\in\zlm$, we have
\begin{multline*}
\mu_H(d)\,=\,
\lim_{t\to\infty}  P\big(D_t=d\big)
\,=\,
\lim_{t\to\infty}  P\big({\mathbb H}(X_t)=d\big)
\cr
\,=\,
\lim_{t\to\infty}  P\big(
X_t\in
{\mathbb H}^{-1}(
d)\big)
\,=\,
\mu\big(
{\mathbb H}^{-1}(
d)\big)
\,.
\end{multline*}
Thus, as it was naturally expected, the probability measure $\mu_H$
is the image of the probability measure $\mu$ through the map
$\mathbb H$.
It follows that,
for any function $f:[0,1]\to\R$,
\begin{multline*}
\int_{[0,1]}f\,d\nu
\,=\,
\int_{\textstyle\Alm}
f\Big(
\frac{1}{m}
N(x)\Big)
\,d\mu(x)\cr
\,=\,
\int_{\textstyle\Alm}
f\Big(
\frac{1}{m}
N_H\big({\mathbb H}(x)\big)
\Big)
\,d\mu(x)\cr
\,=\,
\int_{\textstyle\zlm}
f\Big(
\frac{1}{m}
N_H(d)\Big)
\,d\mu_H(d)\,.
\end{multline*}
Similarly, the invariant probability measure
$\mu_O$ 
%be the invariant probability measure 
of the process
$(O_t)_{t\geq 0}$ is the image measure of $\mu$ through the map $\cO$,
and also 
the image measure of $\mu_H$ through the map $\cO_H$. We have also, 
for any function $f:[0,1]\to\R$,
\begin{equation*}
%\wild(\sigma,\ell,m,p)\,=\,
%\frac{1}{m}
%\int_{\textstyle\zlm}
%N_H(d)\,d\mu_H(d)
%\,=\,
\int_{[0,1]}f\,d\nu
\,=\,
\int_{\textstyle\pml}
f\Big(\frac{1}{m}
o(0)\Big)\,d\mu_O(o)
\,.
\end{equation*}
Another advantage of the lumped processes
is that the spaces $\zlm$ and $\pml$ are naturally
endowed with a partial order.
Since we cannot deal directly with the distance process 
$(D_t)_{t\geq 0}$ or the occupancy process
$(O_t)_{t\geq 0}$, we shall compare them with auxiliary processes
whose dynamics is much simpler.
\vfill\eject
\section{Monotonicity}\label{secmono}
A crucial property for comparing the Moran model with other processes
is monotonicity. We will realize a coupling of 
the lumped Moran processes with different
initial conditions and we will deduce the monotonicity from the coupling
construction.
\subsection{Coupling of the lumped processes}
\label{coulum}
We build here a coupling of the lumped processes, on the same probability space
as the coupling for the process
$(X_t)_{t\geq 0}$ described in section~\ref{couplsec}.
We set
$$\forall n\geq 1\qquad
R_n\,=\,\big(S_n, I_n, J_n,U_{n,1},\dots,U_{n,\ell}\big)\,.
\index{$R_n$}
$$
The vector $R_n$ is the random input which is used to perform the
$n$--th step of the Markov chain
$(X_t)_{t\geq 0}$. By construction the sequence 
$(R_n)_{n\geq 1}$ is a sequence of independent identically distributed
random vectors with values in
$${\cR}\,=\,[0,1]\times\um^2\times [0,1]^\ell\,.
\index{$\cR$}$$
We first define two maps $\cMH$ and $\cS_H$ in order to couple the
mutation and the selection mechanisms.
\smallskip

\par\noindent
{\bf Mutation.}
We define a map %$\cMH$ 
$$\cMH:\zl\times [0,1]^\ell\to\zl
\index{$\cMH$}$$
in order to
couple the mutation mechanism starting with different 
chromosomes. 
Let $b\in\zl$ and let $u_1,\dots,u_\ell\in
[0,1]^\ell$.
%We define separately the maps $\cM^1_H$ and $\cM_H^\ell$ associated 
%to each mutation mechanism.
%
%\noindent
%{\it One point mutation.} 
%The map $\cM_H^1$ is defined by setting
%$$\cM_H^1(b,
%u_1,\dots,u_\ell,l)=
%b-1_{u_l<p/\kappa}1_{l\leq b}
%+1_{u_l>1-p(1-1/\kappa)}1_{l> b}
%\,.$$
%Let $\theta$ be the map
%where $c\in\zl$
%$$\forall h\in\{\,1,\dots,\ell\,\}\qquad
%\noindent
%{\it Simultaneous mutations.} 
The map $\cM_H$ is defined by setting
$$\cM_H(b,
u_1,\dots,u_\ell)=
b-\sum_{k=1}^b1_{u_k<p/\kappa}
+\sum_{k=b+1}^\ell1_{u_k>1-p(1-1/\kappa)}
\,.$$
The map $\cMH$ is built in such a way that,
if $U_1,\dots,U_\ell$ are
random variables 
with uniform law on 
the interval 
$[0,1]$,
all being independent, then for any
$b\in\zl$, the law of 
$\cMH(b,
U_1,\dots,U_\ell)$
is given by the line of the mutation matrix $M_H$ associated to $b$,
i.e.,
$$\forall c\in\zl\qquad P\big(
\cMH(b,
U_1,\dots,U_\ell)=c\big)\,=\,M_H(b,c)\,.$$
%Let $\theta$ be the map
%$$\theta:\zl\times [0,1]^\ell\to\zl$$
%defined by
%\begin{multline*}
%\forall b \in\zl \quad\forall u_1,\dots,u_\ell\in [0,1]\cr
%\theta(b,u_1,\dots,u_\ell)\,=\,
%%b-\sum_{k=1}^b1_{\textstyle u_k<p/\kappa}
%b-\sum_{k=1}^b1_{u_k<p/\kappa}
%+\sum_{k=b+1}^\ell1_{u_k>1-p(1-1/\kappa)}\,.
%\end{multline*}
%\smallskip
%
%\par\noindent
{\bf Selection for the distance process.}
We realize the replication mechanism with the help of a selection
map
$$\cS_H:\zlm\times [0,1]\to\um\,.
\index{$\cS_H$} $$
Let $d\in\zlm$ and let $s\in[0,1[$. We define 
$\cS_H(d,s)=i$ where $i$ is the unique index in $\um$ satisfying
$$
%\frac{1}{m\lambda}
%\sum_{1\leq k\leq i-1}
\frac{A_H(d(1))
+\cdots+ A_H(d(i-1))}
{A_H(d(1))
+\cdots+ A_H(d(m))}
\,\leq\,s\,<\,
%\frac{1}{m\lambda}
%\sum_{1\leq k\leq i}
\frac{A_H(d(1))
+\cdots+ A_H(d(i))}
{A_H(d(1))
+\cdots+ A_H(d(m))}
\,.$$
The map $\cS_H$ is built in such a way that,
if $S$ is
a random variable 
with uniform law on 
the interval 
$[0,1]$, 
%and if $L$ is
%We define next the maps 
%$\Phi_H^c$ and
%$\Phi_H^d$ corresponding to the
%the coupled and the decoupled models.
%and if $L$ is
%a random variable with uniform law on 
%the index set, all being independent, 
then for any
$d\in\zlm$, 
the law of 
$\cS_H(d,S)$
is given by 
%$$\displaylines{
%P\big(
%\cS(x,
%S)=0\big)\,=\,
%1-\frac{1}{m\lambda}
%\sum_{1\leq i\leq m}A(x,i)\,,
%\cr
%\begin{cases}
%\frac{1}{m\lambda}
%\big(A(x,1)+\cdots+A(x,m)\big)
%&\text{if }i=0\\
%\frac{1}{m\lambda}
$$\forall i\in\um\qquad P\big(
\cS_H(d,
S)=i\big)\,=\,
\frac{
{{A_H(d(i))} }
}
{A_H(d(1))
+\cdots+ A_H(d(m))}\,.$$
%\begin{cases}
%\frac{1}{m\lambda}
%\big(A(x,1)+\cdots+A(x,m)\big)
%&\text{if }i=0\\
%\frac{1}{m\lambda}
%A(x,i)
%&\text{if }i\geq 1\\
%\end{cases}
%\,.}$$
%For
%$r=(s,t,i,j,u_1,\dots, u_\ell)\in\cR$, $d\in\zlm$,
%\smallskip
%
%\par\noindent
{\bf Coupling for the distance process.}
%map $\Phi_H$.}
We build a deterministic map
%$\Phi_H$ 
$$\Phi_H:
\zlm\times
{\cR}\to\zlm
\index{$\Phi_H$}
$$
in order to realize the coupling between distance processes
with various initial conditions and
different parameters $\sigma$ or $p$. 
The coupling map $\Phi_H$ is defined by
\begin{multline*}
\forall r\,=\,(s,i,j,u_1,\dots, u_\ell)\in\cR\,,\quad
\forall d\in\zlm\cr
\Phi_H(d,r)\,=\,
d\big(j\leftarrow
\cMH(
d(\cS_H(d,s))
,u_1,\dots,u_\ell)\big)\,.
\end{multline*}
%\begin{cases}
%\quad d&\text{if }
%\cS_H(d,s)
%\geq 1\\
%\quad d\big(j\leftarrow
%\cMH(
%d(\cS_H(d,s))
%),u_1,\dots,u_\ell\big)
%&\text{if }
%\cS_H(d,s)
%= 0 \\
%\end{cases}
%\begin{multline*}
%\forall b \in\zl \quad\forall u_1,\dots,u_\ell\in [0,1]\cr
%\cM(b,u_1,\dots,u_\ell)\,=\,
%%b-\sum_{k=1}^b1_{\textstyle u_k<p/\kappa}
%b-\sum_{k=1}^b1_{u_k<p/\kappa}
%+\sum_{k=b+1}^\ell1_{u_k>1-p(1-1/\kappa)}\,.
%\end{multline*}
%For
%$r=(s,t,i,j,u_1,\dots, u_\ell)\in\cR$, $d\in\zlm$,
Notice that the index~$i$ is not used in the map $\Phi_H$.
The coupling is then built in a standard way with the help of the 
i.i.d. sequence
$(R_n)_{n\geq 1}$ 
and the map
$\Phi_H$.
Let $d\in\zlm$ be the starting point of the process.
We build the distance process 
$(D_t)_{t\geq 0} 
\index{$D_t$}$
by setting
$D(0)=d$ and
$$\forall n\geq 1\qquad
D_n\,=\,\Phi_H\big(D_{n-1}, R_n\big)\,.
$$
A routine check shows that the process
$(D_t)_{t\geq 0}$ 
is a Markov chain starting from $d$
with the adequate transition matrix. This way we have coupled
the distance processes with various initial conditions and
different parameters $\sigma$ or $p$. 
%Notice that our coupling
%is specific to the sharp peak landscape.
\smallskip

\par\noindent
{\bf Selection for the occupancy process.}
We realize the replication mechanism with the help of a selection
map
$$\cS_O:\pml\times [0,1]\to\zl
\index{$\cS_O$}\,.$$
Let $o\in\pml$ and let $s\in[0,1[$. We define 
$\cS_O(o,s)=l$ where $l$ is the unique index in $\zl$ satisfying
$$
%\frac{1}{m\lambda}
%\sum_{1\leq k\leq i-1}
\frac{o(0) A_H(0)
+\cdots+ o(l-1) A_H(l-1)}
{o(0) A_H(0)
+\cdots+ o(\ell) A_H(\ell)}
%\,\leq\,s\,<\,
\leq s<
%\frac{1}{m\lambda}
%\sum_{1\leq k\leq i}
\frac{o(0) A_H(0)
+\cdots+ o(l) A_H(l)}
{o(0) A_H(0)
+\cdots+ o(\ell) A_H(\ell)}
\,.$$
The map $\cS_O$ is built in such a way that,
if $S$ is
a random variable 
with uniform law on 
the interval 
$[0,1]$, 
%and if $L$ is
%We define next the maps 
%$\Phi_H^c$ and
%$\Phi_H^d$ corresponding to the
%the coupled and the decoupled models.
%and if $L$ is
%a random variable with uniform law on 
%the index set, all being independent, 
then for any
$o\in\pml$, 
the law of 
$\cS_O(o,S)$
is given by 
%$$\displaylines{
%P\big(
%\cS(x,
%S)=0\big)\,=\,
%1-\frac{1}{m\lambda}
%\sum_{1\leq i\leq m}A(x,i)\,,
%\cr
%\begin{cases}
%\frac{1}{m\lambda}
%\big(A(x,1)+\cdots+A(x,m)\big)
%&\text{if }i=0\\
%\frac{1}{m\lambda}
$$\forall l\in\zl\qquad P\big(
\cS_O(o,
S)=l\big)\,=\,
\frac{
{{o(l)\,A_H(l)} }
}
{o(0) A_H(0)
+\cdots+ o(\ell) A_H(\ell)}\,.$$
%\begin{cases}
%\frac{1}{m\lambda}
%\big(A(x,1)+\cdots+A(x,m)\big)
%&\text{if }i=0\\
%\frac{1}{m\lambda}
%A(x,i)
%&\text{if }i\geq 1\\
%\end{cases}
%\,.}$$
%For
%$r=(s,t,i,j,u_1,\dots, u_\ell)\in\cR$, $d\in\zlm$,
%\smallskip
%
%\par\noindent
{\bf Coupling for the occupancy process.}
We build a deterministic map
%We perform a similar construction for the occupancy process.
%To this end, 
%we define a deterministic map 
%$\Phi_O$, 
$$
\Phi_O:
\pml\times
{\cR}\to\pml\,
\index{$\Phi_O$}
$$
in order to realize the coupling between occupancy processes
with various initial conditions and
different parameters $\sigma$ or $p$. 
The coupling map $\Phi_O$ is defined as follows.
%For
%$r=(s,t,i,j,u_1,\dots, u_\ell)\in\cR$, $d\in\zlm$,
%The coupling map $\Phi_O$ is defined by
Let
$r=(s,i,j,u_1,\dots, u_\ell)\in\cR$. Let $o\in\pml$, let us 
set
$l=\cS_O(o,s)$
and
let $k$ be the unique index in $\zl$ satisfying
%\begin{align*}
$$o(0)+\cdots+o(k-1)\,<\,j\,\leq
o(0)+\cdots+o(k)\,.
$$
%o(0)+\cdots+o(l-1)\,<\,&\,i\,\leq
%o(0)+\cdots+o(l)\,.
%\end{align*}
%\smallskip
%
The coupling map $\Phi_O$ is defined by
%\begin{equation*}
%\forall r\,=\,(s,i,j,u_1,\dots, u_\ell)\in\cR\,,\quad
%\forall o\in\pml\cr
$$\Phi_O(o,r)\,=\,
 o\big(k\rightarrow
\cMH(l,u_1,\dots,u_\ell)\big)\,.$$
%\begin{cases}
%\quad o&\text{if }k\geq 1\text{ and } s>1/\sigma\\
%\quad o\big(k\rightarrow
%\cM(l,u_1,\dots,u_\ell)\big)
%&\text{if }k= 0 \,\,{\text{ or }}\,\, s<1/\sigma\\
%\end{cases}
%\end{equation*}
Notice that the index~$i$ is not used in the map $\Phi_O$.
Let $o\in\pml$ be the starting point of the process.
We build the occupancy process 
$(O_t)_{t\geq 0} 
\index{$\O_t$}$
by setting
$O(0)=o$ and
$$\forall n\geq 1\qquad
O_n\,=\,\Phi_O\big(O_{n-1}, R_n\big)\,.
$$
A routine check shows that the process
$(O_t)_{t\geq 0}$ 
is a Markov chain starting from $o$
with the adequate transition matrix. This way we have coupled
the occupancy processes with various initial conditions and
different parameters $\sigma$ or $p$. 
%\section{Lower and upper processes}
\subsection{Monotonicity of the model}
The space $\zlm$ is naturally endowed with a partial order:
$$d\leq e\quad \Longleftrightarrow
\quad
\forall i\in\um\quad d(i)\leq e(i)\,.$$
\begin{lemma}\label{monmut}
The map $\cMH$ is non--decreasing with respect to the Hamming class,
	i.e.,
\begin{multline*}
\forall b,c \in\zl \quad
\forall u_1,\dots,u_\ell\in [0,1]\cr
b\leq c\quad\Rightarrow\quad
\cMH(b,u_1,\dots,u_\ell)\,\leq\,
\cMH(c,u_1,\dots,u_\ell)\,.\hfil
\end{multline*}
\end{lemma}
\begin{proof}
We simply use the definition of $\cMH$ 
(see section~\ref{coulum}) and we compute the difference
\begin{multline*}
\cMH(c,u_1,\dots,u_\ell)\,-
\cMH(b,u_1,\dots,u_\ell)\,=\,\cr
c-b
%-\sum_{k=b+1}^c1_{u_k<p/\kappa}
+\sum_{k=b+1}^c\Big(
1_{u_k>1-p(1-1/\kappa)}
-1_{u_k<p/\kappa}\Big)\,.
\end{multline*}
Since $\kappa\geq 2$, the absolute value of the sum is at most $c-b$
and the above difference is non--negative.
\end{proof}
\begin{lemma}\label{monsel}
In the neutral case $\sigma=1$,
the map $\cSH$ is non--decreasing with respect to the Hamming class,
	i.e.,
\begin{multline*}
\forall d,e \in\zlm \quad
\forall s\in [0,1]\cr
d\leq e\quad\Rightarrow\quad
d\big(\cSH(d,s)\big)\,\leq\,
e\big(\cSH(e,s)\big)\,.\hfil
\end{multline*}
\end{lemma}
\begin{proof}
In fact, when $\sigma=1$, the map $\cSH$ depends only on the second variable $s$:
$$
\forall d \in\zlm \quad
\forall s\in [0,1]\qquad 
\cSH(d,s)\,=\,\lfloor m s\rfloor+1\,.
$$
It follows that if 
$d,e \in\zlm$ are such that $d\leq e$, then
$$\forall s\in [0,1]\qquad 
d\big(\lfloor m s\rfloor+1\big)\,\leq\,
e\big(\lfloor m s\rfloor+1\big)
$$
as requested.
\end{proof}
\begin{lemma}\label{monphimor}
In the neutral case $\sigma=1$,
the map $\phi_H$ is non--decreasing with respect to the distances,
	i.e.,
\begin{equation*}
\forall d,e \in\zlm \quad \forall r\in\cR
\,,\qquad
d\leq e
\quad\Rightarrow\quad
\Phi_H(d,r)\,\leq\,
\Phi_H(e,r)\,.
\end{equation*}
\end{lemma}
\begin{proof}
Let
$r=(s,i,j,u_1,\dots, u_\ell)\in\cR$ and let $d,e\in\zlm$, $d\leq e$.
By lemma~\ref{monsel}, we have
$$d\big(\cSH(d,s)\big)\,\leq\,
e\big(\cSH(e,s)\big)\,.$$
This inequality and lemma~\ref{monmut} imply that
$$\cMH\big(
d(\cS_H(d,s))
,u_1,\dots,u_\ell)\big)
\,\leq\,
%\leq
\cMH\big(
e(\cS_H(e,s))
,u_1,\dots,u_\ell\big)\,,
$$
so that
\begin{multline*}
d\big(j\leftarrow
\cMH(
d(\cS_H(d,s))
,u_1,\dots,u_\ell)\big)
\cr
\,\leq\,
%\leq
e\big(j\leftarrow
\cMH(
e(\cS_H(e,s))
,u_1,\dots,u_\ell)\big)\,,
\end{multline*}
whence $\Phi_H(d,r)\,\leq\,
\Phi_H(e,r)$ as requested.
\end{proof}

\noindent
Unfortunately, the map
$\Phi_H$ is not monotone for $\sigma>1$. Indeed, suppose that
%$\kappa=2$, $s<1/\sigma$ and
$$\displaylines{
\kappa=3\,,\quad
\sigma=2\,,\quad m=3\,,\quad
\frac{2}{3}<s<
\frac{3}{4}
\,,\quad\cr
u_1,\dots,u_\ell\in\Big[
\frac{p}{3},1-
\frac{2p}{3}\Big]
\,,\quad j=i=1\,,}$$ 
then
%$\kappa=3$, $m=3$, $2/3<s<3/4$ and
%$u_1,\dots,u_\ell<p/2$, $j=i=1$, then
$$\Phi_H
\left(
\begin{matrix}
0\\
2\\
1
\end{matrix}
\right)
\,=\,
\left(
\begin{matrix}
2\\
2\\
1
\end{matrix}
\right)\,,\qquad
\Phi_H
\left(
\begin{matrix}
1\\
2\\
1
\end{matrix}
\right)\,=\,
\left(
\begin{matrix}
1\\
2\\
1
\end{matrix}
\right)
\,.
$$
%, i.e., elements of ${\cal A}^\ell$.
%This stems from the fact that the distance process 
%is not monotone. 
This creates a serious
complication. This is why we perform a second lumping and we work
with the occupancy process rather than with the distance process.
%\noindent
We define an order $\preceq\index{$\preceq$}$ on
$\pml$
as follows.
Let
$o=(o(0),\dots,o(\ell))$ 
and
$o'=(o'(0),\dots,o'(\ell))$ 
belong to
$\pml$.
We say that $o$ is smaller than or equal to $o'$,
which we denote by $o\preceq o'$, if 
%$$\forall l\leq\ell\qquad
%o(l)+\cdots+o(\ell)\,\geq\,
%o'(l)+\cdots+o'(\ell)\,.$$
%Since $o$ and $o'$ are partitions of $m$, this is equivalent to
$$\forall l\leq\ell\qquad
o(0)+\cdots+o(l)\,\leq\,
o'(0)+\cdots+o'(l)\,.$$
\begin{lemma}\label{monselo}
The map $\cS_O$ is non--increasing with respect to the occupancy distribution,
	i.e.,
\begin{multline*}
\forall o,o' \in
\pml
\quad
\forall s\in [0,1]\cr
o\preceq o'\quad\Rightarrow\quad
\cS_O(o,s)\,\geq\,
\cS_O(o',s)\,.\hfil
\end{multline*}
\end{lemma}
\begin{proof}
Let
$o\preceq o'$.
Let $l\in\zl$. We have
$$o(0) A_H(0)
+\cdots+ o(l) A_H(l)\,=\,
o(0)(\sigma-1)+o(0)+\cdots+o(l)\,.$$
Thus
$$\frac{o(0) A_H(0)
+\cdots+ o(l) A_H(l)}
{o(0) A_H(0)
+\cdots+ o(\ell) A_H(\ell)}
%\,=\,\frac{
%o(0)(\sigma-1)+o(0)+\cdots+o(l)}
%{
%o(0)(\sigma-1)+m
%}
\,=\,\psi\big(o(0),o(0)+\cdots+o(l)\big)\,,$$
where $\psi$ is the function defined by
$$\forall \eta,\xi\in [0,m]\qquad
\psi(\eta,\xi)\,=\,\frac{\eta(\sigma-1)+\xi}
{\eta(\sigma-1)+m}\,.$$
The map $\psi$ is non--decreasing in $\eta$ and $\xi$ on $[0,m]^2$, 
therefore
$$\psi\big(o(0),o(0)+\cdots+o(l)\big)\,\leq\,
\psi\big(o'(0),o'(0)+\cdots+o'(l)\big)\,,$$
%Since the map $A_H$ is non--increasing, 
%%\forall l\in\zl\quad
%then for any $l\in\zl$
%we have
i.e.,
$$\frac{o(0) A_H(0)
+\cdots+ o(l) A_H(l)}
{o(0) A_H(0)
+\cdots+ o(\ell) A_H(\ell)}
\,\leq\,
\frac{o'(0) A_H(0)
+\cdots+ o'(l) A_H(l)}
{o'(0) A_H(0)
+\cdots+ o'(\ell) A_H(\ell)}\,.
$$
It follows that
$\cS_O(o,s)\geq
\cS_O(o',s)$ for any $s\in[0,1]$.
\end{proof}
\begin{lemma}\label{monophio}
%Let $i\in\zm$.
The map $\phi_O$ is non--decreasing with respect to
the occupancy distributions,
	i.e.,
$$\forall o,o' \in\pml \quad \forall r\in\cR
\qquad
%o(0)=o'(0)=i\,,\quad
o\preceq o' 
\quad\Rightarrow\quad
\Phi_O(o,r)\,\preceq\,
\Phi_O(o',r)\,.
$$
\end{lemma}
\begin{proof}
Let
$r=(s,i,j,u_1,\dots, u_\ell)\in\cR$ and let 
$o,o' \in\pml$ be such that 
$o\preceq o'$.
Let us 
set
$l=\cS_O(o,s)$,
$l'=\cS_O(o',s)$
and
let $k,k'$ be the unique indices in $\zl$ satisfying
\begin{align*}
o(0)+\cdots+o(k-1)\,&<\,j\,\leq
o(0)+\cdots+o(k)\,,\cr
o'(0)+\cdots+o'(k'-1)\,&<\,j\,\leq
o'(0)+\cdots+o'(k')\,.
\end{align*}
Since
$o\preceq o'$, then $k\geq k'$. 
Let us set
$$b=\cMH(l,u_1,\dots,u_\ell)\,,\qquad
b'=\cMH(l',u_1,\dots,u_\ell)\,.$$
Since
 $l\geq l'$
by lemma~\ref{monselo}, then $b\geq b'$
by lemma~\ref{monmut}.
We must now compare
$$\Phi_O(o,r)\,=\,o(k\rightarrow b)\,,\qquad
\Phi_O(o',r)
\,=\,o'(k'\rightarrow b')\,.$$
Let $h\in\zl$. We have
$$
o(k\rightarrow b)(0)
+\cdots+
o(k\rightarrow b)(h)
\,=\,
o(0)+\cdots+o(h)-1_{k\leq h}+1_{b\leq h}
\,.
$$
Since $o\preceq o'$, then
$o(0)+\cdots+o(h)\leq
o'(0)+\cdots+o'(h)$.
Since $b\geq b'$, then
$1_{b\leq h}\leq 1_{b'\leq h}$.
The problem comes from the indicator function
$1_{k\leq h}$.
We consider several cases:
\smallskip

\par\noindent
$\bullet $ $k\leq h$. Then
%$1_{k\leq h}=1$ and
\begin{align*}
o(0)+\cdots+o(h)-1_{k\leq h}+1_{b\leq h}
&\,\leq\,
o'(0)+\cdots+o'(h)-1+1_{b\leq h}\cr
&\,\leq\,
o'(0)+\cdots+o'(h)-1_{k'\leq h}+1_{b'\leq h}
\,.
\end{align*}
\noindent
$\bullet $ $k'\leq h<k$. The definition of $k,k'$ implies that
$$o(0)+\cdots+o(h)\,<\,j\,\leq
o'(0)+\cdots+o'(h)$$
whence
$$o(0)+\cdots+o(h)\,\leq
o'(0)+\cdots+o'(h)-1\,.$$ 
It follows that
$$o(0)+\cdots+o(h)
%-1_{k\leq h}
+1_{b\leq h}
\,\leq\,
o'(0)+\cdots+o'(h)-1_{k'\leq h}+1_{b'\leq h}
\,.
$$
\smallskip

\par\noindent
$\bullet $ $h<k'$. Then
$$o(0)+\cdots+o(h)
%-1_{k\leq h}
+1_{b\leq h}
\,\leq\,
o'(0)+\cdots+o'(h)+1_{b'\leq h}
\,.
$$
In each case, we have
$$
o(k\rightarrow b)(0)
+\cdots+
o(k\rightarrow b)(h)
\,\leq\,
o'(k'\rightarrow b')(0)
+\cdots+
o'(k'\rightarrow b')(h)\,.
$$
%\begin{align*}
%o(0)+\cdots+o(h)-1_{k\leq h}+1_{b\leq h}
%&\,\leq\,
%o'(0)+\cdots+o'(h)-1+1_{b\leq h}\cr
%&\,\leq\,
%o'(0)+\cdots+o'(h)-1_{k'\leq h}+1_{b'\leq h}
%\,.
%\end{align*}
%
%It follows that
%$$
 %o\big(k\rightarrow
%\cMH(l,u_1,\dots,u_\ell)\big)\,\preceq\,
 %o'\big(k'\rightarrow
%\cMH(l',u_1,\dots,u_\ell)\big)
%\,$$
%and
Therefore $\Phi_O(o,r)\,\preceq\,
\Phi_O(o',r)$
%$$\Phi_O(o,r)\,=\,
 %o\big(k\rightarrow
%\cM(l,u_1,\dots,u_\ell)\big)
%\,\preceq\,
%\Phi_O(o',r)
 %\,=\,o\big(k\rightarrow
%\cM(l,u_1,\dots,u_\ell)\big)
%$$
as requested.
%\begin{cases}
%\quad o&\text{if }k\geq 1\text{ and } s>1/\sigma\\
%\quad o\big(k\rightarrow
\end{proof}
%\subsection{Positive correlations}
%\label{poco}

\noindent
Let us try to see the implications of the previous results for
the monotonicity of the model (see the appendix
for the definition
of a monotone process).
There is not much to do with the original Moran model, because its
state space is not partially ordered.
So we examine the distance process and the occupancy process.
\begin{corollary}\label{corneu}
In the neutral case $\sigma=1$, 
the distance process 
$(D_t)_{t\geq 0}$ is monotone. 
%If the law of $D_0$ has
%positive correlations, then
%for any $t\geq 0$, the law of $D_t$ also
%has positive correlations.
% and its invariant probability measure
%$\mu_H$
%has positive correlations.
\end{corollary}
Indeed, by lemma~\ref{monphimor}, the map $\Phi_H$ is non--decreasing
in the neutral case $\sigma=1$, hence the coupling is monotone.
%In addition, the jumps of the distance process are up or down and the
%conclusion follows from theorem~\ref{correqu}.
Unfortunately, we did not manage to reach the same conclusion in
the non neutral case.
The main point of lumping further the distance process is to get
a process which is monotone even in the non neutral case.
\begin{corollary}\label{corocc}
The occupancy process 
$(O_t)_{t\geq 0}$ is monotone.
% and its invariant probability measure $\mu_O$ has positive correlations.
%If the law of $O_0$ has
%positive correlations, then
%for any $t\geq 0$, the law of $O_t$ also
%has positive correlations.
\end{corollary}
By lemma~\ref{monophio}, the coupling for the occupancy process is monotone.
%Moreover, all its jumps are up or down.
%The conclusion follows from theorem~\ref{correqu}.
\vfill\eject
\section{Stochastic bounds}\label{secbounds}
In this section, we take advantage of the monotonicity of the map $\Phi_O$
to compare the process 
$(O_t)_{t\geq 0}$ with simpler processes.
\subsection{Lower and upper processes}
We shall construct 
a lower
process
$(O^\ell_t)_{t\geq 0} \index{$O^\ell_t$}$
and an upper
process
$(O^1_t)_{t\geq 0} \index{$O^1_t$}$
satisfying
$$\forall t\geq 0\qquad
O^\ell_t\,\preceq\,
O_t
\,\preceq\,
O^1_t
\,.$$
Loosely speaking, the lower process
evolves as follows. As long as there is no master sequence present
in the population, the process
$(O^1_t)_{t\geq 0}$ evolves exactly as the initial process
$(O_t)_{t\geq 0}$.
When the first master sequence appears, all the other chromosomes
are set in the Hamming class~$1$, i.e., the process jumps to the
state $(1,m-1,0,\dots,0)$. As long as the master sequence is present,
the mutations on non master sequences leading to non master sequences are
suppressed, and any mutation of a master sequence leads to a chromosome
in the first Hamming class. 
The dynamics of the upper process is similar, except that the chromosomes
distinct from the master sequence are sent to the last Hamming 
class~$\ell$
instead of the first one.
%The dynamics of the upper process is similar, the only difference being
%that we put all the non master sequences in the Hamming class~$\ell$
%instead of~$1$.
We shall next construct precisely these dynamics.
We define two maps
$\pi_\ell, \pi_1:\pml\to\pml
\index{$\pi_\ell,\pi_1$}$
by setting 
%We define two maps
%$\pi_1,\pi_\ell:\pml\to\pml$ by setting, for 
%$o\in\pml$,
\begin{align*}
\forall o\in\pml\qquad\
\pi_\ell(o)\,&=\,\big( o(0),0,\dots,0,m-o(0)\big)
\,,\cr
\pi_1(o)\,&=\,\big( o(0),m-o(0),0,\dots,0\big)
\,.
%\qquad
%\pi_\ell(o)\,=\,\big( o(0),0,\dots,0,m-o(0))\,.
\end{align*}
Obviously,
$$
\forall o\in\pml\qquad
%\pi_\ell(o)\,\preceq\,
\pi_\ell(o)
\, \preceq\, 
o 
\, \preceq\, 
\pi_1(o)
\,.$$
We denote by $\cW$ the set of the occupancy distributions
containing the master sequence,
i.e.,
$$\cW\,=\,\big\{\,o\in\pml:o(0)\geq 1\,\big\}
\index{$\cW$}
$$
and 
by $\cN$ the set of the occupancy distributions
which do not contain the master sequence,
i.e.,
$$\cN
\,=\,\big\{\,o\in\pml:o(0)=0\,\big\}\,.
\index{$\cN$}
$$
%\big\{\,o\in\pml:o(0)=0\,\big\}\,.$$
Let $\Phi_O$ be the coupling map defined in
section~\ref{coulum}
We define 
%a lower map $\Phi^{O,<}$ and 
a lower map $\Phi_O^{\ell}$ by setting,
for $o\in\pml$ and $r\in\cR$,
\begin{equation*}
%\Phi^{O,<}(o,r)\,&=\,
%\pi_\ell\big(\Phi^{O}(o,r)\big)\,,
%\begin{cases}
%\quad 
%\Phi^{O}(o,r)
%&\text{if }o\in\cN\\
%\quad 
%\pi_1\big(\Phi^{O}(o,r)\big)
%&\text{if }o\in\cW \,\,{\text{ or }}\,\, 
%\Phi_O(o,r)\in\cW 
%\end{cases}
%\cr
\index{$\Phi_O^{\ell}$}
\Phi_O^{\ell}(o,r)\,=\,
\begin{cases}
\quad 
\Phi_{O}(o,r)
&\text{if }o\in\cN \,\,{\text{ and }}\,\, 
\Phi_O(o,r)\not\in\cW \\
\quad 
\pi_\ell\big(\Phi_{O}(o,r)\big)
&\text{if }o\in\cN \,\,{\text{ and }}\,\, 
\Phi_O(o,r)\in\cW \\
\quad 
\pi_\ell\big(\Phi_{O}(\pi_\ell(o),r)\big)
&\text{if }o\in\cW \\
%\quad 
%\pi_1\big(\Phi^{O}(o,r)\big)
%&\text{if }o\in\cW \,\,{\text{ or }}\,\, 
%\Phi_O(o,r)\in\cW 
\end{cases}
\end{equation*}
Similarly, we define 
%a lower map $\Phi^{O,<}$ and 
an upper map $\Phi_O^{1}$ by setting,
for $o\in\pml$ and $r\in\cR$,
\begin{equation*}
%\Phi^{O,<}(o,r)\,&=\,
%\pi_\ell\big(\Phi^{O}(o,r)\big)\,,
%\begin{cases}
%\quad 
%\Phi^{O}(o,r)
%&\text{if }o\in\cN\\
%\quad 
%\pi_1\big(\Phi^{O}(o,r)\big)
%&\text{if }o\in\cW \,\,{\text{ or }}\,\, 
%\Phi_O(o,r)\in\cW 
%\end{cases}
%\cr
\index{$\Phi_O^{1}$}
\Phi_O^{1}(o,r)\,=\,
\begin{cases}
\quad 
\Phi_{O}(o,r)
&\text{if }o\in\cN \,\,{\text{ and }}\,\, 
\Phi_O(o,r)\not\in\cW \\
\quad 
\pi_1\big(\Phi_{O}(o,r)\big)
&\text{if }o\in\cN \,\,{\text{ and }}\,\, 
\Phi_O(o,r)\in\cW \\
\quad 
\pi_1\big(\Phi_{O}(\pi_1(o),r)\big)
&\text{if }o\in\cW \\
%\quad 
%\pi_1\big(\Phi^{O}(o,r)\big)
%&\text{if }o\in\cW \,\,{\text{ or }}\,\, 
%\Phi_O(o,r)\in\cW 
\end{cases}
\end{equation*}
A direct application of
lemma~\ref{monophio} yields that the map 
$\Phi_O^{\ell}$ is below the map
$\Phi_{O}$ and the map
$\Phi_O^{1}$ is above the map
$\Phi_{O}$ 
in the following sense:
%\begin{lemma}\label{abovephio}
$$ \forall r\in\cR\quad
\forall o\in\pml\qquad
\Phi_O^{\ell}(o,r)
\,\preceq\,
\Phi_{O}(o,r)
\,\preceq\,
\Phi_O^{1}(o,r)
\,.
$$
%\end{lemma}
%Similarly let
%$\pi_\ell:\pml\to\pml$ be the map defined by
%$$\forall f\in\pml\qquad
%\pi_\ell(f)\,=\,\big( f(0),0,\dots,0,m-f(0))\,.$$
We define a lower process 
$(O^\ell_t)_{t\geq 0}$ and
an upper process 
$(O^1_t)_{t\geq 0}$ 
with the help
of the i.i.d. sequence
$(R_n)_{n\geq 1}$ and
the maps $\Phi_O^{\ell}$,
$\Phi_O^{1}$ as follows.
Let $o\in\pml$ be the starting point of the process.
We set
$O^\ell(0)=O^1(0)=o$ and
$$\forall n\geq 1\qquad
O^\ell_n\,=\,\Phi_O^{\ell}\big(O^\ell_{n-1}, R_n\big)\,,
\qquad
O^1_n\,=\,\Phi_O^{1}\big(O^1_{n-1}, R_n\big)\,.
$$
\begin{proposition}\label{domiji}
Suppose that the processes
$(O^\ell_t)_{t\geq 0}$,
$(O_t)_{t\geq 0}$,
$(O^1_t)_{t\geq 0}$,
start from the same occupancy distribution $o$.
We have
$$\forall t\geq 0\qquad
O^\ell_t
\,\preceq\,
O_t
\,\preceq\,
O^1_t
\,.$$
\end{proposition}
\begin{proof}
We prove the inequality by induction over $n\in\mathbb N$.
For $n=0$ we have
$O(0)=
O^\ell(0)=
O^1(0)=
o$. Suppose that the inequality has been proved
at time $t=n\in\mathbb N$, so that
$O_n^\ell\,\preceq\,
O_n\,\preceq\,
O^1_n$.
By construction, we have
$$
O^\ell_{n+1}\,=\,\Phi_O^\ell\big(O^\ell_n, R_n\big)
\,,\quad
O_{n+1}\,=\,\Phi_O\big(O_n, R_n\big)
\,,\quad
O^1_{n+1}\,=\,\Phi_O^1\big(O^1_n, R_n\big)
\,.$$
We use the induction hypothesis and we apply 
lemma~\ref{monophio} to get
$$\Phi_{O}\big(O^\ell_n, R_n\big)\,\preceq
\,\Phi_{O}\big(O_n, R_n\big)
\,\preceq
\Phi_{O}\big(O^1_n, R_n\big)
\,.
$$
Yet
%By proposition~\ref{compupp},
the map 
$\Phi_O^{\ell}$ is below the map 
$\Phi_O$ and
the map 
$\Phi_O^{1}$ is above the map 
$\Phi_O$, 
thus
$$\Phi^\ell_O\big(O^\ell_n, R_n\big)\,\preceq
\,\Phi_{O}\big(O^\ell_n, R_n\big)\,,\qquad
\Phi_O\big(O^1_n, R_n\big)\,\preceq
\,\Phi^1_{O}\big(O^1_n, R_n\big)
.$$
Putting together these inequalities we obtain that
$O^\ell_{n+1}\,\preceq\,
O_{n+1} 
\,\preceq\,
O^1_{n+1}$
and the induction step is completed.
\end{proof}
\subsection{Dynamics
of the bounding processes} 
\label{dynabound}
We study next the dynamics of the processes
$(O^\ell_t)_{t\geq 0}$ 
and
$(O^1_t)_{t\geq 0}$ 
in $\cW$. The computations
are the same for both processes. Throughout the section,
we let $\theta$ be either $1$ or~$\ell$ and we denote by
$(O^\theta_t)_{t\geq 0}$ the corresponding process.
For the process 
$(O^\theta_t)_{t\geq 0}$, the states
$$\cT^\theta\,=\,\big\{\,o\in\pml:
o(0)\geq 1\text{ and }o(0)+o(\theta)<m
\,\big\}\,
\index{$\cT^\theta$}
$$
are transient, 
while the populations in 
$\smash{\cN\cup\big(\cW\setminus\cT^\theta\big)}$ 
form a recurrent class. 
%o(2)+\cdots+o(\ell-1)\geq 1\,\big\}\,.$$
Let us look at the transition mechanism of the process restricted to
$\smash{\cW\setminus\cT^\theta}$.
Since
$$\smash{\cW\setminus\cT^\theta}
\,=\,\big\{\,
o\in\pml:
o(0)\geq 1\text{ and }o(0)+o(\theta)=m
\,\big\}\,,$$
we see that a state of 
$\smash{\cW\setminus\cT^\theta}$
is completely determined by the first occupancy number,
or equivalently the number of copies of the master sequence present in the 
population. 
%Let us set
%$$\forall t\geq 0\qquad N^>_t\,=\,O^>_t(0)\,.$$
Let $\oti$ be the occupancy distribution having one master sequence
and $m-1$ chromosomes in the Hamming class $\theta$:
%$$\oti(0)=1\,,\quad \oti(\theta)=m-1\,,,$$
$$\forall l\in\zl\qquad
\oti(l)\,=\,\begin{cases}
1 &\text{if $l=0$} \\
m-1&\text{if $l=\theta$} \\
0 &\text{otherwise} \\
\end{cases}\,.
\index{$\oti$}
$$
The 
process
$(O^\theta_t)_{t\geq 0}$ always enters the set $\cW\setminus\cT^\theta$ at
$\oti$.
The only possible transitions for the first occupancy number of the
process 
$(O^\theta_t)_{t\geq 0}$
starting from a point in 
$\smash{\cW\setminus\cT^\theta}$ are
\begin{align*}
o(0)\quad &\longrightarrow\quad
o(0)-1
\,,\qquad 1\leq o(0)\leq m\,,\cr
o(0)\quad &\longrightarrow\quad
o(0)+1
\,,\qquad 0\leq o(0)\leq m-1\,.
\end{align*}
%\begin{align*}
%(o(0),o(\theta))\quad &\longrightarrow\quad
%(o(0)-1,o(\theta)+1)
%\,,\qquad 1\leq o(0)\leq m\,,\cr
%(o(0),o(\theta))\quad &\longrightarrow\quad
%(o(0)+1,o(\theta)-1)
%\,,\qquad 0\leq o(0)\leq m-1\,.
%\end{align*}
Let $\oto$
be the occupancy distribution having 
$m$ chromosomes in the Hamming class $\theta$:
$$\forall l\in\zl\qquad
\qquad
\oto(l)\,=\,\begin{cases}
%1 &\text{if $l=0$} \\
m&\text{if $l=\theta$} \\
0 &\text{otherwise} \\
\end{cases}\,.
\index{$\oto$}
$$
%the occupancy distribution $o$ satisfying $o(0)=1$ and $o(\theta)=m-1$,
The 
process
$(O^\theta_t)_{t\geq 0}$ 
always exits
$\cW\setminus\cT^\theta$ at
$\oto$.
%the occupancy distribution $o$ satisfying $o(\theta)=m$.
%Moreover the 
%process
%$(O^>_t(0))_{t\geq 0}$ always enter the set $\cW\setminus\cT$ at
%the point
%$(1,m-1,0,\dots,0)$.
From the previous observations,
we conclude that,
whenever
$(O^\theta_t)_{t\geq 0}$ starts in
$\smash{\cW\setminus\cT^\theta}$,
the dynamics of
$(O^\theta_t(0))_{t\geq 0}$ 
is the one of a standard birth and death process,
until the time of exit from
$\cW\setminus \cT^\theta$. 
We denote 
by $(Z^\theta_t)_{t\geq 0}\index{$Z^\theta_t$}$ 
a birth and death process on $\zm$
starting at $Z^\theta_0=1$ with the following 
transition probabilities:
\smallskip

\par\noindent
$\bullet $ Transitions to the left. For
$i\in\um$,
%\forall i\in\ul\quad
\begin{multline*}
P\big(Z^\theta_{t+1}=i-1\,|\,
Z^\theta_t=i\big)\,=\,
%\gamma_i^\theta\,=\,
 P\big(O^\theta_{t+1}(0)= i-1
\,|\,
O^\theta_t(0)= i\big)
\cr
\,=\,
\frac{\displaystyle \sigma i^2 
\big(1-M_H(0,0) \big)+
i(m-i)
\big(1-M_H(\theta,0) \big)
}{\displaystyle m(\sigma i + m -i)}\,.
%\frac{i}{m(\sigma i + m -i)}
%\Big(\sigma i 
%\big(1-M_H(0,0) \big)+
%(m-i)
%\big(1-M_H(\theta,0) \big)
%\Big)\,.
\end{multline*}
\par\noindent
$\bullet $ Transitions to the right. For
$i\in\{\,0,\dots,m-1\,\}$, 
\begin{multline*}
%\forall i\in\{\,0,\dots,\ell-1\,\}\quad
P\big(Z^\theta_{t+1}=i+1\,|\,
Z^\theta_t=i\big)\,=\,
P\big(O^\theta_{t+1}(0)= i+1
\,|\,
O^\theta_t(0)= i\big)
\cr
\,=\,
\frac{\displaystyle\sigma i(m-i) 
M_H(0,0)+ 
(m-i)^2
M_H(\theta,0) }
{\displaystyle m(\sigma i + m -i)}\,.
\end{multline*}
\subsection{A renewal argument}
Let
$(X_t)_{t\geq 0}$ be a discrete time Markov chain with
values in a finite state space $\cE$ which is irreducible and aperiodic.
Let $\mu$ be the invariant probability 
measure of the Markov chain
$(X_t)_{t\geq 0}$.
\begin{proposition}\label{renewal}
Let $\cW$ be a subset of $\cE$ and let $e$ be a point of $\cE\setminus\cW$.
Let $f$ be a map from $\cE$ to $\R$ which vanishes on $\cE\setminus \cW$.
Let
$$\tau^* \,=\,\inf\,\big\{\,t\geq 0: 
X_t\in\cW
%X_t=(1,m-1,0,\dots,0)
\,\big\}\,,\qquad
\tau \,=\,\inf\,\big\{\,t\geq \tau^*: 
X_t=e
%X_t=(0,m,0,\dots,0)
\,\big\}\,.
$$
We have
$$
\int_\cE f(x)\,d\mu(x)\,=\,
\frac{1}{E(\tau\,|\,X_0=e)}\,
E\bigg(\int_{\tau^*}^{\tau}
f(X_s)\,ds\,\Big|\,X_0=e\bigg)
%\frac{
%\displaystyle
%E\bigg(\int_{\tau^*}^{\tau}
%f(X_s)\,ds\,\Big|\,X_0=e\bigg)
%}{E(\tau\,|\,X_0=e)}\,
\,.$$
\end{proposition}
\begin{proof}
We define two sequences 
$(\tau^*_k)_{k\geq 1}$,
$(\tau_k)_{k\geq 0}$
of stopping times by 
setting $\tau_0=0$ and
\begin{align*}
%\tau_\,0=\,0\,,
&\tau^*_1 \,=\,\inf\,\big\{\,t\geq 0: 
X_t\in\cW
%X_t=(1,m-1,0,\dots,0)
\,\big\}\,,\quad
&&\tau_1 \,=\,\inf\,\big\{\,t\geq \tau^*_1: 
X_t=e
%X_t=(0,m,0,\dots,0)
\,\big\}\,,
\cr
&\,\,\,\vdots
&&
\,\,\,\vdots
\cr
&\tau^*_k \,=\,\inf\,\big\{\,t\geq \tau_{k-1}: 
X_t\in\cW
%X_t=(1,m-1,0,\dots,0)
\,\big\}\,,\quad
&&\tau_k \,=\,\inf\,\big\{\,t\geq \tau^*_k: 
X_t=e
%X_t=(0,m,0,\dots,0)
\,\big\}\,,
\cr
&\,\,\,\vdots
&&
\,\,\,\vdots
\cr
\end{align*}
Our first goal is to evaluate the asymptotic behavior of $\tau_k$ as $k$ goes to
$\infty$.
For any $k\geq 1$,
by the strong Markov property, the trajectory 
$(X_t)_{t\geq \tau_k}$ 
of the process after time $\tau_k$
%$$\big(X_t,\, t\geq \tau_k\big)$$ 
is independent from the trajectory 
$(X_t)_{t\leq \tau_k}$
of
the process until time $\tau_k$,
and its law is the same as the law
of the whole process 
$(X_t)_{t\geq 0}$ starting from $e$.
%Let us consider 
%the process 
%$\big(X_t,\, t\geq 0\big)$ starting from $(m,0,\dots,0)$.
As a consequence,
the successive excursions 
$$\big(X_t,\, \tau_k\leq t\leq \tau_{k+1}\big)\,,\qquad k\geq 1\,,$$
are independent identically distributed. 
%In fact, they have the same
%law as
%$$\big(Z^>(t), m-Z^>(t),0,\dots,0\big)\,,\qquad 0\leq t\leq \tau_0\,,$$
In particular, the sequence
$$\big(\tau_{k+1}-\tau_k\big)_{k\geq 1}$$
is a sequence of i.i.d. random variables, having the same law
as the random time $\tau_1$ whenever the
process 
$(X_t)_{t\geq 0}$ starts from $e$.
%(we made the convention that $\tau_0=0$).
%For $k\geq 1$, reminding that $\tau_0=0$, we decompose $\tau_k$ as the sum
For $k\geq 1$, we decompose $\tau_k$ as the sum
$$\tau_k\,=\,\tau_1+\sum_{h=1}^{k-1}\big(\tau_{h+1}-\tau_{h})\,.$$
We denote by $E_e(\cdot)$ the expectation for 
the process $(X_t)_{t\geq 0}$ starting from~$e$.
Since the state space $\cE$ is finite, then the random time $\tau_1$
is finite with probability one, and it is also integrable.
Applying the classical law of large numbers, we get
$$\lim_{k\to\infty}\,\frac{\tau_k}{k}\,=\,E_e(\tau_1)
%O^>(0)=(m,0,\dots,0)\big)\,
\qquad 
\text{with probability $1$}.$$
Whenever the process $(X_t)_{t\geq 0}$ starts from $e$,
the random times $\tau^*_1$, $\tau_1$ satisfy $\tau^*_1\geq 1$, $\tau_1\geq 2$,
therefore the expected mean $E_e(\tau_1)$ is strictly positive and we conclude that
$$\lim_{k\to\infty}\,{\tau_k}\,=\,+\infty
\qquad 
\text{with probability $1$}.$$
We define next
$$\forall t\geq 0\qquad
K(t)\,=\,\max\,\big\{\,k\geq 0:\tau_k\leq t\,\big\}\,.$$
From the previous discussion, we see that, with probability one,
$K(t)$ is finite for any $t\geq 0$. From the very definition of $K(t)$, 
we have
$$\forall t\geq 0\qquad
\tau_{K(t)}\,\leq\,t\,<\,
\tau_{K(t)+1}\,,$$
and since $\tau_k$ goes to $\infty$ with $k$, then
$$\lim_{t\to\infty}\,{K(t)}\,=\,+\infty
\qquad 
\text{with probability $1$}.$$
We rewrite the previous double inequality as
$$
\frac{\tau_{K(t)}}{K(t)}
\,\leq\,
\frac{t}{K(t)}
\,<\,
\frac{\tau_{K(t)+1}}{K(t)+1}\times
\frac{{K(t)+1}}{K(t)}
\,.$$
Sending $t$ to $\infty$, we conclude that
$$\lim_{t\to\infty}\,
\frac{{K(t)}}{t}
\,=\,
\frac{1}{
E_e(\tau_1)}
\qquad 
\text{with probability $1$}.$$
We suppose that
%the process $\big(X_t,\, t\geq 0\big)$ starts from $(m,0,\dots,0)$,
the process $(X_t)_{t\geq 0}$ starts from $e$.
Let $f$ be a map from $\cE$ to $\R$ which vanishes on $\cE\setminus \cW$.
By the ergodic theorem~\ref{ergodic}, we have
$$
\lim_{t\to\infty} 
E_e\big(
%\frac{1}{m}
f(X_t)\big)
\,=\,
\lim_{t\to\infty} \,
\frac{1}{t}
\int_0^tf(X_s)\,ds\,.$$
We decompose the last integral as follows:
$$
%\frac{1}{t}
\int_0^tf(X_s)\,ds\,=\,
\sum_{k=1}^{K(t)}\int_{\tau^*_k}^{\tau_k}
f(X_s)\,ds\,+\,
\int_{
\tau^*_{K(t)+1}\wedge t
}^{t}
f(X_s)\,ds
\,,$$
where
$\tau^*_{K(t)+1}\wedge t$ stands for
$\min(\tau^*_{K(t)+1}, t)$.
For $k\geq 1$,
the integral
$$N_k\,=\,\int_{\tau^*_k}^{\tau_k}
f(X_s)\,ds$$
is a deterministic function of the excursion
$\big(X_t,\, \tau_{k-1}\leq t\leq \tau_{k}\big)$,
hence the random variables
%$(N_{k})_{k\geq 1}$
$(N_{k},k\geq 1)$
are independent identically distributed.
With probability one,
$K(t)$ goes to $\infty$ as $t$ goes to $\infty$, thus
by the classical law of large numbers, we have
$$\lim_{t\to\infty}\,
\frac{1}{K(t)}
\sum_{k=1}^{K(t)}
N_k\,=\,
E_e(N_1)
%O^>(0)=(m,0,\dots,0)\big)\,
\qquad 
\text{with probability $1$}.$$
Writing
$$
\frac{1}{t}
\int_0^tf(X_s)\,ds
\,=\,
\frac{{K(t)}}{t}\times
\frac{1}{K(t)}
\sum_{k=1}^{K(t)}
N_k\,+\,
\frac{1}{t}
\int_{
\tau^*_{K(t)+1}\wedge t
}^{t}
f(X_s)\,ds
\,,$$
and letting $t$ go to $\infty$,
we conclude
$$\lim_{t\to\infty}\,
\frac{1}{t}
\int_0^tf(X_s)\,ds
\,=\,
\frac{E_e(N_1)}{E_e(\tau_1)}
\qquad 
\text{with probability $1$}\,.$$
This yields the desired formula.
\end{proof}

\noindent

\subsection{Bounds on $\nu$}
\label{bounds}
%$\wild(\sigma,\ell,m,p)$}
We denote by 
%$\mu^<_O$, 
$\mu_O^\ell \index{$\mu^\ell_O$}$, 
$\mu_O \index{$\mu_O$}$, 
$\mu^1_O \index{$\mu^1_O$}$ 
the invariant probability measures of the processes
%$(O^<_t)_{t\geq 0}$,
$(O^\ell_t)_{t\geq 0}$, 
$(O_t)_{t\geq 0}$, 
$(O^1_t)_{t\geq 0}$.
From section~\ref{inmeas}, 
the probability $\nu$
is the image of $\mu_O$
through the map
$$o\in\pml\mapsto
\frac{1}{m}
o(0)\in[0,1]\,.$$
Thus,
for any function $f:[0,1]\to\R$,
$$\int_{[0,1]}f\,d\nu
\,=\,
\int_{\textstyle\pml}
f\Big(\frac{
o(0)
}{m}
\Big)\,d\mu_O(o)
\,=\,
\lim_{t\to\infty} 
E\Big(
f\Big(
\frac{1}{m}
O_t(0)\Big)
\Big)
\,.
$$
%We suppose that the processes
We fix now a non--decreasing function 
$f:[0,1]\to\R$ such that $f(0)=0$.
Proposition~\ref{domiji} yields the inequalities
$$\forall t\geq 0\qquad
f\Big(
\frac{1}{m}
O^\ell_t(0)\Big)
\,\leq\,
f\Big(
\frac{1}{m}
O_t(0)\Big)
\,\leq\,
f\Big(
\frac{1}{m}
O^1_t(0)\Big)
\,.
$$
Taking the expectation and sending $t$ to $\infty$, we get
$$
%\begin{multline*}
\int_{\textstyle\pml}
f\Big(
\frac{
o(0)
}{m}
\Big)
\,d\mu_O^\ell(o)
%\cr
\,\leq\,
\int_{[0,1]}f\,d\nu
%\wild(\sigma,\ell,m,p)
\,\leq\,
%\cr
%\lim_{t\to\infty} 
%E\Big(
%\frac{1}{m}
%O^>_t(0)\Big)
%\,=\,
\int_{\textstyle\pml}
f\Big( \frac{ o(0) }{m} \Big)
\,d\mu_O^1(o)
\,.
$$
%\end{multline*}
We seek next estimates on the above integrals.
The strategy is the same for the lower and the upper integral.
Thus we fix $\theta$ to be either $1$ or~$\ell$ and 
we study 
the invariant probability measure $\mu^\theta_O$.
For the
process 
$(O^\theta_t)_{t\geq 0}$, 
%the states of $\cW\setminus\cT^\theta$ are 
the states of $\cT^\theta$ are 
transient, 
while the populations in 
$\smash{\cN\cup\big(\cW\setminus\cT^\theta\big)}$ 
form a recurrent class. 
%To simplify the analysis, we suppose that the starting point
%of the process is in $\cN$. Then the process
%$(O^>_t)_{t\geq 0}$ will never enter the set $\cT$ and it will 
%visit $\cW$ only on 
%$\cW\setminus \cT$. 
%Let $f:\zl\to\R$ be a non--decreasing function such that $f(0)=0$.
We apply the renewal result of proposition~\ref{renewal}
to the process
$(O^\theta_t)_{t\geq 0}$ restricted to 
$\smash{\cN\cup\big(\cW\setminus\cT^\theta\big)}$,
the set $\cW\setminus\cT^\theta$, 
the occupancy distribution $\oto$
and
the function $o\mapsto f(o(0)/m)$.
Setting
\begin{align*}
\index{$\tau^*$}
\tau^* \,&=\,\inf\,\big\{\,t\geq 0: 
O^\theta_t\in\cW
%X_t=(1,m-1,0,\dots,0)
\,\big\}\,,\cr
\tau \,&=\,\inf\,\big\{\,t\geq \tau^*: 
%O^\theta_t(\theta)= m
O^\theta_t=\oto
\,\big\}\,,
\end{align*}
we have
$$%\frac{1}{m}
%\,\leq\,
\int_{\textstyle\pml}
f\Big( \frac{ o(0) }{m} \Big)
\,d\mu_O^\theta(o)
\,=\,
\frac{
\displaystyle
E\bigg(\int_{\tau^*}^{\tau}
f\Big( \frac{ 
O^\theta_s(0)
 }{m} \Big)
\,ds\,\Big|\,
O^\theta_0=\oto
\bigg)
}{
\displaystyle
E\big(\tau\,|\,
O^\theta_0=\oto
\big)}\,
\,.
$$
Yet, whenever the process
$(O^\theta_t)_{t\geq 0}$ is in $\cW\setminus\cT^\theta$,
%until the time of exit from
%$\cW\setminus \cT$. 
the dynamics of
$(O^\theta_t(0))_{t\geq 0}$ 
is the same as the birth and death process 
$(Z^\theta_t)_{t\geq 0}$ defined at the end of section~\ref{dynabound}.
We suppose that
$(Z^\theta_t)_{t\geq 0}$ starts from $Z_0^\theta=1$.
Let $\tau_0$ be the hitting time of $0$, defined by
$$\tau_0\,=\,\inf\,\big\{\,n\geq 0: Z^\theta_n=0\,\big\}\,.
\index{$\tau_0$}$$
%Suppose that $O^\theta_0=\oto$.
The process
$(O^\theta_t)_{t\geq 0}$ always enters 
$\cW$ at $\oti$ and
it always exits
$\cW\setminus\cT^\theta$
at $\oto$. In particular $\tau$ coincides with the exit time of
$\cW\setminus\cT^\theta$ after $\tau^*$.
From the previous elements, we see that
$\big(O^\theta_t(0),\, {\tau^*}\leq t  \leq {\tau}\big)$
has the same
law as
$\big(Z^\theta_t\,, 0\leq t\leq \tau_0\big)$, whence
$$
E\bigg(\int_{\tau^*}^{\tau}
f\Big(\frac{O^\theta_s(0)}{m}\Big)\,ds\,\Big|\,
O^\theta_0=\oto
\bigg)
\,=\,
E\bigg(\int_{0}^{\tau_0}
f\Big(\frac{Z^\theta_s}{m}\Big)\,ds\,\Big|\,
Z^\theta_0= 1\bigg)\,.
$$
Moreover, using the Markov property, we have
$$
E\Big(\tau-{\tau^*}
\,\big|\,
O^\theta_0=\oto
\Big)
\,=\,
E\Big(\tau
\,\big|\,
O^\theta_0=\oti
\Big)
\,=\,
E\Big({\tau_0}
\,\big|\,
Z^\theta_0= 1\Big)\,.
$$
Reporting back in the formula for the invariant probability measure 
$\mu^\theta_O$, we get
$$ \int_{\textstyle\pml}
f\Big( \frac{ o(0) }{m} \Big)
\,d\mu_O^\theta(o)
\,=\,
\frac{
\displaystyle
E\bigg(\int_{0}^{\tau_0}
f\Big(\frac{Z^\theta_s}{m}\Big)\,ds\,\Big|\,
Z^\theta_0= 1\bigg)
}{
\displaystyle
E\big(\tau^*\,|\,
O^\theta_0=\oto
\big)+
E\big({\tau_0}
\,\big|\,
Z^\theta_0= 1\big)
}\,
\,.
$$
In order to reinterpret this formula, 
we apply the renewal result stated in proposition~\ref{renewal}
to the process
$(Z^\theta_t)_{t\geq 0}$,
the set $\um$, the point $0$ and
the map~$f(\cdot/m)$.
Setting
$$\tau_1 \,=\,\inf\,\big\{\,t\geq 0: 
Z^\theta_t=1
%X_t=(1,m-1,0,\dots,0)
\,\big\}\,,
$$
and denoting by $\nu^\theta$ the invariant probability measure of
$(Z^\theta_t)_{t\geq 0}$,
we have, with the help of the Markov property, 
$$\sum_{i=1}^m 
f\Big( \frac{ i }{m} \Big)\,
\nu^\theta(i)
\,=\,
\frac{
\displaystyle
E\bigg(\int_{0}^{\tau_0}
f\Big(\frac{Z^\theta_s}{m}\Big)\,ds\,\Big|\,
Z^\theta_0= 1\bigg)
}{
\displaystyle
E\big(\tau_1\,|\,
Z^\theta_0= 0
\big)+
E\big({\tau_0}
\,\big|\,
Z^\theta_0= 1\big)
}
\,.
$$
Yet
$$E\big(\tau_1\,|\,
Z^\theta_0= 0
\big)\,=\,
\frac{1}{
P\big(Z^\theta_{1}=1\,|\,
Z^\theta_0=0\big)
}
\,=\,
\,\frac{1 }
{M_H(\theta,0) }
\,.
%\begin{multline*}
%\forall i\in\ul\quad
%P\big(Z^>_{t+1}=i-1\,|\,
$$
We conclude finally that
\begin{multline*}
\int_{\textstyle\pml}
f\Big( \frac{ o(0) }{m} \Big)
\,d\mu_O^\theta(o)
\,=\,
\cr
\frac{
\displaystyle
%E\big(\tau_1\,|\,
%Z^>_0= 0
%\big)+
\,\frac{1 }
{ M_H(\theta,0) }
+
E\big({\tau_0}
\,\big|\,
Z^\theta_0= 1\big)
}{
\displaystyle
E\big(\tau^*\,|\,
O^\theta_0=\oto
\big)+
E\big({\tau_0}
\,\big|\,
Z^\theta_0= 1\big)
}
\sum_{i=1}^m 
f\Big( \frac{ i }{m} \Big)
\nu^\theta(i)
\,.
\end{multline*}
%Denoting by $\nu^\theta$ the invariant probability measure of
%$(Z^\theta_t)_{t\geq 0}$,
%we conclude finally that
%\begin{equation*}
%\frac{1}{m}\sum_{i=1}^m i\nu^\theta(i)\,\leq\,
%\wild(\sigma,\ell,m,p)
%\,.
%\end{equation*}
To estimate the integral, we must estimate each term appearing on the 
right--hand side. 
In section~\ref{bide},
we deal with the terms involving the 
birth and death processes. In section~\ref{disc}, 
we deal with the discovery time $\tau^*$.
\vfill\eject
\section{Birth and death processes}\label{bide}
We first give explicit formulas for a birth and death Markov chain
that are well adapted to our situation. The formula for the invariant probability
measure can be found in classical books, for instance \cite{SH}.
\subsection{General formulas}
We consider a birth and death Markov chain
$(Z_n)_{n\geq 0}$ on the finite set $\zm$ with transition
probabilities given by
% $$%\forall n\geq 0\quad
\begin{align*}
P(Z_{n+1}=i+1\,|\,Z_n=i)
\,&=\,
\delta_i\,,\quad 0\leq i\leq m-1\,,\cr
P(Z_{n+1}=i-1\,|\,Z_n=i)
\,&=\,
\gamma_i\,,\quad 1\leq i\leq m\,,
\index{$\delta_i$}
\index{$\gamma_i$}
\end{align*}
for any $n\geq 0$.
% on the lower and the upper bounds}
We define 
$$
\pi(0)=1\,,\qquad
\pi(i)\,=\,
\frac{{\delta_1\cdots\delta_{i}}}{{\gamma_{1}\cdots\gamma_{i}}}
\,,\qquad
1\leq i\leq m-1\,.
\index{$\pi(i)$}
$$
Let $\tau_0$ be the hitting time of $0$, defined by
$$\tau_0\,=\,\inf\,\big\{\,n\geq 0: Z_n=0\,\big\}\,.$$
We have the following explicit formula for the expected value of $\tau_0$:
$$
%\label{exitbd}
E(\tau_0\,|\,Z_0=1)\,=\,
\sum_{i=1}^m
\frac{1}
{\gamma_{i}}
\pi(i-1)
\,.
$$
Let
$\nu$ be the invariant probability measure of
$(Z_n)_{n\geq 0}$. 
We have the following explicit formula 
for $\nu$:
$$\displaylines{
\nu(0)\,=\,
\frac{
1
}
{\displaystyle 1+
\delta_0\,E(\tau_0\,|\,Z_0=1)
%\sum_{1\leq j\leq m}
%\frac{\delta_0}{\gamma_j}\pi(j-1)
}
\,,\cr
\forall i\in\um\qquad
\nu(i)\,=\,
\frac{
\displaystyle
%\frac{{\delta_0\cdots\delta_{i-1}}}{{\gamma_{1}\cdots\gamma_{i}}}
\frac{\delta_0}{\gamma_i}\pi(i-1)
}
{\displaystyle 1+
\delta_0\,E(\tau_0\,|\,Z_0=1)
%\sum_{1\leq j\leq m}
%\frac{\delta_0}{\gamma_j}\pi(j-1)
}
\,.
}$$
\subsection{The case of
$(Z^\theta_t)_{t\geq 0}$}
We will now apply these formula to the birth and death chains
$(Z^\theta_t)_{t\geq 0}$ introduced at the end of section~\ref{dynabound}.
For these two processes, we have the following explicit formula for the
transition probabilities:
\begin{align*}
\gamma_i
\,=\,
\frac{\displaystyle \sigma i^2 
\big(1-M_H(0,0) \big)+
i(m-i)
\big(1-M_H(\theta,0) \big)
}{\displaystyle m(\sigma i + m -i)}
\,,
&\quad 1\leq i\leq m\,,\cr
\delta_i
\,=\,
\frac{\displaystyle\sigma i(m-i) 
M_H(0,0)+ 
(m-i)^2
M_H(\theta,0) }
{\displaystyle m(\sigma i + m -i)}
\,,
&\quad 0\leq i\leq m-1\,.
\index{$\delta_i$}
\index{$\gamma_i$}
\end{align*}
% on the lower and the upper bounds}
The transition probabilities 
$\delta_i,\gamma_i$ depend on the parameters 
$\sigma,\ell,m,q$ as well as $\theta$.
%We consider the will work with the variables
%$$a\,=\,\ell\, q,\quad\alpha\,=\,\frac{m}{\ell}\,.$$
We seek estimates of 
the expected value of $\tau_0$ and of the asymptotic behavior of $\nu$ 
in the regime where
%$m$ and $\ell$ go to $\infty$, $q$ goes to $0$.
$$m,\ell\to +\infty\,,\quad q\to 0\,.$$
For this reason, we choose the above specific forms
of the formulas, which are well suited for our purposes. 
Since the results are the same for $\theta=1$ and $\theta=\ell$,
we drop the superscript $\theta$ from the notation, and we write
simply
$Z_n$, $\nu$ instead of
$Z^\theta_n$, $\nu^\theta$.
Our first goal is to
estimate the products $\pi(i)$.
%, $1\leq i\leq m-1$. 
We start by studying the ratio
${\delta_i}/{\gamma_i}$.
We have
$$\forall i\in\umu
\qquad
\frac{\delta_i}{\gamma_i}\,=\,
\phi\Big(M_H(0,0),M_H(\theta,0),\frac{i}{m}\Big)\,,
$$
where
$\phi:\,
]0,1]\times[0,1[\times]0,1[\to ]0,+\infty[$ is the function defined by
%$ be the function defined on
%with values in $]0,+\infty[$ given by
$$\phi(\beta,\ve,\rho)\,=\,
\index{$\phi(\beta,\ve,\rho)$}
\frac{\displaystyle (1-\rho)\big(\sigma\beta \rho+(1-\rho)\ve\big)}
{\displaystyle\rho\big(\sigma(1-\beta)\rho+(1-\rho)(1-\ve)\big)}\,.$$
What matters for the behavior of the products $\pi(i)$ is
whether the values of $\phi$ are larger or smaller than~$1$.
The equation $\phi(\beta,\ve,\rho)=1$ can be rewritten as
$$(\sigma-1)\rho^2+(1-\sigma\beta+\ve)\rho-\ve\,=\,0\,.$$
This equation admits one positive root, given by
$$\rho(\beta,\ve)\,=\,
\index{$\rho(\beta,\ve)$}
\frac{1}{2(\sigma-1)}\Big(
\sigma\beta-1-\ve+
\sqrt{(\sigma\beta-1-\ve)^2+4\ve(\sigma-1)}\Big)\,.$$
Therefore we have
\begin{align*}
\phi(\beta,\ve,\rho)
>1 & \quad\text{ if }\quad \rho<\rho(\beta,\ve)\,,\cr
\phi(\beta,\ve,\rho)
<1 & \quad\text{ if }\quad \rho>\rho(\beta,\ve)\,.
\end{align*}
This readily implies that
%in $\umu$, we have
\begin{align*}
1\leq i\leq j\leq\lfloor \rho(\beta,\ve) m\rfloor\qquad
&\Longrightarrow\qquad \pi(i)\,\leq \pi(j)\,,\cr
%m\geq i\geq j\geq\lfloor \rho(\beta,\ve) m\rfloor+1\qquad
\lfloor \rho(\beta,\ve) m\rfloor
\leq i\leq j\leq
m
\qquad
&\Longrightarrow\qquad \pi(i)\,\geq \pi(j)\,,\cr
\pi\big(\lfloor \rho(\beta,\ve) m\rfloor\big)\,\geq\,
&\pi\big(\lfloor \rho(\beta,\ve) m\rfloor+1\big)
\,.
\end{align*}
It follows 
that the product $\pi(i)$ is maximal for 
$i= \lfloor \rho(\beta,\ve) m\rfloor$:
$$\max_{1\leq i\leq m}\pi(i)\,=\,
\pi\big(\lfloor \rho(\beta,\ve) m\rfloor\big)\,.$$
We notice in addition that $\phi(\beta,\ve,\rho)$ is
continuous and non--decreasing with respect to the first
two variables $\beta,\ve$.
%$$\rho^*\,=\,\begin{cases}
%\frac{\displaystyle\sigma\beta-1}{\displaystyle\sigma-1}
%&\text{if $\sigma\beta>1$} \\
%0&\text{if $\sigma\beta\leq 1$} \\
%\end{cases}
%$$
In the next two sections, we compute the relevant asymptotic
estimates on the birth 
and death process.
Lemma~\ref{mhlump} yields
$$\displaylines{
M_H(0,0)\,=\,\Big(1-p\big(1-\frac{1}{\kappa}\big)\Big)^\ell\,,
\cr
M_H(1,0)\,=\,\Big(1-p\big(1-\frac{1}{\kappa}\big)\Big)^{\ell-1}\frac{p}{\kappa}\,,
\quad
M_H(\ell,0)\,=\,\Big(\frac{p}{\kappa}\Big)^{\ell}\,.}$$
As in theorem~\ref{mainth},
we suppose that 
%$\ell$ goes to $\infty$, $m$ goes to $\infty$ and $q$ goes to $0$ 
$$\ell\to +\infty\,,\qquad m\to +\infty\,,\qquad q\to 0\,,$$
%$$\ell\to +\infty\,,\qquad q\to 0\,,$$
in such a way that
$${\ell q} \to a\in ]0,+\infty[\,.$$
%$${\ell q} \to a\in [0,+\infty]\,,
%\qquad\frac{m}{\ell}\to\alpha\in [0,+\infty]\,.$$
%We consider the regime where
%$\ell,m$ go to $\infty$ 
%and
%$q$ to $0$ 
%in such a way that $m/\ell$ and $\ell q$ are kept constant.
%We set
%$$a\,=\,\ell\, q,\quad\alpha\,=\,\frac{m}{\ell}\,.$$
%The associated limit operator is denoted
In this regime, we have
$$\displaylines{
\lim_{
\genfrac{}{}{0pt}{1}{\ell\to\infty,\,
q\to 0}
{{\ell q} \to a}
}
\,
M_H(0,0)\,=\,\exp(-a)\,,\cr
%\lim_{
%\genfrac{}{}{0pt}{1}{\ell,m\to\infty}
%{q\to 0}
%}
\lim_{
\genfrac{}{}{0pt}{1}{\ell\to\infty,\,
q\to 0}
{{\ell q} \to a}
}
\,
M_H(1,0)\,=\,
%\lim_{
%\genfrac{}{}{0pt}{1}{\ell,m\to\infty}
%{q\to 0}
%}
\lim_{
\genfrac{}{}{0pt}{1}{\ell\to\infty,\,
q\to 0}
{{\ell q} \to a}
}
\,
M_H(\ell,0)\,=\,
0\,.}$$
%$$\lim_{\ell\to\infty}\,
%M_H(0,0)\,=\,\exp(-a)\,,\qquad
%\lim_{\ell\to\infty}\,
%M_H(1,0)\,=\,
%\lim_{\ell\to\infty}\,
%M_H(\ell,0)\,=\,
%0\,.$$
\subsection{Persistence time}
In this section, we will estimate the expected hitting time
$E(\tau_0\,|\,Z_0=1)$. This 
%expected hitting time 
quantity
approximates
the persistence time of the master
sequence~$w^*$.
We estimate first the products $\pi(i)$.
\begin{proposition}\label{bde}
%Let $\rho^*=\rho(\beta,0)$. Let $\eta>0$. For $
Let $a \in ]0,+\infty[$.
For $\rho\in[0,1]$, we have
$$
%\lim_{l,m,q\to \infty,\infty,0}
\lim_{
\genfrac{}{}{0pt}{1}{\ell,m\to\infty}
{q\to 0,\,
{\ell q} \to a}
}
\,\frac{1}{m}\ln\pi(\lfloor\rho m\rfloor)\,=\,
\int_0^\rho\ln \phi(
\exa
%\exp(-a)
,0,s)\,ds\,.$$
\end{proposition}
\begin{proof}
Let $\rho\in[0,1]$.
For $m\geq 1$, we have
$$\frac{1}{m}\ln\pi(\lfloor\rho m\rfloor)\,=\,
\frac{1}{m}\sum_{i=1}^{\lfloor\rho m\rfloor}
\ln
\phi\Big(M_H(0,0),M_H(\theta,0),\frac{i}{m}\Big)\,.$$
Let $\ve\in\,]0,\exa[$. For $\ell,m$ large enough and $q$ small enough, we have
$$\big|M_H(0,0)-\exa\big|\,<\,\ve\,,\qquad
0\,<\, M_H(\theta,0) \,<\,\ve\,,$$
therefore, using the monotonicity properties of $\phi$,
\begin{multline*}
\frac{1}{m}\sum_{i=1}^{\lfloor\rho m\rfloor}
\ln
\phi\Big(\exa-\ve,0,\frac{i}{m}\Big)
\,\leq\,
\frac{1}{m}\ln\pi(\lfloor\rho m\rfloor)
\cr
\,\leq\,
\frac{1}{m}\sum_{i=1}^{\lfloor\rho m\rfloor}
\ln
\phi\Big(\exa+\ve,\ve,\frac{i}{m}\Big)\,.
\end{multline*}
These sums are Riemann sums.
Letting $\ell,m$ go to $\infty$ and $q$ go to $0$, we get
\begin{align*}
\liminf_{
\genfrac{}{}{0pt}{1}{\ell,m\to\infty}
{q\to 0,\,
{\ell q} \to a}
}
\,\frac{1}{m}\ln\pi(\lfloor\rho m\rfloor)\,&\geq\,
\int_0^\rho\ln \phi(\exa-\ve,0,s)\,ds\,,\cr
\limsup_{
\genfrac{}{}{0pt}{1}{\ell,m\to\infty}
{q\to 0,\,
{\ell q} \to a}
}
\,\frac{1}{m}\ln\pi(\lfloor\rho m\rfloor)\,&\leq\,
\int_0^\rho\ln \phi(\exa+\ve,\ve,s)\,ds\,.
\end{align*}
We send $\ve$ to $0$ to obtain the result stated in the 
proposition.
\end{proof}

\noindent
We define
$$\rho^*(a)\,=\,\rho(\exa,0)\,=\,\begin{cases}
\quad\displaystyle\frac{\displaystyle\sigma\exa-1}{\displaystyle\sigma-1}
&\text{if $\sigma\exa>1$} \\
\quad\phantom{aaa}0&\text{if $\sigma\exa\leq 1$} \\
\end{cases}
$$
Since 
$\phi(\exa,0,s)>1$ for $s<\rho^*(a)$ and
$\phi(\exa,0,s)<1$ for $s>\rho^*(a)$, then
the integral
$$\int_0^{\rho}\ln \phi(\exa,0,s)\,ds$$
is maximal for $\rho=\rho^*(a)$.
\begin{corollary}\label{exta}
%The rough asymptotics of the expected hitting time of $0$ are given by
Let $a \in ]0,+\infty[$.
The expected hitting time of $0$ starting from $1$ satisfies
$$
%\lim_{
%\genfrac{}{}{0pt}{1}{\ell,m\to\infty}
%{q\to 0}
%}
\lim_{
\genfrac{}{}{0pt}{1}{\ell,m\to\infty}
{q\to 0,\,
{\ell q} \to a}
}
\,\frac{1}{m}\ln
E(\tau_0\,|\,Z_0=1)\,=\,
\int_0^{\rho^*(a)}\ln \phi(\exa,0,s)\,ds\,.$$
\end{corollary}
\begin{proof}
We have the explicit formula
% for the expected value of $\tau_0$:
$$
E(\tau_0\,|\,Z_0=1)\,=\,
\sum_{i=1}^m
\frac{1}
{\gamma_{i}}
\pi(i-1)
%\,.
$$
and the following bounds on $\gamma_i$:
$$\forall i\in\um\qquad 
\frac{1-M_H(0,0)}{m^2}\,\leq\,
\gamma_i\,\leq\,2\sigma
%\frac{\sigma(1-M_H(0,0))}{m^2(\sigma+1)}\,\geq\,
\,.$$
Let $\ve\in\,]0,\exa[$. For $\ell,m$ large enough and $q$ small enough, we have
$$\big|M_H(0,0)-\exa\big|\,<\,\ve\,,\qquad
0\,<\, M_H(\theta,0) \,<\,\ve\,.$$
We first compute an upper bound:
\begin{multline*}
E(\tau_0\,|\,Z_0=1)\,\leq\,
\frac
{m^3
}
{1-M_H(0,0)}
\max_{1\leq i\leq m}\pi(i)
\cr
\,\leq\,
\frac
{m^3}
{1-M_H(0,0)}
\pi\Big(\big\lfloor \rho\big(M_H(0,0),M_H(\theta,0)\big) m\big\rfloor\Big)
\,.
\end{multline*}
%and
%$$E(\tau_0\,|\,Z_0=1)\,\geq\,
%\frac
%{1}
%{2\sigma}
%\pi\Big(\big\lfloor \rho\big(M_H(0,0),M_H(\theta,0)\big) m\big\rfloor\Big)
%\,.
%$$
Using the monotonicity properties of $\phi$, we get
\begin{multline*}
\pi\big(\lfloor \rho\big(M_H(0,0),M_H(\theta,0)\big) m\rfloor\big)
\cr
\,=\,
\prod_{i=1}^
{\lfloor \rho(M_H(0,0),M_H(\theta,0)) m\rfloor}
\phi\Big(M_H(0,0),M_H(\theta,0),\frac{i}{m}\Big)
\cr
\,\leq\,
\prod_{i=1}^
{\lfloor \rho(M_H(0,0),M_H(\theta,0)) m\rfloor}
\phi\Big(\exa+\ve,\ve,\frac{i}{m}\Big)
\cr
\,\leq\,
\prod_{i=1}^
{\lfloor \rho(\exa+\ve,\ve) m\rfloor}
\phi\Big(\exa+\ve,\ve,\frac{i}{m}\Big)\,.
\end{multline*}
The last inequality holds because the product $\pi(i)$ corresponding to the parameters
$\exa+\ve$, $\ve$ is maximal for $i=
{\lfloor \rho(\exa+\ve,\ve) m\rfloor}$.
We obtain that
$$
E(\tau_0\,|\,Z_0=1)\,\leq\,
%\frac
%{2m^2}
%{1-M_H(0,0)}
%\sum_{i=1}^m
%\pi(i-1)
\frac
{m^3
%\pi\big(\lfloor \rho\big(\exa+\ve,\ve\big) m\rfloor\big)
}
{1-M_H(0,0)}
\prod_{i=1}^
{\lfloor \rho(\exa+\ve,\ve) m\rfloor}
\phi\Big(\exa+\ve,\ve,\frac{i}{m}\Big)\,.
$$
%\phi\Big(\exa+\ve,\ve,\frac{j}{m}\Big)
Taking logarithms, we recognize a Riemann sum, 
hence
$$
\limsup_{
\genfrac{}{}{0pt}{1}{\ell,m\to\infty}
{q\to 0,\,
{\ell q} \to a}
}
\,\frac{1}{m}\ln
E(\tau_0\,|\,Z_0=1)\,\leq\,
\int_0^{\rho(\exa+\ve,\ve)}\ln \phi(\exa+\ve,\ve,s)\,ds\,.$$
Conversely,
%E(\tau_0\,|\,Z_0=1)\,\geq\,
%\pi\big(\lfloor \rho^*(a) m\rfloor\big)\,,
\begin{multline*}
E(\tau_0\,|\,Z_0=1)\,\geq\,
\frac{1}{2\sigma}
\prod_{i=1}^
{\lfloor \rho(\exa,0) m\rfloor}
\phi\Big(M_H(0,0),M_H(\theta,0),\frac{i}{m}\Big)
\cr
\,\geq\,
\frac{1}{2\sigma}
\prod_{i=1}^
{\lfloor \rho(\exa,0) m\rfloor}
\phi\Big(\exa-\ve,0,\frac{i}{m}\Big)
\,.
\end{multline*}
%We obtain that
%$$
%E(\tau_0\,|\,Z_0=1)\,\geq\,
%\frac
%{1}
%{2\sigma}
%\prod_{i=1}^
%{\lfloor \rho(\exa-\ve,0) m\rfloor}
%\phi\Big(\exa-\ve,0,\frac{i}{m}\Big)\,.
%$$
Taking logarithms, we recognize a Riemann sum, 
hence
$$
\liminf_{
\genfrac{}{}{0pt}{1}{\ell,m\to\infty}
{q\to 0,\,
{\ell q} \to a}
}
\,\frac{1}{m}\ln
E(\tau_0\,|\,Z_0=1)\,\geq\,
\int_0^{\rho(\exa,0)}\ln \phi(\exa-\ve,0,s)\,ds\,.$$
We let $\ve$ go to $0$ in the upper bound and in the lower bound
to obtain the desired conclusion.
\end{proof}
\subsection{Invariant probability measure}
In this section, we estimate the invariant probability measure of the process
$(Z^\theta_t)_{t\geq 0}$, or rather the numerator of the last formula
of section~\ref{bounds}.
As usual, we drop the superscript $\theta$ from the notation when it is not necessary, 
and we put it back when we need to emphasize the differences between the cases
$\theta=\ell$ and $\theta=1$.
We define, as before
corollary~\ref{exta},
$$\rho^*(a)\,=\,\rho(\exa,0)\,=\,\begin{cases}
\quad\displaystyle\frac{\displaystyle\sigma\exa-1}{\displaystyle\sigma-1}
&\text{if $\sigma\exa>1$} \\
\quad\phantom{aaa}0&\text{if $\sigma\exa\leq 1$} \\
\end{cases}
$$
%We have
%$$\nu(0)\,=\,
%\frac{
%1
%}
%{\displaystyle 1+
%\sum_{1\leq j\leq m}
%\frac{\delta_0}{\gamma_j}\pi(j-1)
%}\,=\,
%\frac{
%1
%}
%{\displaystyle 1+
%\delta_0
%E(\tau_0\,|\,Z_0=1)}
%\,.$$
%If $\delta_0
%E(\tau_0\,|\,Z_0=1)$ goes to $0$ as $m$ goes to $\infty$, then
%$\nu(0)$ converges to $1$, and
%$$\sum_{1\leq i\leq m}\frac{i}{m}\nu(i)\,\leq\,1-\nu(0)\,$$
%which converges to $0$. We suppose now that
%$\delta_0
%E(\tau_0\,|\,Z_0=1)$ goes to $\infty$ as $m$ goes to $\infty$.
Let
$f:[0,1]\to\R$ 
be a non--decreasing function 
such that $f(0)=0$.
We have the formula
$$
\sum_{1\leq i\leq m}f\Big(\frac{i}{m}\Big)\nu(i)\,=\,
\frac{
{\delta_0}\displaystyle\sum_{1\leq i\leq m}f\Big(\frac{i}{m}\Big)
\displaystyle
\frac{1}{\gamma_i}\pi(i-1)
}
{\displaystyle 1+
\delta_0
E(\tau_0\,|\,Z_0=1)
%\sum_{1\leq j\leq m}
%\frac{\delta_0}{\gamma_j}\pi(j-1)
}
\,.
$$
%$$\forall i\in\um\qquad
%\nu(i)\,=\,
%\frac{
%\displaystyle
%\frac{\delta_0}{\gamma_i}\pi(i-1)
%}
%{\displaystyle 1+
%\delta_0
%E(\tau_0\,|\,Z_0=1)
%%\sum_{1\leq j\leq m}
%%\frac{\delta_0}{\gamma_j}\pi(j-1)
%}
%\,
%$$
Moreover
$\delta_0=M_H(\theta,0)$, thus
the numerator of the last formula
of section~\ref{bounds} can be rewritten as
\begin{multline*}
{
\displaystyle
\Big(\,\frac{1 }
{ \delta_0 }
+
E\big({\tau_0}
\,\big|\,
Z_0= 1\big)
}\Big)
\sum_{i=1}^m 
f\Big(\frac{i}{m}\Big)
\,\nu(i)\,=\,
\displaystyle\sum_{1\leq i\leq m}
f\Big(\frac{i}{m}\Big)
\displaystyle
\frac{1}{\gamma_i}\pi(i-1)
\,.
\end{multline*}
Our goal is to estimate the asymptotic behavior of the right--hand quantity.
\begin{proposition}\label{numest}
Let
$f:[0,1]\to\R$ 
be a continuous non--decreasing function 
such that $f(0)=0$.
Let $a \in ]0,+\infty[$.
We have
$$
\lim_{
\genfrac{}{}{0pt}{1}{\ell,m\to\infty}
{q\to 0,\,
{\ell q} \to a}
}
\,%\sum_{1\leq i\leq m}f\Big(\frac{i}{m}\Big)\nu(i)
\frac{
\displaystyle\sum_{1\leq i\leq m}f\Big(\frac{i}{m}\Big)
\displaystyle
\frac{1}{\gamma_i}\pi(i-1)
}
{\displaystyle 
E(\tau_0\,|\,Z_0=1)}
\,=\,f\big(\rho^*(a)\big)\,.
$$
\end{proposition}
\begin{proof}
%Let us set $\rho^*=\rho(\beta,0)$ and
Throughout the proof, we write simply $\rho^*$ instead of
$\rho^*(a)$.
Let $\eta>0$. 
For $\ell,m$ large enough and $q$ small enough, we have
$$\big|\rho^*-\rho\big(M_H(0,0),M_H(\theta,0)\big)\big|\,<\,\eta\,,$$
whence
\begin{multline*}
\sum_{1\leq i\leq m}
f\Big(\frac{i}{m}\Big)\frac{1}{\gamma_i}\pi(i-1)
\cr
\,=\,
\sum_{
\genfrac{}{}{0pt}{1}
{1\leq i\leq m}
{|i/m-\rho^*|\leq\eta}}
\kern-3pt
f\Big(\frac{i}{m}\Big)\frac{1}{\gamma_i}\pi(i-1)
+\sum_{
\genfrac{}{}{0pt}{1}
{1\leq i\leq m}
{|i/m-\rho^*|>\eta}}
\kern-3pt
f\Big(\frac{i}{m}\Big)\frac{1}{\gamma_i}\pi(i-1)
%\cr
%\hfill
%+\sum_{
%\genfrac{}{}{0pt}{1}
%{1\leq i\leq m}
%{i/m-\rho^*<-\eta}}
%\kern-3pt
%f\Big(\frac{i}{m}\Big)\frac{1}{\gamma_i}\pi(i-1)
%+\sum_{
%\genfrac{}{}{0pt}{1}
%{1\leq i\leq m}
%{i/m-\rho^*>\eta}}
%\kern-3pt
%f\Big(\frac{i}{m}\Big)\frac{1}{\gamma_i}\pi(i-1)\cr
\cr
\,\leq\,
\kern-3pt
\sum_{
\genfrac{}{}{0pt}{1}
{1\leq i\leq m}
{|i/m-\rho^*|\leq\eta}}
\kern-3pt
f(\rho^*+\eta)
\frac{1}{\gamma_i}\pi(i-1)
+\sum_{
\genfrac{}{}{0pt}{1}
{1\leq i\leq m}
{|i/m-\rho^*|>\eta}}
%{i/m-\rho^*<-\eta}}
\kern-3pt
f(1)\frac{1}{\gamma_i}\pi(i-1)
%+\sum_{
%\genfrac{}{}{0pt}{1}
%{1\leq i\leq m}
%{i/m-\rho^*>\eta}}\kern-3pt
%f(1)\frac{1}{\gamma_i}\pi(i-1)
\cr
\,\leq\,
%\kern-3pt
%\sum_{
%\genfrac{}{}{0pt}{1}
%{1\leq i\leq m}
%{|i/m-\rho^*|\leq\eta}}
%\kern-3pt
f(\rho^*+\eta)
E(\tau_0\,|\,Z_0=1)
%\frac{\delta_0}{\gamma_i}\pi(i-1)
+\hfil\cr
\hfill
\frac{m^3 f(1)}{1-M_H(0,0)}\Big(
\pi\big(\lfloor (\rho^*-\eta) m\rfloor\big)
+
\pi\big(\lfloor (\rho^*+\eta) m\rfloor\big)\Big)\,.
\end{multline*}
To obtain 
the last inequality,
we have used the monotonicity properties of $\pi(i)$
and the bounds on~$\gamma_i$ given
at the beginning of the proof of 
corollary~\ref{exta}.
The properties of $\phi$ and the definition of $\rho^*$
imply that
\begin{multline*}
\int_0^{\rho^*}\ln \phi(\exa,0,\rho)\,d\rho\,>\,\cr
\max\bigg(
\int_0^{\rho^*-\eta}\kern-7pt
\ln \phi(\exa,0,\rho)\,d\rho\,,
\int_0^{\rho^*+\eta}\kern-7pt
\ln \phi(\exa,0,\rho)\,d\rho\bigg)\,,
\end{multline*}
so that, using proposition~\ref{bde}, for $m$ large enough,
$$\frac{m^3 f(1)}{1-M_H(0,0)}\Big(
\pi\big(\lfloor (\rho^*-\eta) m\rfloor\big)
+
\pi\big(\lfloor (\rho^*+\eta) m\rfloor\big)\Big)\,\leq\,
\eta 
%f(\rho^*+\eta)
E(\tau_0\,|\,Z_0=1)\,.$$
Adding together the previous inequalities, we arrive at
$$\sum_{1\leq i\leq m}
f\Big(\frac{i}{m}\Big)\frac{1}{\gamma_i}\pi(i-1)
\,\leq\,
\big(f(\rho^*+\eta)+\eta\big)
E(\tau_0\,|\,Z_0=1)\,.$$
Passing to the limit,
we obtain that
$$
\limsup_{
\genfrac{}{}{0pt}{1}{\ell,m\to\infty}
{q\to 0,\,
{\ell q} \to a}
}
\frac{
\displaystyle\sum_{1\leq i\leq m}f\Big(\frac{i}{m}\Big)
\displaystyle
\frac{1}{\gamma_i}\pi(i-1)
}
{\displaystyle 
E(\tau_0\,|\,Z_0=1)}
%\,\sum_{1\leq i\leq m}
%f\Big(\frac{i}{m}\Big)\frac{1}{\gamma_i}\pi(i-1)
\,\leq\,
%\rho^*+2\eta
f(\rho^*+\eta)+\eta
\,.$$
We seek next a complementary lower bound.
If $\sigma\exa\leq 1$, then $\rho^*=0$, and obviously
$$\sum_{1\leq i\leq m}
f\Big(\frac{i}{m}\Big)\frac{1}{\gamma_i}\pi(i-1)
\,\geq\,f(0)\,=\,0\,.$$
If $\sigma\exa> 1$, then $\rho^*>0$ and
$$\int_0^{\rho^*}\ln \phi(\exa,0,\rho)\,d\rho\,>\,
\int_0^{\rho^*-\eta}\kern-7pt
\ln \phi(\exa,0,\rho)\,d\rho\,.$$
%$$
%\lim_{
%\genfrac{}{}{0pt}{1}{\ell,m\to\infty}
%{q\to 0,\,
%{\ell q} \to a}
%}
%\displaystyle 
%\,\frac{1}{m}\ln
%E(\tau_0\,|\,Z_0=1)\,=\,+\infty\,.$$
By corollary~\ref{exta},
for $\ell,m$ large enough and $q$ small enough,
$$\displaylines{
\sum_{
\genfrac{}{}{0pt}{1}
{1\leq i\leq m}
{i/m-\rho^*<-\eta}}\kern-3pt
\frac{1}{\gamma_i}\pi(i-1)
\,\leq\,
\frac{m^3}{1-M_H(0,0)}
\pi\big(\lfloor (\rho^*-\eta) m\rfloor\big)
\,\leq\,\eta
E(\tau_0\,|\,Z_0=1)\,.}$$
%m\pi\big(\lfloor (\rho^*-\eta) m\rfloor\big)
%\,\leq\,\eta
%E(\tau_0\,|\,Z_0=1)\,.}$$
Combining these inequalities, we obtain
\begin{multline*}
\sum_{1\leq i\leq m}
f\Big(\frac{i}{m}\Big)\frac{1}{\gamma_i}\pi(i-1)
\,\geq\,
\sum_{
\genfrac{}{}{0pt}{1}
{1\leq i\leq m}
{i/m-\rho^*\geq-\eta}}
\kern-3pt
f(\rho^*-\eta)
\frac{1}{\gamma_i}\pi(i-1)\cr
\,=\,
f(\rho^*-\eta)\bigg(
\sum_{1\leq i\leq m}
\kern-3pt
\frac{1}{\gamma_i}\pi(i-1)
-
%f(\rho^*-\eta)
\sum_{
\genfrac{}{}{0pt}{1}
{1\leq i\leq m}
{i/m-\rho^*<-\eta}}
\kern-3pt
\frac{1}{\gamma_i}\pi(i-1)\bigg)
\cr
%\,\geq\,
%f(\rho^*-\eta)\Big(
%E(\tau_0\,|\,Z_0=1)
%-
%m\pi\big(\lfloor (\rho^*-\eta) m\rfloor\big)
%\Big)\cr
\,\geq\,
f(\rho^*-\eta)
E(\tau_0\,|\,Z_0=1)
(1-\eta)
%(\rho^*-\eta)+
%\sum_{1\leq i\leq m}
%%\kern-3pt
%(\rho^*-\eta)
%\frac{1}{\gamma_i}\pi(i-1)
%-2\eta
%E(\tau_0\,|\,Z_0=1)\cr
%\,\geq\,
%(\rho^*-3\eta)\Big(1+
%E(\tau_0\,|\,Z_0=1)
%%\sum_{1\leq i\leq m}
%%\frac{\delta_0}{\gamma_i}\pi(i-1)
%\Big)
\,.
\end{multline*}
Passing to the limit,
we obtain that
$$
\liminf_{
\genfrac{}{}{0pt}{1}{\ell,m\to\infty}
{q\to 0,\,
{\ell q} \to a}
}
\frac{
\displaystyle\sum_{1\leq i\leq m}f\Big(\frac{i}{m}\Big)
\displaystyle
\frac{1}{\gamma_i}\pi(i-1)
}
{\displaystyle 
E(\tau_0\,|\,Z_0=1)}
%\,\sum_{1\leq i\leq m}
%f\Big(\frac{i}{m}\Big)\frac{1}{\gamma_i}\pi(i-1)
\,\geq\,
f(\rho^*-\eta)
(1-\eta)\,.
%\rho^*+2\eta
$$
%$$
%\liminf_{
%\genfrac{}{}{0pt}{1}{\ell,m\to\infty}
%{q\to 0,\,
%{\ell q} \to a}
%}
%\,\sum_{1\leq i\leq m}
%f\Big(\frac{i}{m}\Big)\frac{1}{\gamma_i}\pi(i-1)
%\,\geq\,
%f(\rho^*-\eta)
%E(\tau_0\,|\,Z_0=1)
%(1-\eta)
%\,.$$
We finally let $\eta$ go to $0$ in the lower and the upper bounds
to obtain the claim of the proposition.
\end{proof}
\vfill\eject
\section{The neutral phase}\label{disc}
%$\cN$
%the chromosomes distinct from the master sequence $w^*$:
%$$\cN=\Al\setminus\,\{\,w^*\,\}\,.$$ 
%The
%set $\cN$ is thus 
We denote by $\cN$
the set 
of the populations which do not contain the master sequence $w^*$, i.e.,
$$\cN=\Big(\Al\setminus\,\{\,w^*\,\}\Big)^m\,.$$ 
Since we deal with the sharp peak landscape, the transition mechanism
of the process restricted to the set $\cNm$ is neutral. 
We consider a Moran process
$(X_n)_{n\geq 0}$ 
starting from a population of $\cNm$.
%In this section, we study the processes restricted to 
%$$\cN\,=\,\big\{\,x\in\Alm:x(1)\neq w^*,\dots, x(m)\neq w^*\,\big\}\,.$$
We wish to evaluate the first time when a master sequence appears in the
population:
$$\tau_*\,=\,\inf\,\big\{\,n\geq 0: X_n\not\in\cNm\,\big\}\,.$$
We call the time $\tau_*$ the discovery time.
Until the time $\tau_*$, the process evolves in
$\cNm$ and the dynamics of the Moran model in $\cNm$ does not depend on $\sigma$.
In particular, the law
of the discovery time $\tau_*$ is the same for the Moran model with $\sigma>1$
and the neutral Moran model with $\sigma=1$. Therefore,
%To estimate the mean discovery time,
%we need
%only to consider the neutral case where $\sigma=1$.
we compute the estimates for the latter model.
\smallskip

\noindent
{\bf Neutral hypothesis.} Throughout this section, we suppose that $\sigma=1$.
\smallskip

\noindent

\subsection{Ancestral lines}
It is a classical fact that neutral evolutionary processes are much easier to
analyze than evolutionary processes with selection. 
The main reason is that the mutation mechanism and the sampling mechanism
can be decoupled. For instance, it is possible to compute explicitly the
law of a chromosome in the population at time~$n$.

\noindent
Let $\mu_0$ be an exchangeable probability distribution on 
%the population space
$\Alm$. 
Let
$(X_n)_{n\geq 0}$ be the normalized neutral Moran process
with mutation matrix~$M$ and initial law $\mu_0$.
Let $\nu_0$ be the component marginal of $\mu_0$:
$$\forall u\in\Al\qquad\nu_0(u)\,=\,\mu_0\big(\{\,x\in\Alm:x(1)=u\,\}\big)\,.$$
Let $(W_n)_{n\geq 0} 
\index{$W_n$}$
be a Markov chain 
with state space $\Al$, having for
transition matrix the mutation matrix $M$ and with initial law~$\nu_0$.
Let $(\ve_n)_{n\geq 1}$ be a sequence of i.i.d. Bernoulli random variables 
with parameter
${1}/{m}$:
$$\forall n\geq 1\qquad 
P(\ve_n=0)\,=\,1-\frac{1}{m}
\,,\qquad
P(\ve_n=1)\,=\,\frac{1}{m}
\,$$
and let us set
$$
\forall n\geq 1\qquad N(n)\,=\,\ve_1+\cdots+\ve_n\,.
\index{$N(n)$}
$$
We suppose also that the sequence
$(\ve_n)_{n\geq 1}$ and
the Markov chain 
$(W_n)_{n\geq 0}$ are independent.
\begin{proposition}\label{nemarg}
Let $i\in\um$. For any $n\geq 0$, the law of the $i$--th chromosome of
$X_n$ is equal to the law of 
$W_{N(n)}$.
\end{proposition}
\begin{proof}
We start by
computing the transition matrix of the process
$(W_{N(n)})_{n\geq 0}$.
For $u\in\Al$ and $n\geq 0$,
\begin{multline*}
P( W_{N(n+1)}=u)\,=\,
P( W_{N(n+1)}=u,\,\ve_{n+1}=0)\cr
\hfill
+
P( W_{N(n+1)}=u,\,\ve_{n+1}=1)\cr
\,=\,
\Big(1-\frac{1}{m}\Big)
P( W_{N(n)}=u)+
\frac{1}{m}
P( W_{N(n)+1}=u)\,.
\end{multline*}
Moreover
\begin{multline*}
P( W_{N(n)+1}=u)\,=\,
\sum_{v\in\Al}
P(
W_{N(n)+1}=u,\,
W_{N(n)}=v)\cr
\,=\,
\sum_{v\in\Al}
P(
W_{N(n)+1}=u
\,|\,
W_{N(n)}=v)
P(W_{N(n)}=v)\cr
\,=\,
\sum_{v\in\Al} 
P(W_{N(n)}=v) M(v,u)
\,.
\end{multline*}
Therefore the transition matrix of
the process
$(W_{N(n)})_{n\geq 0}$ is
$$
\Big(1-\frac{1}{m}\Big) I +
\frac{1}{m}M\,,$$
where $I$ is the identity matrix.
We do now the proof by induction over $n$. The result holds for $n=0$.
Suppose that it has been proved until time $n$.
Let $i\in\um$. We have, for any $u\in\Al$,
\begin{multline*}
P(X_{n+1}(i)=u)\,=\,
\sum_{x\in\alm}
P(X_{n+1}(i)=u,\,X_n=x)\cr
\,=\,
\sum_{x\in\alm}
P(X_{n+1}(i)=u\,|\,X_n=x)
P(X_n=x)\,.
\end{multline*}
Yet we have
$$
P(X_{n+1}(i)=u\,|\,X_n=x)
\,=\,
\Big(1-\frac{1}{m}\Big)
1_{x(i)=u}+
\frac{1}{m^2}
\sum_{1\leq j\leq m}
M(x(j),u)\,.
$$
Thus
\begin{multline*}
P(X_{n+1}(i)=u)\,=\,
\sum_{x\in\alm}
\Big(1-\frac{1}{m}\Big)
1_{x(i)=u}
P(X_n=x)\cr
\hfill+
\sum_{x\in\alm}
\frac{1}{m^2}
\sum_{1\leq j\leq m}
M(x(j),u)
P(X_n=x)\cr
\,=\,\Big(1-\frac{1}{m}\Big)
P(X_n(i)=u)+
\sum_{v\in\Al}
\frac{1}{m^2}
\sum_{1\leq j\leq m}
M(v,u)
P(X_n(j)=v)
\,.
\end{multline*}
By the induction hypothesis, 
%$P(X_n(j)=v)\,=\,
%P(
%W_{N(n)}
%=v)$ for 
%${v\in\Al}$ and $j\in\um$, 
$$
\forall {v\in\Al}\quad\forall
j\in\um\qquad
P(X_n(j)=v)\,=\,
P(
W_{N(n)}
=v)\,,$$
whence
\begin{multline*}
P(X_{n+1}(i)=u)\,=\,\cr
\Big(1-\frac{1}{m}\Big)
P(
W_{N(n)}
=u)+
\frac{1}{m}
\sum_{v\in\Al}
P(
W_{N(n)}
=v)
M(v,u)
\cr
\,=\,
P(
W_{N(n+1)}
=u)
\,.
\end{multline*}
The result still holds at time $n+1$.
\end{proof}
%In fact
%law of a chromosome in the population at time~$n$.

\noindent
We perform next a similar computation to obtain
the law of an ancestral line. Let us first define an ancestral line.
For $i\in\um$ and $n\geq 1$, we denote by
$\cI(i,n,n-1) 
\index{$\cI(i,n,n-1)$}$
the index of the ancestor at time $n-1$ of the $i$--th chromosome
at time~$n$. Let us explicit its value. If $X_{n-1}=x$ and
$X_n=y$ with $y=x(j\leftarrow u)$, where the chromosome $u$ 
has been obtained by replicating the $k$--th chromosome of $x$, then
$$\cI(i,n,n-1)\,=\,
\begin{cases}
i &\text{if $i\neq j$}\\
k &\text{if $i=j$}\\
\end{cases}
$$
For $s\leq n$,
the index
$\cI(i,n,s)$ of the ancestor at time $s$ of the $i$--th chromosome
at time~$n$ is then defined recursively with the help of the following
formula:
$$
\cI(i,n,s)\,=\,
\cI(
\cI(i,n,n-1)
,n-1,s)\,.$$
The ancestor at time $s$ of the $i$--th chromosome
at time~$n$ is the chromosome 
$$
\anc(i,n,s)\,=\,
X_s(\cI(i,n,s))\,.
\index{$\anc(i,n,s)$}
$$
The ancestral line
of the $i$--th chromosome 
at time~$n$ is the sequence of its ancestors until time $0$,
$$(\anc(i,n,s),\,0\leq s\leq n)\,=\,
(X_s(\cI(i,n,s)),\,0\leq s\leq n)\,.$$
\begin{proposition}\label{neal}
Let $i\in\um$. For any $n\geq 0$, the law of the ancestral line
$(\anc(i,n,s),\,0\leq s\leq n)$
of the $i$--th chromosome of $X_n$
is equal to the law of 
$(W_{N(0)},\dots,W_{N(n)})$.
\end{proposition}
\begin{proof}
We do the proof by induction over $n$. The result is true at rank $n=0$. Suppose
it has been proved until time $n$.
Let  $i\in\um$ and let
$u_0,\dots,u_{n+1}\in\Al$. We compute
\begin{multline*}
P\big(\anc(i,n+1,s)=u_s,\,0\leq s\leq n+1\big)
\cr
%\,=\,
%\sum_{x\in\alm}
%\sum_{1\leq j\leq m}
%P\bigg(
%\begin{matrix}
%\anc(i,n+1,s)=u_s,\,0\leq s\leq n+1\cr
%X_n=x,\,\cI(i,n+1,n)=j
%\end{matrix}
%\bigg)\cr
\,=\,
\sum_{x\in\alm}
\sum_{1\leq j\leq m}
P\bigg(
\begin{matrix}
%\anc(i,n+1,n+1)=u_{s+1}, 
X_{n+1}(i)=u_{n+1},
\,\cI(i,n+1,n)=j\\
X_n=x,\,
\anc(j,n,s)=u_s,\,0\leq s\leq n
\end{matrix}
\bigg)\cr
\,=\,
\sum_{x\in\alm}
\sum_{1\leq j\leq m}
P\bigg(\,
\begin{matrix}
X_{n+1}(i)=u_{n+1}\\
%\anc(i,n+1,n+1)=u_{s+1}\\
\cI(i,n+1,n)=j\\
\end{matrix}
\,\Big|\,
\begin{matrix}
\anc(j,n,s)=u_s\\
0\leq s\leq n,\,X_n=x
\end{matrix}
\bigg)\cr
\hfill\times
P\Big(\begin{matrix}
\anc(j,n,s)=u_s\\
0\leq s\leq n,\,X_n=x
\end{matrix}
\Big)\,.
\end{multline*}
Since we deal with the neutral process, we have
\begin{multline*}
%\sum_{1\leq j\leq m}
P\bigg(\,
\begin{matrix}
X_{n+1}(i)=u_{n+1}\\
%\anc(i,n+1,n+1)=u_{s+1}\\
\cI(i,n+1,n)=j\\
\end{matrix}
\,\Big|\,
\begin{matrix}
\anc(j,n,s)=u_s\\
0\leq s\leq n,\,X_n=x
\end{matrix}
\bigg)\cr
\,=\,
%\sum_{1\leq j\leq m}
P\bigg(\,
\begin{matrix}
X_{n+1}(i)=u_{n+1}\\
%\anc(i,n+1,n+1)=u_{s+1}\\
\cI(i,n+1,n)=j\\
\end{matrix}
\,\Big|\,
%X_n(j)=u_n
X_n=x
\bigg)
\cr 
\,=\,
\begin{cases}
\displaystyle\Big(1-\frac{1}{m}\Big) 1_{x(i)=u_{n+1}} +
\frac{1}{m^2}M(x(i),u_{n+1}) 
 &\text{if $j=i$}\\
\displaystyle\frac{1}{m^2}M(x(j),u_{n+1})
 &\text{if $j\neq i$}\\
\end{cases}
%\,=\,
%\Big(1-\frac{1}{m}\Big) 1_{u_n=u_{n+1}} +
%\frac{1}{m}M(u_n,u_{n+1})\cr
%\,=\,
%P\big(W_{N(n+1)}=u_{n+1}\,|\,W_{N(n)}=u_n\big)\,.
\end{multline*}
Reporting in the previous equality, we get
\begin{multline*}
P\big(\anc(i,n+1,s)=u_s,\,0\leq s\leq n+1\big)
\,=\,\cr
\sum_{x\in\alm}
\displaystyle\Big(1-\frac{1}{m}\Big) 1_{x(i)=u_{n+1}} 
%\frac{1}{m^2}M(x(i),u_{n+1}) 
\,
P\Big(\begin{matrix}
\anc(i,n,s)=u_s\\
0\leq s\leq n,\,X_n=x
\end{matrix}
\Big)\,\hfill
\cr
\hfill+\sum_{x\in\alm}
\sum_{1\leq j\leq m}
\frac{1}{m^2}M(x(j),u_{n+1})
\,
P\Big(\begin{matrix}
\anc(j,n,s)=u_s\\
0\leq s\leq n,\,X_n=x
\end{matrix}
\Big)
\cr
\,=\,
\displaystyle\Big(1-\frac{1}{m}\Big) 1_{u_n=u_{n+1}} 
P\big(
\anc(i,n,s)=u_s,\,
0\leq s\leq n
\big)\,\hfill
\cr\hfill
+
\sum_{1\leq j\leq m}
\frac{1}{m^2}M(u_n,u_{n+1})
\,
P\big(
\anc(j,n,s)=u_s,\,
0\leq s\leq n
\big)
\,.
\end{multline*}
By the induction hypothesis, we have,
for any $j\in\{\,1,\cdots,m\,\}$,
$$
P\big(
\anc(i,n,s)=u_s,\,
0\leq s\leq n
\big)
\,=\,
P\big(W_{N(0)}=u_0,\dots,W_{N(n)}=u_n\big)\,.$$
Therefore
\begin{multline*}
P\big(\anc(i,n+1,s)=u_s,\,0\leq s\leq n+1\big)
\,=\,\cr
P\big(W_{N(n+1)}=u_{n+1}\,|\,W_{N(n)}=u_n\big)
P\big(W_{N(0)}=u_0,\dots,W_{N(n)}=u_n\big)
\cr
\,=\,
%P\big(W_{N(n+1)}=u_{n+1}\,|\,W_{N(n)}=u_n\big)
P\big(W_{N(0)}=u_0,\dots,W_{N(n+1)}=u_{n+1}\big)
\end{multline*}
and the induction step is completed.
\end{proof}
\subsection{Mutation dynamics}\label{mutdyn}
Throughout the section, we consider a Markov chain 
$(Y_n)_{n\geq 0} 
\index{$Y_n$}$
with state space $\zl$ and having for
transition matrix the lumped mutation matrix $M_H$.
By lemma~\ref{mhlump}, 
for $b,c\in\zl$,
the coefficient $M_H(b,c)$ of the matrix $M_H$ is equal to
$$
\sum_{
\genfrac{}{}{0pt}{1}{0\leq k\leq\ell-b}{
\genfrac{}{}{0pt}{1}
 {0\leq h\leq b}{k-h=c-b}
}
}
{ \binom{\ell-b}{k}}
{\binom{b}{h}}
\Big(p\Big(1-\frac{1}{\kappa}\Big)\Big)^k
\Big(1-p\Big(1-\frac{1}{\kappa}\Big)\Big)^{\ell-b-k}
\Big(\frac{p}{\kappa}\Big)^h
\Big(1-\frac{p}{\kappa}\Big)^{b-h}\,.
$$
Such a Markov chain can be realized on our common probability
space. Its construction requires only 
the family 
of random variables 
$$(U_{n,l},\,n\geq 1,\,1\leq l\leq\ell)$$
with uniform law on 
the interval 
$[0,1]$.
Let $b\in\zl$ be the starting point of the chain.
We set $Y_0=b$ and we define inductively for $n\geq 1$
\begin{align*}
Y_{n}\,&=\,Y_{n-1}
-\sum_{k=1}^{Y_{n-1}}1_{U_{n,k}<p/\kappa}
+\sum_{k=Y_{n-1}+1}^\ell1_{U_{n,k}>1-p(1-1/\kappa)}\cr
\,&=\,\cMH(Y_{n-1},
U_{n,1},\dots,
U_{n,\ell})
\,.
\end{align*}
By lemma~\ref{monmut}, the map~$\cMH$ is non--decreasing with respect
to its first argument.
Thus the above
construction provides a monotone coupling of the processes starting
with different initial conditions and we conclude that the 
Markov chain
$(Y_n)_{n\geq 0}$ is monotone.
\begin{proposition}\label{ivy}
The matrix $M_H$ is reversible with respect to 
the binomial law 
$\cB(\ell, 1-1/\kappa)$
with parameters 
$\ell$ and $1-1/\kappa$.
This binomial law is the invariant probability measure of the 
Markov chain
$(Y_n)_{n\geq 0}$.
\end{proposition}
{\bf Notation}. 
%To alleviate the notation, 
We 
denote simply by $\cB\index{$\cB$}$
the binomial law 
$\cB(\ell, 1-1/\kappa)$. Thus
%For $b\in\zl$, 
$$
\forall b\in\zl \qquad
\cB(b)\,=\,
\binom{\ell}{b}
\Big(1-\frac{1}{\kappa}\Big)^b
\Big(\frac{1}{\kappa}\Big)^{\ell-b}
\,.$$
%$$\cB\,=\,\cB(\ell, 1-1/\kappa) \,.$$
\begin{proof}
We check that the matrix $M_H$ is reversible with respect to $\cB$.
Let $b,c\in\zl$. We use the identity
$$\binom{\ell}{b}
{ \binom{\ell-b}{k}}
{\binom{b}{h}}
\,=\,\frac{\ell!}{k!\,h!\,(\ell-b-k)!\,(b-h)!}$$
to write
\begin{multline*}
\cB(b)\,M_H(b,c)\,=\,
\binom{\ell}{b}
\Big(1-\frac{1}{\kappa}\Big)^b
\Big(\frac{1}{\kappa}\Big)^{\ell-b}\,\times\cr
\sum_{
\genfrac{}{}{0pt}{1}{0\leq k\leq\ell-b}{
\genfrac{}{}{0pt}{1}
 {0\leq h\leq b}{k-h=c-b}
}
}\kern-3pt
{ \binom{\ell-b}{k}}
{\binom{b}{h}}
\Big(p\Big(1-\frac{1}{\kappa}\Big)\Big)^k
\Big(1-p\Big(1-\frac{1}{\kappa}\Big)\Big)^{\ell-b-k}
\Big(\frac{p}{\kappa}\Big)^h
\Big(1-\frac{p}{\kappa}\Big)^{b-h}\cr
\,=\,
\kern-5pt
\sum_{
\genfrac{}{}{0pt}{1}{0\leq k\leq\ell-b}{
\genfrac{}{}{0pt}{1}
 {0\leq h\leq b}{k-h=c-b}
}
}
\kern-5pt
\frac{
\displaystyle
\ell!\, p^{k+h}
\Big(1-\frac{1}{\kappa}\Big)^{b+k}
\Big(\frac{1}{\kappa}\Big)^{\ell-b+h}
}{k!\,h!\,(\ell-b-k)!\,(b-h)!}
\Big(1-p\Big(1-\frac{1}{\kappa}\Big)\Big)^{\ell-b-k}
\Big(1-\frac{p}{\kappa}\Big)^{b-h}\kern-5pt\,.
\end{multline*}
We eliminate the variable $h=k+b-c$ in this formula:
\begin{multline*}
\cB(b)\,M_H(b,c)\,=\,
\sum_{
\genfrac{}{}{0pt}{1}{0\leq k\leq\ell-b}
{c-b\leq k\leq c}
}
\frac{\ell!
}{k!\,(k+b-c)!\,(\ell-b-k)!\,(c-k)!}
\,\times\cr
\, p^{2k+b-c}
\Big(1-\frac{1}{\kappa}\Big)^{b+k}
\Big(\frac{1}{\kappa}\Big)^{\ell-c+k}
\Big(1-p\Big(1-\frac{1}{\kappa}\Big)\Big)^{\ell-b-k}
\Big(1-\frac{p}{\kappa}\Big)^{c-k}\,.
\end{multline*}
If we set now $h=k+b-c$ and we eliminate $k$, we get
\begin{multline*}
\cB(b)\,M_H(b,c)\,=\,
\sum_{
\genfrac{}{}{0pt}{1}
{b-c\leq h\leq\ell-c}
{0\leq h\leq b}
}
\frac{\ell!
}{(h+c-b)!\,h!\,(\ell-c-h)!\,(b-h)!}
\,\times\cr
\, p^{2h+c-b}
\Big(1-\frac{1}{\kappa}\Big)^{c+h}
\Big(\frac{1}{\kappa}\Big)^{\ell-b+h}
\Big(1-p\Big(1-\frac{1}{\kappa}\Big)\Big)^{\ell-c-h}
\Big(1-\frac{p}{\kappa}\Big)^{b-h}\cr
\,=\,
\cB(c)\,M_H(c,b)
\,.
\end{multline*}
We obtain the same expression as before, but with $b$ and $c$ exchanged.
Thus the matrix $M_H$ is reversible with respect to $\cB$ and 
$\cB$ is the invariant probability measure of $\cMH$.
\end{proof}

\noindent
When $\ell$ grows, the law $\cB$ concentrates exponentially fast
in a neighborhood of its mean
$$\lk\,=\,\sum_{l=0}^\ell l\,\cB(l)\,=\,
\ell
(1-1/\kappa)\,.$$
We estimate next the probability of the points at the left of $\lk$.
\begin{lemma}\label{exco}
%For $b\leq\lk$, we have
For $b\leq \ell/2$, we have
$$\frac{1}{\kappa^\ell}
\left(\frac{\ell}{2b}\right)^b\,\leq\,
\cB(b)\,\leq\,\frac{\ell^b}{\kappa^{\ell-b}}\,.$$
\end{lemma}
\begin{proof} Let $b\leq\ell/2$. Then
%\begin{multline*}
$$\cB(b)\,=\,
\binom{\ell}{b}
\Big(1-\frac{1}{\kappa}\Big)^b
\Big(\frac{1}{\kappa}\Big)^{\ell-b}
\,\geq\,
\binom{\ell}{b}
\frac{1}{\kappa^\ell}
\,\geq\,
\frac{(\ell -b)^b}{b^b}
\frac{1}{\kappa^\ell}
\,\geq\,
\left(\frac{\ell}{2b}\right)^b
\frac{1}{\kappa^\ell}\,.
$$
The upper bound on
$\cB(b)$ is straightforward.
\end{proof}

\noindent
The estimates of lemma~\ref{exco} can be considerably enhanced.
In the next lemma, we present the fundamental large deviation
estimates for the binomial distribution.
This is the simplest case of the famous Cram\'er theorem.
\begin{lemma}\label{gdb}
For $\rho\in[0,1]$, we have
$$\lim_{\ell\to\infty}
\,\frac{1}{\ell}\ln\cB(\lfloor \rho \ell\rfloor)\,=\,
-(1-\rho)\ln\big(\kappa(1-\rho)\big)
-\rho\ln
\frac{\kappa\rho}{\kappa-1}\,.
%\Big(\big(1-\frac{1}{\kappa}\big)\rho\Big)\,.
$$
\end{lemma}
\begin{proof} 
We write
\begin{multline*}
\ln\cB(\lfloor \rho \ell\rfloor)\,=\,
\ln\frac{\ell\cdots(\ell-\lfloor \rho\ell\rfloor +1)}
{1\cdots
\lfloor \rho \ell\rfloor}
+
(\ell-\lfloor \rho \ell\rfloor)\ln\frac{1}{\kappa}
+\lfloor \rho \ell\rfloor\ln\big(1-\frac{1}{\kappa}\big)
\cr
\,=\,
\sum_{k=0}^{\lfloor \rho \ell\rfloor-1}
\ln\Big(1-\frac{k}{\ell}\Big)
-
\sum_{k=1}^{\lfloor \rho \ell\rfloor}
\ln\frac{k}{\ell}
+
(\ell-\lfloor \rho \ell\rfloor)\ln\frac{1}{\kappa}
+\lfloor \rho \ell\rfloor\ln\big(1-\frac{1}{\kappa}\big)\,.\hfil
\end{multline*}
We recognize Riemann sums for the functions $\ln(1-x)$ and $\ln x$, thus
$$\lim_{\ell\to\infty}
\,\frac{1}{\ell}\ln\cB(\lfloor \rho \ell\rfloor)\,=\,
\int_0^\rho\ln\frac{1-x}{x}\,dx
+(1-\rho)\ln\frac{1}{\kappa}
+\rho \ln\big(1-\frac{1}{\kappa}\big)\,.$$
We conclude by performing the integration.
\end{proof}

\noindent
The minimum of the rate function appearing in lemma~\ref{gdb}
is $\lk$.
The typical behavior of the Markov chain 
$(Y_n)_{n\geq 0}$ is the following.
Starting from $1$, it very quickly reaches a neighbor of
its stable equilibrium
$\lk$.
%$\ell (1-1/\kappa)$. 
Then it starts exploring the surrounding space by performing
larger and larger excursions outside
$\lk$.
Starting from $\lk$,
the time needed to hit the point $c\in\zl$ is of order $\cB(c)^{-1}$.
Once the process is close to $\lk$,
it is unlikely to visit $0$ before time
$\cB(0)^{-1}\,=\,\kappa^\ell$.
This is why 
the expected value of the hitting time of $0$ starting from $1$ is of
order 
$\kappa^\ell$.
In the next sections, we derive quantitative bounds on the behavior of the
chain
$(Y_n)_{n\geq 0}$, starting from $1$ or from $\ell$.
We need only crude bounds, hence we use elementary techniques, namely, we compare
the process with a sum of i.i.d. random variables and we use 
the classical Chebyshev inequality, as well as the
exponential Chebyshev inequality.
The resulting proofs are somehow clumsy, and better estimates
could certainly be derived with more sophisticated tools.
\subsection{Falling to equilibrium from the left}
For $b\in\zl$, we define
the hitting time $\tau(b)$ of $\{\,b,\dots,\ell\,\}$ by
$$\tau(b)\,=\,\inf\,\big\{\,n\geq 0: Y_n\geq b\,\big\}\,.
\index{$\tau(b)$}
$$
Our first goal is to estimate, for $b$ smaller than 
$\lk=\ell(1-1/\kappa)$ and $n\geq 1$, the
probability
$$P\big(\tau(b)>n\,|\,Y_0=0\big)\,.$$
%\medskip

\noindent
{\bf Rough bound on the drift.}
Suppose that
$\tau(b)> n$. Then $Y_{n-1}<b$ and
$$Y_{n}\,\geq\,Y_{n-1}
-\sum_{k=1}^{b}1_{U_{n,k}<p/\kappa}
+\sum_{k=b+1}^\ell1_{U_{n,k}>1-p(1-1/\kappa)}\,.$$
Iterating this inequality, we see that, on the event
$\{\,\tau(b)> n\,\}$, we have
$Y_n\geq V_n$ where
$$V_n\,=\,
\sum_{t=1}^n\bigg(
-\sum_{k=1}^{b}1_{U_{t,k}<p/\kappa}
+\sum_{k=b+1}^\ell1_{U_{t,k}>1-p(1-1/\kappa)}\bigg)\,.$$
Therefore
$$P\big(\tau(b)>n\,|\,Y_0=0\big)\,\leq\,
P(V_n<b)\,.$$
We shall bound
$P(V_n<b)$ with the help of Chebyshev's inequality. Let us compute
the mean and the variance of $V_n$. Since $V_n$ is a sum of independent
Bernoulli random variables, we have
\begin{align*}
E(V_n)\,&=\,
n\big(-b\,\frac{p}{\kappa}+(\ell-b)\,p\,(1-\frac{1}{\kappa})\big)
\,=\,
np\,\big(\lK-b\big)\,,\cr
\var(V_n)\,&=\,
n\Big(b\,
\frac{p}{\kappa}
\big(1-\frac{p}{\kappa}\big)
+(\ell-b)\,
p\,\big(1-\frac{1}{\kappa}\big)
\big(1-p\,\big(1-\frac{1}{\kappa}\big)\big)
\Big)\cr
\,&\leq\,n\big(bp+(\ell-b)p\big)\,=\,n\ell p\,.
\end{align*}
We suppose that $n$ is large enough so that
$2b\,<\,E(V_n)$,
%$$b\,<\,E(V_n)\,=\,
%np\,\big(\lK-b\big)\,,$$
that is,
$$n\,>\,\frac{2b}{
p\,\big(\lK-b\big)}\,.$$
By Chebyshev's inequality, we have then
\begin{multline*}
P(V_n<b)\,=\,
P\big(V_n-E(V_n)<b-E(V_n)\big)\cr
\,\leq\,
P\big(\big|V_n-E(V_n)\big|>\frac{1}{2}E(V_n)\big)\cr
\,\leq\,
\frac{4\var(V_n)}{
\big(E(V_n)\big)^2}
\,\leq\,
\frac{4n\ell p}{
\Big(np\,\big(\lK-b\big)\Big)^2
}\,.
\end{multline*}
We have thus proved the following estimate.
\begin{lemma}\label{bdhit}
For $n$ such that
$$n\,>\,\frac{2b}{
p\,\big(\lK-b\big)}\,,$$
we have
$$P\big(\tau(b)>n\,|\,Y_0=0\big)\,\leq\,
\frac{4\ell }{
np\,\big(\lK-b\big)^2
}\,.$$
\end{lemma}
%We use this inequality to get a crude lower bound on the 
We derive next a crude lower bound on the 
descent from $0$ to $\lk$. This lower bound will be used to derive
the upper bound on the discovery time.
%We next specialize this inequality to get estimates relevant
%to our regime.
%As in theorem~\ref{mainth},
%we suppose that 
%$\ell$ goes to $\infty$, $m$ goes to $\infty$ and $q$ goes to $0$ 
%$$\ell\to +\infty\,,\qquad m\to +\infty\,,\qquad q\to 0\,,$$
%$$\ell\to +\infty\,,\qquad q\to 0\,,$$
%in such a way that
%$${\ell q} \to a\in ]0,+\infty[\,.$$
%$${\ell q} \to a\in [0,+\infty]\,,
%\qquad\frac{m}{\ell}\to\alpha\in [0,+\infty]\,.$$
%We next specialize this inequality to get estimates relevant
%to our regime.
%As in theorem~\ref{mainth},
\begin{proposition}\label{lowerdesc}
We suppose that 
$\ell\to +\infty$, $q\to 0$
in such a way that
$${\ell q} \to a\in ]0,+\infty[\,.$$
For $\ell$ large enough and $q$ small enough, 
we have
$$P\big(\tau(\lk)\leq\ell^2\,|\,Y_0=0\big)\,\geq\,
\Big(1-\frac{5}{a(\ln\ell)^2}\Big)
\Big(\frac{p}{\kappa}\Big)^{\ln\ell}
e^{-2a}\,.
$$
\end{proposition}
\begin{proof}
%We apply lemma~\ref{bdhit} with
%$n=\ell^2$ and $b=\lk-\ln\ell$. 
%For $\ell$ large enough, 
%$$\ell^2>\frac{2\lk}{p\ln\ell}\,,
%$$
%whence
%$$P\big(\tau(
%\lk-\ln\ell
%)> \ell^2
%\,|\,Y_0=0\big)\,\leq\,
%%P\Big(\ell+V_{\ell^2}>
%%\ln\ell+\lk
%%\Big)\,\leq\,
%\frac{5}
%{a( \ln\ell)^2}\,.
%$$
%Moreover, we have
%$$P\big(Y_1\leq\lk\,|\,Y_0=\ln\ell+\lk)\,\geq\,
%$$
%%P\Big(\ell+V_{\ell^2}>
%%\ln\ell+\lk
%%\Big)\,\leq\,
%Therefore
We decompose
\begin{multline*}
P\big(\tau(\lk)\leq\ell^2\,|\,Y_0=0\big)\,\geq\,
P\big(\tau(\lk-\ln\ell)<\ell^2,\,
\tau(\lk)\leq\ell^2\,|\,Y_0=0\big)\cr
\,=\,
\sum_{t<\ell^2}\sum_{b\geq\lk-\ln\ell}
P\big(\tau(\lk-\ln\ell)=t,\,Y_t=b,\,
\tau(\lk)\leq\ell^2\,|\,Y_0=0\big)\cr
\,=\,
\sum_{t<\ell^2}\sum_{b\geq\lk-\ln\ell}
P\big(
\tau(\lk)\leq\ell^2\,|\,
\tau(\lk-\ln\ell)=t,\,
Y_t=b,\,
Y_0=0\big)
\cr
\hfill\times
P\big(\tau(\lk-\ln\ell)=t,\,Y_t=b \,|\,Y_0=0\big)\,.
\end{multline*}
By the Markov property and the monotonicity of the process
$(Y_n)_{n\geq 0}$,
we have,
for $t<\ell^2$ and $b\geq \lk-\ln\ell$,
\begin{multline*}
P\big(
\tau(\lk)\leq\ell^2\,|\,
\tau(\lk-\ln\ell)=t,\,
Y_t=b,\,
Y_0=0\big)
\cr
\,\geq\,
P\big(
\tau(\lk)\leq\ell^2-t\,|\,Y_0=b\big)
\,\geq\,
P\big(
\tau(\lk)\leq\ell^2-t\,|\,Y_0=\lk-\ln\ell\big)
\cr
\,\geq\,
P\big(
Y_1=\lk\,|\,Y_0=\lk-\ln\ell\big)
\,=\,M_H
(\lk-\ln\ell,\lk)\,.
\end{multline*}
Reporting this inequality in the previous sum, we get
$$
P\big(\tau(\lk)\leq\ell^2\,|\,Y_0=0\big)\,\geq\,
P\big(\tau(\lk-\ln\ell)<\ell^2
\,|\,Y_0=0\big)\,
M_H (\lk-\ln\ell,\lk)\,.
$$
By lemma~\ref{bdhit} applied with $b=\lk-\ln\ell$ and $n=\ell^2-1$,
we have 
for $\ell$ large enough and $q$ small enough 
$$P\big(\tau(\lk-\ln\ell)\geq\ell^2
\,|\,Y_0=0\big)\,\leq\,\frac{5}{a(\ln\ell)^2}\,.$$
Moreover,
for $\ell$ large enough and $q$ small enough, 
$$M_H (\lk-\ln\ell,\lk)\,\geq\,
\Big(\frac{p}{\kappa}\Big)^{\ln\ell}(1-q)^\ell\,\geq\,
\Big(\frac{p}{\kappa}\Big)^{\ln\ell}e^{-2a}\,.
$$
Putting the previous inequalities together, we obtain the
desired lower bound.
\end{proof}

\noindent
We will need more information in order to derive the lower bound on the discovery time.
We wish to control the time and speed at which
the Markov chain 
$(Y_n)_{n\geq 0}$, starting from $1$, reaches a neighborhood of its equilibrium $\lk$
without visiting $0$.
This will require a stronger inequality than the
one stated in lemma~\ref{bdhit}, this is the purpose of next lemma.
\begin{lemma}\label{tbdhit}
For $n\geq 1$, 
$b\in\zl$ and $\lambda>0$, we have
%$$n\,>\,\frac{2b}{
%p\,\big(\lK-b\big)}\,,$$
%we have
$$P\big(\tau(b)>n\,|\,Y_0=0\big)\,\leq\,
\exp\Big(\lambda b+nb\frac{p}{\kappa}(e^\lambda-1)
+n(\ell-b)p\frac{\kappa-1}{\kappa}(e^{-\lambda}-1)\Big)
\,.$$
\end{lemma}
\begin{proof}
We obtain this inequality as a consequence of Tchebytcheff's exponential
inequality. Indeed, we have
\begin{multline*}
P\big(\tau(b)>n\,|\,Y_0=0\big)\,\leq\,
P(V_n<b)\cr
\,=\,
P(-\lambda V_n>-\lambda b)\,=\,
P\big(\exp(-\lambda V_n)>\exp(-\lambda b)\big)\cr
\,\leq\,
\exp(\lambda b)E\big(\exp(-\lambda V_n)\big)
\,=\,
\exp(\lambda b)\Big(E\big(\exp(-\lambda V_1)\big)\Big)^n
\,.
\end{multline*}
Yet
\begin{multline*}
E\big(\exp(-\lambda V_1)\big)\,=\,
E\Big(\exp\Big(
\lambda\sum_{k=1}^{b}1_{U_{1,k}<p/\kappa}
-\lambda\sum_{k=b+1}^\ell1_{U_{1,k}>1-p(1-1/\kappa)}\Big)\Big)
\cr
\,=\,
\Big(1+\frac{p}{\kappa}(e^\lambda-1)\Big)^b
\Big(1+p\frac{\kappa-1}{\kappa}(e^{-\lambda}-1)\Big)^{\ell-b}
\,.$$
\end{multline*}
Thus
\begin{multline*}
P\big(\tau(b)>n\,|\,Y_0=0\big)\,\leq\,
%\,\leq\,
\cr\exp\Big(
\lambda b+
nb\ln\big(1+\frac{p}{\kappa}(e^\lambda-1)\big)
+n(\ell -b)\ln
\big(1+p\frac{\kappa-1}{\kappa}(e^{-\lambda}-1)\big)
\Big)
\,.
\end{multline*}
Using the inequality $\ln(1+t)\leq t$, we obtain the
desired result.
\end{proof}

\noindent
We derive next
two kinds of estimates: first for the start of the fall, and second
for the completion of the fall. 

\medskip

\noindent
{\bf Start of the fall.}
We show here that, after a time $\sqrt{\ell}$, the Markov chain 
$(Y_n)_{n\geq 0}$ is with high probability in the interval
$[\ln\ell, \ell]$.
%by taking $n$ and $p$ as functions of $\ell$.
\begin{proposition}\label{firstdesc}
We suppose that 
%$\ell$ goes to $+\infty$, $q$ goes to $0$ and
%$q\ell$ converges towards $a\in [0,+\infty]$.
%$$\ell\to +\infty\,,\qquad q\to 0\,,\qquad
%{\ell q} \to a\in [0,+\infty]\,.$$
$\ell\to +\infty\,, q\to 0\,,
{\ell q} \to a\in ]0,+\infty[$.
For $\ell$ large enough and $q$ small enough, we have
$$%\forall \ell\geq \ell_0\quad
\forall t\geq\sqrt{\ell}\qquad
P\big(Y_{t}\geq\ln\ell
\,|\,Y_0=0\big)\,\geq\,
1-
%\frac{6 }{a \sqrt{\ell}}
\exp\Big(-\frac{1}{2}(\ln\ell)^2\Big)
\,.$$
\end{proposition}
\begin{proof}
We write, for $t\geq\sqrt{\ell}$,
\begin{multline*}
P\big(Y_{t}\geq\ln\ell
\,|\,Y_0=0\big)\,\geq\,
P\big(Y_{t}\geq\ln\ell,\,
\tau({2\ln\ell})\leq\sqrt{\ell}\,|\,Y_0=0\big)\cr
\,=\,
\sum_{n=1}^{\sqrt{\ell}}
\sum_{k=2\ln\ell}^{{\ell}}
P\big(Y_{t}\geq\ln\ell\,,
\tau({2\ln\ell})=n,\,Y_n=k\,|\,Y_0=0\big)\cr
\,=\,
\sum_{n=1}^{\sqrt{\ell}}
\sum_{k=2\ln\ell}^{{\ell}}
P\big(Y_{t}\geq\ln\ell\,|\,
\tau({2\ln\ell})=n,\,Y_n=k,\,Y_0=0\big)\times\cr
\hfill
P\big(
\tau({2\ln\ell})=n,\,Y_n=k\,|\,Y_0=0\big)\,.
\end{multline*}
Now, for $n\leq \sqrt{\ell}$ and $k\geq 2\ln\ell$,
by the Markov property, and thanks to the
monotonicity of 
the process $(Y_t)_{t\geq 0}$,
\begin{multline*}
P\big(Y_{t}\geq\ln\ell\,|\,
\tau({2\ln\ell})=n,\,
Y_n=k,\,
Y_0=0\big)
\cr
\,=\,
P\big(Y_{t}\geq\ln\ell\,|\,
Y_n=k
%2\ln\ell,\,
%\tau({2\ln\ell})=n,\,Y_0=0
\big)
\,\geq\,
P\big(Y_{t}\geq\ln\ell\,|\,
Y_n={2\ln\ell}\big)\cr
\,=\,
P\big(Y_{t-n}\geq\ln\ell\,|\,
Y_0={2\ln\ell}\big)\,.
\end{multline*}
For $b<\ln\ell$, we have 
by lemmas~\ref{majinv} and~\ref{exco},
$$P\big(Y_{t-n}=b\,|\,
Y_0={2\ln\ell}\big)
\,\leq\,
\frac{\cB(b)}
{\cB(2\ln\ell)}
\,\leq\,
\frac{\cB(\ln\ell)}
{\cB(2\ln\ell)}
\,\leq\,
\left(
\frac{(4\kappa\ln\ell)^{2}}{\ell}
\right)^{\ln\ell}\,,$$
whence
$$P\big(Y_{t-n}\geq\ln\ell\,|\,
Y_0={2\ln\ell}\big)\,\geq\,
1-\ln\ell\left(
\frac{(4\kappa\ln\ell)^{2}}{\ell}
\right)^{\ln\ell}\,.$$
%\frac{(2\kappa\ln\ell)^{2\ln\ell}}
%{\ell^{\ln\ell}}
Reporting this inequality in the previous sum, we get
\begin{multline*}
P\big(Y_{t}\geq\ln\ell
\,|\,Y_0=0\big)\,\geq\,\cr
\bigg(1-\ln\ell\left(
\frac{(4\kappa\ln\ell)^{2}}{\ell}
\right)^{\ln\ell}\bigg)
P\big(
\tau({2\ln\ell})\leq\sqrt{\ell}\,|\,Y_0=0\big)\,.
\end{multline*}
By lemma~\ref{tbdhit} applied with $\lambda=\ln 2$, $n=\sqrt\ell$,
$b=2\ln\ell$, 
for $\ell$ large enough and $q$ small enough, 
$$P\big(
\tau({2\ln\ell})>\sqrt{\ell}\,|\,Y_0=0\big)
\,\leq\,
\exp-\frac{a\sqrt\ell}{4}
\,,
$$
whence
\begin{multline*}
P\big(Y_{t}\geq\ln\ell
\,|\,Y_0=0\big)\,\geq\,
%\cr
%\,\geq\,
\bigg(1-\ln\ell\left(
\frac{(4\kappa\ln\ell)^{2}}{\ell}
\right)^{\ln\ell}\bigg)
\Big(1-
\exp-\frac{a\sqrt\ell}{4}
\Big)\cr
\,\geq\,
1-\exp\Big(-\frac{1}{2}(\ln\ell)^2\Big)
\,,
\end{multline*}
where the last inequality holds for $\ell$ large enough.
\end{proof}
\medskip

\noindent
{\bf Completion of the fall.}
We show here that, for $\ve>0$, after a time 
$4{\ell}/{a\ve}$,
the Markov chain 
$(Y_n)_{n\geq 0}$ is with high probability in the interval
$[ \lk(1-\ve)
, \ell]$.
\begin{proposition}\label{secdesc}
%Suppose that $\ell$ goes to $+\infty$, $q$ goes to $0$ and
%$q\ell$ converges towards $a\in [0,+\infty]$.
We suppose that 
$\ell\to +\infty\,, q\to 0\,,
{\ell q} \to a\in ]0,+\infty[$.
Let $\ve\in]0,1[$.
There exists $c(\ve)>0$ such that,
for $\ell$ large enough and $q$ small enough, we have
$$%\forall \ell\geq \ell_0\qquad
\forall t\geq\frac{4\ell}{a\ve}\qquad
P\big(Y_{t}\geq\lk(1-\ve)
\,|\,Y_0=0\big)\,\geq\,
1-
\exp(-c(\ve)\ell)
%\frac{6}{ {\ell}\ve}
\,.$$
\end{proposition}
\begin{proof}
Let $t\geq{4\ell}/{(a\ve)}$.
We write
\begin{multline*}
P\big(Y_{t}\geq
\lk(1-\ve)
\,|\,Y_0=0\big)\cr
\,\geq\,
P\Big(Y_{t}\geq
\lk(1-\ve),\,
\tau( \lk(1-\ve/2)
)\leq 
\frac{4\ell}{a\ve}
\,|\,Y_0=0\Big)\cr
\,=\,
\sum_{n=1}^{
{4\ell}/({a\ve})
%\frac{2\ell}{a\ve}
}
\sum_{k\geq 
\lk(1-\ve/2)}
\kern-5pt
P\big(Y_{t}\geq
\lk(1-\ve)
\,|\,
\tau( \lk(1-\ve/2)
)=n,\,Y_n=k,\,Y_0=0\big)
\cr
\times
P\big(
\tau( \lk(1-\ve/2)
)=n,\,Y_n=k\,|\,Y_0=0\big)\,.
\end{multline*}
Now, for $n\leq 
{4\ell}/({a\ve})$
and ${k\geq 
\lk(1-\ve/2)}$,
by the Markov property,  
\begin{multline*}
P\big(Y_{t}\geq
\lk(1-\ve)
\,|\,
\tau( \lk(1-\ve/2)
)=n,\,Y_n=k,\,Y_0=0\big)
\cr
\,=\,
P\big(Y_{t}\geq
\lk(1-\ve)
\,|\,
Y_n=k\big)
\,=\,
P\big(Y_{t-n}\geq
\lk(1-\ve)
\,|\,
Y_0=k\big)
\cr
%\sum_{l=
 %\lk(1-2\ve)}^\ell
%P\big(Y_{t}\geq
%\lk(1-4\ve)\,,Y_n=l
%\,|\,
%\tau( \lk(1-2\ve)
%)=n,\,Y_0=0\big)
%\,=\,
%\cr
%\sum_{l=
 %\lk(1-2\ve)}^\ell
%P\big(Y_{t}\geq
%\lk(1-4\ve)
%\,|\,
%Y_n=l
%\big)
%P(Y_n=l\,|\,
%\tau( \lk(1-2\ve)
%)=n,\,Y_0=0\big)
%\cr
\,\leq\,
P\big(Y_{t-n}\geq
\lk(1-\ve)
\,|\,
Y_0=
 \lk(1-\ve/2)
\big)\,.
\end{multline*}
We have used the monotonicity of the process $(Y_t)_{t\geq 0}$
with respect to the starting point to get the last inequality.
For $b<
\lk(1-\ve)
$, we have 
by lemmas~\ref{majinv} and~\ref{exco},
$$P\big(Y_{t-n}=b\,|\,
Y_0=
 \lk(1-\ve/2)
\big)
\,\leq\,
\frac{\cB(b)}
{\cB(
 \lk(1-\ve/2)
)}
\,\leq\,
\frac{\cB(
 \lk(1-\ve))}
{\cB(
 \lk(1-\ve/2)
)}
\,,$$
whence
$$P\big(Y_{t-n}\geq
\lk(1-\ve)
\,|\,
Y_0=
 \lk(1-\ve/2)
\big)
\,\geq\,
1-
 \lk(1-\ve)
\frac{\cB(
 \lk(1-\ve))}
{\cB(
 \lk(1-\ve/2)
)}\,.
$$
%\frac{(2\kappa\ln\ell)^{2\ln\ell}}
%{\ell^{\ln\ell}}
Thanks
to the large deviation estimates of lemma~\ref{gdb}, we have
$$\limsup_{\ell\to\infty}
\frac{1}{\ell}\ln\bigg(
 \lk(1-\ve)
\frac{\cB(
 \lk(1-\ve))}
{\cB(
 \lk(1-\ve/2)
)}\bigg)\,<\,0\,,$$
thus there exists $c(\ve)>0$ such that,
for $\ell$ large enough
$$P\big(Y_{t-n}\geq
\lk(1-\ve)
\,|\,
Y_0=
 \lk(1-\ve/2)
\big)
\,\geq\,
1-\exp(-c(\ve)\ell)
\,.$$
Reporting this inequality in the previous sum, we get
\begin{multline*}
P\big(Y_{t}\geq
\lk(1-\ve)
\,|\,Y_0=0\big)
\,\geq\,\cr
\Big(
1-\exp(-c(\ve)\ell)
\Big)
P\Big(
\tau( \lk(1-\ve/2)
)\leq 
\frac{4\ell}{a\ve}
\,|\,Y_0=0\Big)\,.
\end{multline*}
We apply
lemma~\ref{tbdhit} with
$b=\lk(1-\ve/2)$, $\lambda>0$ and $n=4\ell/(a\ve)$:
for $\ell$ large enough and $q$ small enough, 
$$\displaylines{
P\Big(
\tau( \lk(1-\frac{\ve}{2})
)>
\frac{4\ell}{a\ve}
\,|\,Y_0=0\Big)\,\leq\,\hfill\cr
\exp\Big(\lambda \lk(1-\frac{\ve}{2})+
\frac{4\ell}{a\ve}\Big(
\lk(1-\frac{\ve}{2})\frac{p}{\kappa}(e^\lambda-1)
+(\ell-\lk(1-\frac{\ve}{2}))
p\frac{\kappa-1}{\kappa}(e^{-\lambda}-1)\Big)\Big)
\,.
}$$
We send $\ell$ to $\infty$ and $q$ to $0$ in such a way that
$\ell q$ converges to $a>0$. We obtain
\begin{multline*}
\limsup_{
\genfrac{}{}{0pt}{1}{\ell\to\infty}
{q\to 0,\,
{\ell q} \to a}
}
\,\frac{1}{\ell}\ln
P\Big(
\tau( \lk(1-\frac{\ve}{2})
)>
\frac{4\ell}{a\ve}
\,|\,Y_0=0\Big)\,\leq\,\cr
\lambda \frac{\kappa-1}{\kappa}(1-\frac{\ve}{2})+
\frac{4}{\ve}\Big(
(1-\frac{\ve}{2})\frac{1}{\kappa}(e^\lambda-1)
+\big(\frac{1}{\kappa-1}+\frac{\ve}{2}\big)
\frac{\kappa-1}{\kappa}(e^{-\lambda}-1)\Big)
\,.
\end{multline*}
Expanding the last term as $\lambda$ goes to $0$, we see that it is
negative for $\lambda$ small enough, therefore there exists $c'(\ve)>0$
such that
for $\ell$ large enough and $q$ small enough, 
$$P\Big(
\tau( \lk(1-\frac{\ve}{2})
)>
\frac{4\ell}{a\ve}
\,|\,Y_0=0\Big)\,\leq\,
\exp(-c'(\ve)\ell)\,.$$
Reporting in the previous inequality on $Y_t$, 
we obtain that
$$P\big(Y_{t}\geq
\lk(1-\ve)
\,|\,Y_0=0\big)
\,\geq\,
\Big(
1-\exp(-c(\ve)\ell)
\Big)
\Big(
1-\exp(-c'(\ve)\ell)
\Big)
$$
and this yields the desired result.
\end{proof}
\subsection{Falling to equilibrium from the right}
For $b\in\zl$, we define
the hitting time $\theta(b)$ of $\{\,0,\dots,b\,\}$ by
$$\theta(b)\,=\,\inf\,\big\{\,n\geq 0: Y_n\leq b\,\big\}\,.
\index{$\theta(b)$}
$$
\begin{proposition}\label{lowerrdesc}
We suppose that 
$\ell\to +\infty$, $q\to 0$
in such a way that
$${\ell q} \to a\in ]0,+\infty[\,.$$
For $\ell$ large enough and $q$ small enough, we have
$$P\big(\theta(\lk)\leq\ell^2\,|\,Y_0=\ell\big)\,\geq\,
\Big(1-\frac{1}{a(\ln\ell)^2}\Big)
\Big(\frac{p}{\kappa}\Big)^{\ln\ell}
e^{-2a}\,.
%(1-p)^\ell\,.
$$
\end{proposition}
\begin{proof}
Our first goal is to estimate, for $b$ larger than 
$\lk=\ell(1-1/\kappa)$ and $n\geq 1$, the
probability
$$P\big(\theta(b)>n\,|\,Y_0=\ell\big)\,.$$
%\medskip
%
%\noindent
%{\bf Falling from $\ell$ to $\lk$.}
Suppose that
$\theta(b)> n$. Then $Y_{n-1}>b$ and
$$Y_{n}\,\leq\,Y_{n-1}
-\sum_{k=1}^{b}1_{U_{n,k}<p/\kappa}
+\sum_{k=b+1}^\ell1_{U_{n,k}>1-p(1-1/\kappa)}\,.$$
Iterating this inequality, we see that, if $Y_0=\ell$,
on the event
$\{\,\theta(b)> n\,\}$, we have
$Y_n\leq \ell + V_n$, where
$$V_n\,=\,
\sum_{t=1}^n\bigg(
-\sum_{k=1}^{b}1_{U_{t,k}<p/\kappa}
+\sum_{k=b+1}^\ell1_{U_{t,k}>1-p(1-1/\kappa)}\bigg)\,.$$
Therefore
$$P\big(\theta(b)>n\,|\,Y_0=\ell\big)\,\leq\,
P(\ell+V_n>b)\,.$$
We shall bound
$P(\ell+V_n>b)$ with the help of Chebyshev's inequality. Let us compute
the mean and the variance of $V_n$. Since $V_n$ is a sum of independent
Bernoulli random variables, we have
\begin{align*}
E(V_n)\,&=\,
n\big(-b\,\frac{p}{\kappa}+(\ell-b)\,p\,(1-\frac{1}{\kappa})\big)
\,=\,
np\,\big(\lK-b\big)\,,\cr
\var(V_n)\,&=\,
n\Big(b\,
\frac{p}{\kappa}
\big(1-\frac{p}{\kappa}\big)
+(\ell-b)\,
p\,\big(1-\frac{1}{\kappa}\big)
\big(1-p\,\big(1-\frac{1}{\kappa}\big)\big)
\Big)\cr
\,&\leq\,n\big(bp+(\ell-b)p\big)\,=\,n\ell p\,.
\end{align*}
We suppose that
$b-\ell>
np\,\big(\lK-b\big)$.
By Chebyshev's inequality, we have then
\begin{multline*}
P(\ell+V_n\geq b)\,=\,
P\Big(V_n-E(V_n)\geq 
b-\ell- np\,\big(\lK-b\big)
\Big)
\cr
\,\leq\,
\frac{\var(V_n)}{
\big(b-\ell- np\,\big(\lK-b\big)\big)^2}
\,.\hfil
\end{multline*}
%We next specialize this inequality to get estimates relevant
%to our regime.
%As in theorem~\ref{mainth},
%we suppose that 
%$$\ell\to +\infty\,,\qquad q\to 0\,,$$
%in such a way that
%$${\ell q} \to a\in ]0,+\infty[\,.$$
We take $n=\ell^2$ and $b=\ln\ell+\lk$. Then,
for $\ell$ large enough, 
$$b-\ell-
np\,\big(\lK-b\big)\,=\,\ln\ell+\lk-\ell+\ell^2p\ln\ell
\,\sim\,
%\frac{1}{2}
%\ell^2p\ln\ell
\ell^2 p \ln\ell
\,>\,0\,,
$$
whence, by the previous inequalities,
for $\ell$ large enough and $q$ small enough, 
$$P\big(\theta(
\ln\ell+\lk
)\geq \ell^2
\,|\,Y_0=\ell\big)\,\leq\,
%P\Big(\ell+V_{\ell^2}>
%\ln\ell+\lk
%\Big)\,\leq\,
%\frac{4\ell^3p}
%{(\ell a\ln\ell)^2}
%\,\leq\, 
\frac{1}
{ a(\ln\ell)^2}\,.
$$
%Moreover, we have
%$$P\big(Y_1\leq\lk\,|\,Y_0=\ln\ell+\lk)\,\geq\,
%$$
%P\Big(\ell+V_{\ell^2}>
%\ln\ell+\lk
%\Big)\,\leq\,
We decompose next
\begin{multline*}
P\big(\theta(\lk)\leq\ell^2\,|\,Y_0=\ell\big)\,\geq\,
P\big(\theta(\lk+\ln\ell)<\ell^2,\,
\theta(\lk)\leq\ell^2\,|\,Y_0=\ell\big)\cr
\,=\,
\sum_{t<\ell^2}\sum_{b\leq\lk+\ln\ell}
P\big(\theta(\lk+\ln\ell)=t,\,Y_t=b,\,
\theta(\lk)\leq\ell^2\,|\,Y_0=\ell\big)\cr
\,=\,
\sum_{t<\ell^2}\sum_{b\leq\lk+\ln\ell}
P\big(
\theta(\lk)\leq\ell^2\,|\,
\theta(\lk+\ln\ell)=t,\,
Y_t=b,\,
Y_0=\ell\big)
\cr
\hfill\times
P\big(\theta(\lk+\ln\ell)=t,\,Y_t=b \,|\,Y_0=\ell\big)\,.
\end{multline*}
By the Markov property and the monotonicity of the process
$(Y_n)_{n\geq 0}$, we have,
for $t<\ell^2$ and $b\leq \lk+\ln\ell$,
\begin{multline*}
P\big(
\theta(\lk)\leq\ell^2\,|\,
\theta(\lk+\ln\ell)=t,\,
Y_t=b,\,
Y_0=\ell\big)
\cr
\,=\,
P\big(
\theta(\lk)\leq\ell^2-t\,|\,Y_0=b\big)
\,\geq\,
P\big(
\theta(\lk)\leq\ell^2-t\,|\,Y_0=\lk+\ln\ell\big)
\cr
\,\geq\,
P\big(
Y_1=\lk\,|\,Y_0=\lk+\ln\ell\big)
\,=\,M_H
(\lk+\ln\ell,\lk)\,.
\end{multline*}
Reporting this inequality in the previous sum, we get
$$
P\big(\theta(\lk)\leq\ell^2\,|\,Y_0=\ell\big)\,\geq\,
P\big(\theta(\lk+\ln\ell)<\ell^2
\,|\,Y_0=\ell\big)\,
	M_H (\lk+\ln\ell,\lk)\,.
$$
We have already proved that
$$P\big(\theta(
\ln\ell+\lk
)\geq \ell^2
\,|\,Y_0=\ell\big)\,\leq\,
\frac{1}
{ a(\ln\ell)^2}\,.
$$
%Moreover, we have
%$$P\big(Y_1\leq\lk\,|\,Y_0=\ln\ell+\lk)\,\geq\,
%By lemma~\ref{bdhit} applied with $b=\lk-\ln\ell$ and $n=\ell^2$,
%we have for $\ell$ large enough
%$$P\big(\theta(\lk-\ln\ell)>\ell^2
%\,|\,Y_0=0\big)\,\leq\,\frac{5}{a(\ln\ell)^2}\,.$$
Moreover,
for $\ell$ large enough and $q$ small enough, 
$$M_H (\lk+\ln\ell,\lk)\,\geq\,
\Big(\frac{p}{\kappa}\Big)^{\ln\ell}(1-q)^\ell\,\geq\,
\Big(\frac{p}{\kappa}\Big)^{\ln\ell}
e^{-2a}\,.
$$
Putting the previous inequalities together, we obtain the
desired lower bound.
\end{proof}

\noindent
We derive next a large deviation upper bound for the time needed
to go from $\ell$ to $0$. This will yield an upper bound on the discovery time.
We define
$$\tau_0\,=\,\inf\,\big\{\,n\geq 0: Y_{n}=0\,\big\}\,.$$
\begin{proposition}\label{visitz}
For any $a\in]0,+\infty[$,
$$
\limsup_{
\genfrac{}{}{0pt}{1}{\ell\to\infty,\,
q\to 0}
{{\ell q} \to a}
}
\,\frac{1}{\ell}\ln
E(\tau_0\,|\,Y_0=\ell)\,\leq\,\ln\kappa\,.
$$
\end{proposition}
\begin{proof}
We prove that, starting from $\ell$, the walker has probability of order
$1/\kappa^\ell$ to visit $0$ before time $\ell^2$.
To do this, we decompose the trajectory until time $\ell^2$
into two parts: the
descent to the
equilibrium $\lk$, which is very likely to occur, and
the ascent to $0$, which is very unlikely to occur.
We  estimate the probability of the ascent with the help of
a beautiful technique
developed by Schonmann \cite{NESC}
in a different context, namely the study of the metastability of
the Ising model. More precisely, we use the reversibility of the process
to relate the probability of an ascending path 
%going from $\lk$ to $0$
to 
the probability of a descending path.
% going from $0$ to $\lk$. 
It turns
out that the most likely way to go from $\lk$ to $0$
is obtained as the time--reverse of a typical path going
from $0$ to $\lk$.
Thanks to the monotonicity of the process, this estimate
yields a lower bound on the hitting time of $0$ which 
is uniform with respect to the starting point.
We bound then easily
$E(\tau_0\,|\,Y_0=\ell)$
by summing over intervals of length $\ell^2$
and using the Markov property.

We should normally work with $\lfloor\lk\rfloor$ instead of $\lk$.
To alleviate the notation, we do as
if $\lk$ was an integer.
We write
\begin{multline*}
P\big(\tau_0\leq 2\ell^2\,|\,Y_0=\ell\big)\,\geq\,
P\big(\theta(\lk)\leq\ell^2,\,
\tau_0\leq 2\ell^2\,|\,Y_0=\ell\big)\cr
\,=\,
\sum_{t\leq\ell^2}\sum_{b\leq\lk}
P\big(\theta(\lk)=t,\,Y_t=b,\,
\tau_0\leq 2\ell^2\,|\,Y_0=\ell\big)\cr
\,=\,
\sum_{t\leq\ell^2}\sum_{b\leq\lk}
P\big(
\tau_0\leq 2\ell^2\,|\,
\theta(\lk)=t,\,
Y_t=b,\,
Y_0=\ell\big)
\cr
\hfill\times
P\big(\theta(\lk)=t,\,Y_t=b \,|\,Y_0=\ell\big)\,.
\end{multline*}
By the Markov property and the monotonicity of the process
$(Y_n)_{n\geq 0}$,
we have,
for $t\leq\ell^2$ and $b\leq \lk$,
\begin{multline*}
P\big(
\tau_0\leq 2\ell^2\,|\,
\theta(\lk)=t,\,
Y_t=b,\,
Y_0=\ell\big)
\cr
\,\geq\,
P\big(
\exists n\in\{\,t,\dots,2\ell^2\,\}\quad Y_n=0\,|\,
\theta(\lk)=t,\,
Y_t=b,\,
Y_0=0\big)
\cr
\,=\,
P\big(
\tau_0\leq 2\ell^2-t\,|\,Y_0=b\big)
\,\geq\,
P\big(
\tau_0\leq \ell^2\,|\,Y_0=\lk\big)\,.
\end{multline*}
Reporting this inequality in the previous sum, we get
$$
P\big(\tau_0\leq 2\ell^2\,|\,Y_0=\ell\big)\,\geq\,
P\big(\theta(\lk)\leq\ell^2
\,|\,Y_0=\ell\big)\,
P\big(
\tau_0\leq \ell^2\,|\,Y_0=\lk\big)\,.
$$
We estimate next the probability of the ascending part, i.e., the last probability
in the above formula.
We start with the estimate of proposition~\ref{lowerdesc}:
$$P\big(\tau(\lk)\leq\ell^2\,|\,Y_0=0\big)\,\geq\,
\Big(1-\frac{5}{a(\ln\ell)^2}\Big)
\Big(\frac{p}{\kappa}\Big)^{\ln\ell}
e^{-2a}\,.
$$
Yet
\begin{multline*}
P\big(\tau(\lk)\leq\ell^2\,|\,Y_0=0\big)\,=\,
P\big(\exists t\leq\ell^2\quad Y_t\geq\lk\,|\,Y_0=0\big)\cr
\,\leq\,\sum_{b\geq\lk}
P\big(\exists t\leq\ell^2\quad Y_t=b\,|\,Y_0=0\big)\,.
\end{multline*}
From the last inequalities, we see that there exists $b\geq\lk$ such that
$$P\big(\exists t\leq\ell^2\quad Y_t=b\,|\,Y_0=0\big)\,\geq\,
\frac{1}{\ell}\Big(1-\frac{5}{a(\ln\ell)^2}\Big)
\Big(\frac{p}{\kappa}\Big)^{\ln\ell}
e^{-2a}\,.$$
Using the reversibility of $M_H$ with respect to $\cB$
(see proposition~\ref{ivy}), we have
\begin{multline*}
\cB(b)\,
P\big(\tau_0\leq\ell^2\,|\,Y_0=b\big)
\cr
\,=\,
\sum_{t\leq\ell^2}\sum_{y_1,\dots,y_{t-1}>0}
\cB(b)\,M_H(b,y_1)\,\dots
\,M_H(y_{t-1},0)
\cr
\,=\,
\sum_{t\leq\ell^2}\sum_{y_1,\dots,y_{t-1}>0}
\cB(0)
M_H(0,y_{t-1})
\,\dots\,
M_H(y_1,b)
\cr
\,=\,
\cB(0)\,
P\big(\exists t\leq \ell^2\quad Y_t=b\,|\,Y_0=0\big)\,.
%P\big(\tau_\kappa\leq\ell^2\,|\,Y_0=0\big)\,.
\end{multline*}
Thus
%P\big(\tau_\kappa\leq\ell^2\,|\,Y_0=0\big)\,=\,
\begin{multline*}
P\big(\tau_0\leq\ell^2\,|\,Y_0=b\big)\,=\,
\frac{\cB(0)}{\cB(b)}
P\big(\exists t\leq \ell^2\quad Y_t=b\,|\,Y_0=0\big)
\cr
\,\geq\,
\frac{\kappa^{-\ell}}{\ell}\Big(1-\frac{5}{a(\ln\ell)^2}\Big)
\Big(\frac{p}{\kappa}\Big)^{\ln\ell}
e^{-2a}\,.
\end{multline*}
By monotonicity of the process
$(Y_t)_{t\geq 0}$, since $b\geq\lk$, then
$$P\big(\tau_0\leq\ell^2\,|\,Y_0=\lk\big)\,\geq\,
P\big(\tau_0\leq\ell^2\,|\,Y_0=b\big)\,.$$
Using proposition~\ref{lowerrdesc} and the previous inequalities, we conclude that
$$P\big(\tau_0\leq2\ell^2\,|\,Y_0=\ell\big)\,\geq\,
\frac{\kappa^{-\ell}}{\ell}
\bigg(\
\Big(1-\frac{5}{a(\ln\ell)^2}\Big)
\Big(\frac{p}{\kappa}\Big)^{\ln\ell}
e^{-2a}
\bigg)^2\,.
$$
Let $\ve>0$. For $\ell$ large enough and $q$ small enough,
$$P\big(\tau_0\leq2\ell^2\,|\,Y_0=\ell\big)\,\geq\,
\frac{1}{\kappa^{\ell(1+\ve)}}\,.$$
Now, for $n\geq 0$,
\begin{multline*}
P\big(\tau_0> 2n\ell^2\,|\,Y_0=\ell\big)\cr
\,=\,
\sum_{b\geq 1}
P\big(\tau_0> 2n\ell^2,
Y_{2(n-1)\ell^2}=b,\,
\tau_0> 2(n-1)\ell^2
\,|\,Y_0=\ell\big)
\cr
\,=\,
\sum_{b\geq 1}
P\big(\tau_0> 2n\ell^2\,|\,
Y_{2(n-1)\ell^2}=b,\,
\tau_0> 2(n-1)\ell^2
,\,Y_0=\ell\big)
\cr
\times
P\big(
Y_{2(n-1)\ell^2}=b,\,
\tau_0> 2(n-1)\ell^2
\,|\,\,Y_0=\ell\big)\,.
\end{multline*}
By the Markov property and the monotonicity of the process, we have
%for $t\leq\ell^2$ and $b\geq \lk-\ln\ell$,
\begin{multline*}
P\big(\tau_0> 2n\ell^2\,|\,
Y_{2(n-1)\ell^2}=b,\,
\tau_0> 2(n-1)\ell^2
,\,Y_0=\ell\big)\cr
\,=\,
P\big(\tau_0> 2n\ell^2\,|\,
Y_{2(n-1)\ell^2}=b
\big)\,=\,
P\big(\tau_0> 2\ell^2\,|\,
Y_{0}=b
\big)\cr
\,\leq\,
P\big(\tau_0> 2\ell^2\,|\,
Y_{0}=\ell
\big)\,\leq\,
1-\frac{1}{\kappa^{\ell(1+\ve)}}\,.$$
\end{multline*}
Reporting in the previous sum, we get
$$P\big(\tau_0> 2n\ell^2\,|\,Y_0=\ell\big)
\,\leq\,
\Big(1-\frac{1}{\kappa^{\ell(1+\ve)}}\Big)
P\big(\tau_0> 2(n-1)\ell^2\,|\,Y_0=\ell\big)\,.
$$
Iterating, we obtain
$$P\big(\tau_0> 2n\ell^2\,|\,Y_0=\ell\big)
\,\leq\,
\Big(1-\frac{1}{\kappa^{\ell(1+\ve)}}\Big)^n\,.$$
Thus
\begin{multline*}
E(\tau_0\,|\,Y_0=\ell)\,=\,
\sum_{n\geq 1}
P(\tau_0\geq n\,|\,Y_0=\ell)\cr
\,\leq\,
\sum_{n\geq 0}
\sum_{t=2n\ell^2+1}^{2(n+1)\ell^2}
P(\tau_0\geq t\,|\,Y_0=\ell)
\,\leq\,
\sum_{n\geq 0}
2\ell^2P(\tau_0> 2n\ell^2\,|\,Y_0=\ell)\cr
\,\leq\,
2\ell^2\sum_{n\geq 0}
\Big(1-\frac{1}{\kappa^{\ell(1+\ve)}}\Big)^n
\,=\,
2\ell^2 {\kappa^{\ell(1+\ve)}}
\,.$$
\end{multline*}
This bound is true for any $\ve>0$. Sending successively $\ell$ to
$\infty$ and $\ve$ to $0$, we obtain the desired upper bound.
\end{proof}
\subsection{Discovery time}
The dynamics of the processes 
$(O^\ell)_{t\geq 0}$,
$(O^1)_{t\geq 0}$ 
in $\cN$ are the same as the original process
$(O_t)_{t\geq 0}$, therefore we can use the original process to compute 
their corresponding
discovery times.
%Letting, for $\theta=1,\ell$,
Letting
$$\displaylines{\tau^{*,\ell} \,=\,\inf\,\big\{\,t\geq 0: 
O^\ell_t\in\cW
\,\big\}\,,\qquad
\tau^{*,1} \,=\,\inf\,\big\{\,t\geq 0: 
O^1_t\in\cW
\,\big\}\,,\cr
\tau^{*} \,=\,\inf\,\big\{\,t\geq 0: 
O_t\in\cW
\,\big\}\,,
\index{$\tau^*$}
}
$$
we have indeed
%both processes have the same discovery time:
%$$E\big(\tau^*\text{ for the process
%$(O^\ell)_{t\geq 0}$}
$$\displaylines{
E\big(\tau^{*,\ell}
\,|\,
O^\ell_0=\ota
\big)\,=\,
E\big(\tau^*\,|\,
%O_0= (0,m,0,\dots,0)
O_0= (0,0,0,\dots,m)
\big)\,,\cr
E\big(\tau^{*,1}
\,|\,
O^1_0=\otb
\big)\,=\,
E\big(\tau^*\,|\,
O_0= (0,m,0,\dots,0)
%O_0= (0,0,0,\dots,m)
\big)\,.
}$$
In addition, the law of the discovery time $\tau^*$ is the same for the
distance process and the occupancy process.
With a slight abuse of notation, we let
$$\tau^{*} \,=\,\inf\,\big\{\,t\geq 0: 
D_t\in\cW
\,\big\}\,.$$
{\bf Notation.}
For $b\in\zl$, we denote by 
$(b)^m\index{$(b)^m$}$ 
the vector column whose components are all equal to $b$:
$$(b)^m\,=\,\left(
\begin{matrix}
b\\
\vdots\\
b
\end{matrix}
\right)
\,.$$
%\medskip

\noindent
We have
$$\displaylines{
E\big(\tau^*\,|\,
O_0= (0,0,\dots,0,m)
\big)\,=\,
E\big(\tau^*\,|\,
D_0= (\ell)^m
%\left(
%\begin{matrix}
%\ell\\
%\vdots\\
%\ell
%\end{matrix}
\big)\,,\cr
E\big(\tau^*\,|\,
O_0= (0,m,0,\dots,0)
\big)\,=\,
E\big(\tau^*\,|\,
D_0= (1)^m\big)
\,.
}$$
We will carry out the estimates of $\tau^*$ for the distance process
$(D_n)_{n\geq 0}$. 
Notice that the case $\alpha=+\infty$ is not covered by the result
of next proposition. This case will be handled separately, with the
help of the intermediate inequality of corollary~\ref{majext}.
%As in theorem~\ref{mainth},
%we suppose that 
%%$\ell$ goes to $\infty$, $m$ goes to $\infty$ and $q$ goes to $0$ 
%$$\ell\to +\infty\,,\qquad m\to +\infty\,,\qquad q\to 0\,,$$
%%$$\ell\to +\infty\,,\qquad q\to 0\,,$$
%in such a way that
%$${\ell q} \to a\in ]0,+\infty[\,,
%\qquad\frac{m}{\ell}\to\alpha\in [0,+\infty]\,.$$
%$${\ell q} \to a\in [0,+\infty]\,,
%\qquad\frac{m}{\ell}\to\alpha\in [0,+\infty]\,.$$
%We consider the regime where
%$\ell,m$ go to $\infty$ 
%and
%$q$ to $0$ 
%in such a way that $m/\ell$ and $\ell q$ are kept constant.
%We set
%$$a\,=\,\ell\, q,\quad\alpha\,=\,\frac{m}{\ell}\,.$$
%The associated limit operator is denoted
%In this regime, we have
%$$\displaylines{
%\lim_{
%\genfrac{}{}{0pt}{1}{\ell\to\infty,\,
%q\to 0}
%{{\ell q} \to a}
%}
%\,
%M_H(0,0)\,=\,\exp(-a)\,,\cr
%%\lim_{
%%\genfrac{}{}{0pt}{1}{\ell,m\to\infty}
%%{q\to 0}
%%}
%%\lim_{
%%\genfrac{}{}{0pt}{1}{\ell\to\infty,\,
%q\to 0}
%{{\ell q} \to a}
%}
%\,
%M_H(1,0)\,=\,
%%\lim_{
%%\genfrac{}{}{0pt}{1}{\ell,m\to\infty}
%%{q\to 0}
%%}
%\lim_{
%\genfrac{}{}{0pt}{1}{\ell\to\infty,\,
%q\to 0}
%{{\ell q} \to a}
%}
%\,
%M_H(\ell,0)\,=\,
%0\,.}$$
%$$\lim_{\ell\to\infty}\,
%M_H(0,0)\,=\,\exp(-a)\,,\qquad
%\lim_{\ell\to\infty}\,
%M_H(1,0)\,=\,
%\lim_{\ell\to\infty}\,
%M_H(\ell,0)\,=\,
%0\,.$$
\begin{proposition}\label{bdd}
%Let $\rho^*=\rho(\beta,0)$. Let $\eta>0$. For $
Let $a\in ]0,+\infty[$ and
 $\alpha\in [0,+\infty[$.
For any $d\in\cN$, 
%we have
%\lim_{l,m,q\to \infty,\infty,0}
$$ \lim_{
\genfrac{}{}{0pt}{1}{\ell,m\to\infty,\,
q\to 0
}
{{\ell q} \to a,\,
\frac{\scriptstyle m}{\scriptstyle \ell}\to\alpha
}
}
\,\frac{1}{\ell}\ln
E\big(\tau^*\,|\,
D_0= d\big)\,=\,\ln\kappa\,.$$
%Moreover
%$$ \lim_{
%\genfrac{}{}{0pt}{1}{\ell,m\to\infty,\,
%q\to 0
%}
%{{\ell q} \to a,\,
%\frac{\scriptstyle m}{\scriptstyle \ell}\to \infty
%}
%}
%\,\frac{1}{\ell}\ln
%E\big(\tau^*\,|\,
%D_0= d\big)\,\leq\,\ln\kappa\,.$$
\end{proposition}
\begin{proof}
Since we are in the neutral case $\sigma=1$, 
then,
by corollary~\ref{corneu},
the distance process
$(D_n)_{n\geq 0}$  is monotone. Therefore, 
for any $d\in\cN$, we have
$$\displaylines{
E\big(\tau^*\,|\,
D_0= (1)^m
\big)
\,\leq\,
E\big(\tau^*\,|\,
D_0= d\big)\,\leq\,
E\big(\tau^*\,|\,
D_0= (\ell)^m
\big)\,.
}$$
As in the section~\ref{mutdyn}, 
we consider a Markov chain 
$(Y_n)_{n\geq 0}$ with state space $\zl$ and having for
transition matrix the lumped mutation matrix $M_H$.
We consider also
a sequence 
$(\ve_n)_{n\geq 1}$ 
of i.i.d. Bernoulli random variables 
with parameter
${1}/{m}$ and we set
%, as in section~\ref{mutdyn}, 
$$
\forall n\geq 1\qquad N(n)\,=\,\ve_1+\cdots+\ve_n\,.
\index{$N(n)$}
$$
We suppose also that the processes
$(N(n))_{n\geq 1}$ and
$(Y_n)_{n\geq 0}$ are independent.
Let us look at the distance process at time $n$ starting from
$(\ell)^m$.
From proposition~\ref{nemarg},
we know that the law of the $i$--th chromosome in $D_n$ is the same as the
law of $Y_{N(n)}$ starting from~$\ell$.
The main difficulty is that, because of the replication events, the $m$
chromosomes present at time~$n$ are not independent, nor are their
genealogical lines.
However, this dependence does not improve significantly the efficiency
of the search mechanism, as long as the population is in the neutral
space $\cN$. To bound the discovery time $\tau^*$ from above,
we consider the time needed for a single chromosome to discover
the Master sequence $w^*$, that is
$$\wtau_0\,=\,\inf\,\big\{\,n\geq 0: Y_{N(n)}=0\,\big\}\,$$
and we observe that, if the master sequence has not been discovered until time~$n$
in the distance process, that is,
$$\forall t\leq n\quad\forall i\in\um\qquad D_t(i)\geq 1\,,$$
then certainly the ancestral line of any chromosome present at time~$n$
does not contain the master sequence. 
%In fact, even more is true. If we look at the ancestral line of each 
%chromosome present at time $n$ in the population $D_n$, we obtain a sequence
%chromosomes whose law is equal to the law of
By
proposition~\ref{neal}, the ancestral line of any chromosome 
present at time~$n$
has the same law as
$$Y_{N(0)}, Y_{N(1)},\dots, Y_{N(n)} \,.$$
From the previous observations, we conclude that
$$\forall n\geq 0\qquad
P\big(\tau^*>n\,|\,
D_0= (\ell)^m
\big)\,\leq\,
P(\wtau_0>n\,|\,Y_0=\ell)
\,.$$
Summing this inequality over $n\geq 0$, we have
$$E\big(\tau^*\,|\,
D_0= 
(\ell)^m
\big)\,\leq\,
E(\wtau_0\,|\,Y_0=\ell)\,.
$$
For $n\geq 0$, let
$$T_n
\,=\,\inf\,\big\{\,t\geq 0: {N(t)}=n\,\big\}\,.$$
The variables $T_n-T_{n-1}$, $n\geq 1$, have the same law,
therefore
$$\forall n\geq 0\qquad E(T_n)\,=\,n E(T_1)\,=\,nm\,.$$
We will next express the upper bound on $\tau^*$
as a function of
$$\tau_0\,=\,\inf\,\big\{\,n\geq 0: Y_{n}=0\,\big\}\,.$$
We compute
\begin{multline*}
E(\wtau_0\,|\,Y_0=\ell)\,=\,
\sum_{t\geq 1}
P(\wtau_0\geq t\,|\,Y_0=\ell)
\cr
\,=\,
\sum_{t\geq 1}
\sum_{n\geq 1}
P(\wtau_0\geq t,\,\tau_0=n\,|\,Y_0=\ell)
\cr
\,=\,
\sum_{t\geq 1}
\sum_{n\geq 1}
P(T_n\geq t,\,\tau_0=n\,|\,Y_0=\ell)
\cr
\,=\,
\sum_{n\geq 1}
\sum_{t\geq 1}
P(T_n\geq t)\,P(\tau_0=n\,|\,Y_0=\ell)
\cr
\,=\,
\sum_{n\geq 1}
E(T_n)\,P(\tau_0=n\,|\,Y_0=\ell)
\cr
\,=\,
\sum_{n\geq 1}
nm\,P(\tau_0=n\,|\,Y_0=\ell)
\,=\,
mE(\tau_0\,|\,Y_0=\ell)\,.
\end{multline*}
With the help of
proposition~\ref{visitz}, we conclude that
$$ \limsup_{
\genfrac{}{}{0pt}{1}{\ell,m\to\infty,\,
q\to 0
}
{{\ell q} \to a,\,
\frac{\scriptstyle m}{\scriptstyle \ell}\to\alpha
}
}
\,\frac{1}{\ell}\ln
E\big(\tau^*\,|\,
D_0= d\big)\,\leq\,\ln\kappa\,.$$
In fact, we have derived the following upper bound on the discovery time.
\begin{corollary}\label{majext}
Let $\tau_0$ be the hitting time of $0$ for the process
$(Y_n)_{n\geq 0}$.
For any $d\in\cN$, any $m\geq 1$, we have
$$E\big(\tau^*\,|\,
D_0= d\big)\,\leq\,
m\,E(\tau_0\,|\,Y_0=\ell)\,.
$$
\end{corollary}
The harder part is 
to bound the discovery time $\tau^*$ from below.
The main difficulty to obtain the adequate lower bound on $\tau^*$
is that the process starts very close to the master sequence, hence 
the probability of creating quickly a master sequence is not very small.
Our strategy consists in exhibiting a scenario in which the whole
population
is driven into a neighborhood of the equilibrium $\lk$.
Once the whole population is close to $\lk$, the probability to create
a master sequence in a short time is of order $1/\kappa^\ell$, thus it
requires a time of order $\kappa^\ell$.
The key point is to design a scenario whose probability is much larger
than
$1/\kappa^\ell$. Indeed, the discovery time is bounded from below
by the probability of the scenario multiplied by $\kappa^\ell$.
We rely on the following scenario.
First we ensure that until time $m\ltq$, no mutation can recreate
the master sequence. This implies that
$\tau^*> m\ltq$.
Let us look at the population at time $m\ltq$.
Each chromosome present at this time has undergone an evolution
whose law is the same as the mutation dynamics studied
in section~\ref{mutdyn}. 
The initial drift of the mutation dynamics
is quite violent,
therefore at time $m\ltq$, it is very unlikely that a 
chromosome
%evolving with the mutation dynamics 
is still in 
$\{\,0,\cdots,\ln\ell\,\}$.
The problem is that the chromosomes are not independent. 
We take care of this problem with the help of 
the FKG inequality and an exponential estimate.
Thus, at time $m\ltq$, in this scenario, all the chromosomes of the population
are at distance larger than $\ln\ell$ from the master sequence.
We wait next until time $m\ell^2$. 
Because of the mutation drift, 
a chromosome starting at $\ln\ell$ has a very low probability
of hitting $0$ before time $m\ell^2$.
%Starting from chromosomes
%larger than
%$\ln\ell$, the
Thus the process is 
very unlikely to discover the master sequence before time $m\ell ^2$.
Arguing again as before, we obtain that, for any $\ve>0$, 
at time $m\ell^2$, it is very unlikely that a particle
evolving with the mutation dynamics is still in 
$\{\,0,\cdots,\lk(1-\ve)\,\}$.
Thus, according to this scenario, we have $\tau^*>m\ell^2$
and
$$\forall i\in\um\qquad D_{m\ell^2}(i)\,\geq\,\lk(1-\ve)\,.$$
Let us precise next the scenario and the corresponding estimates.
We suppose that
the distance process starts from
$(1)^m$
and we will estimate the probability
of a specific scenario leading to a discovery time close 
to~${\kappa^{\ell}}$.
%${\kappa^{\ell(1-\ve)}}$.
%the following scenario:
%
%Let 
Let 
$\cE$ be the event
$$\cE\,=\,\big\{\,
\forall n\leq m\ltq\quad\forall l\leq\ln\ell\quad
U_{n,l}>p/\kappa\,\big\}\,.$$
If the event $\cE$ occurs, then, until time
$m\ltq$, none of the mutation events in the process
$(D_n)_{n\geq 0}$ can create a master sequence. Indeed,
on $\cE$,
\begin{multline*}
\forall b\in\ul\quad
\forall n\leq m\ltq
\qquad\cr
\cMH(b,U_{n,1},\dots,U_{n,\ell})\,\geq\,
\cMH(1,U_{n,1},\dots,U_{n,\ell})\cr
\,\geq\,
1+\sum_{l=2}^\ell1_{U_{n,l}>1-p(1-1/\kappa)}\,\geq\,1\,.
\end{multline*}
Thus, on the event $\cE$, we have $\tau^*\geq
m\ltq$.
The probability of $\cE$ is
$$P(\cE)\,=\,\Big(1-\frac{p}{\kappa}\Big)^{
m\ltq\ln\ell}\,.$$
Let $\ve>0$.
We suppose that the process starts from 
$(1)^m$
and we
estimate the probability
\begin{multline*}
P\big(\tau^*>
{\kappa^{\ell(1-\ve)}}\big)\,\geq\,
P\big(\tau^*>
{\kappa^{\ell(1-\ve)}},\,\cE\big)\cr
\,\geq\,
%P(\cE)\,P(\tau^*\geq
%{\kappa^{\ell(1-\ve)}}\,|\,\cE)
P\Big(\forall t\in\{\,m\ltq,\dots,
{\kappa^{\ell(1-\ve)}}\,\}\quad
D_t\in\cN,\,\cE\Big)\cr
\,=\,\sum_{d\in\cN}
P\Big(\forall t\in\{\,m\ltq,\dots,
{\kappa^{\ell(1-\ve)}}\,\}\quad
D_t\in\cN,\,
D_{m\ltq}=d,\,\cE\Big)\cr
\,\geq\,\sum_{d\geq(\ln \ell)^m}
P\Big(\forall t\in\{\,m\ltq,\dots,
{\kappa^{\ell(1-\ve)}}\,\}\quad
D_t\in\cN\,|\,
D_{m\ltq}=d,\,\cE\Big)\cr
\hfill
\times
P( D_{m\ltq}=d
,\,\cE)\,.
\end{multline*}
Using the Markov property, we have
\begin{multline*}
P\Big(\forall t\in\{\,m\ltq,\dots,
{\kappa^{\ell(1-\ve)}}\,\}\quad
D_t\in\cN\,|\,
D_{m\ltq}=d,\,\cE\Big)
\cr
\,=\,
P\Big(
\forall t\in\{\,0,\dots,
{\kappa^{\ell(1-\ve)}}
-m\ltq
\,\}\quad
D_t\in\cN\,|\,
D_{0}=d\Big)\cr
\,=\,
P\Big(
\tau^*>
{\kappa^{\ell(1-\ve)}}
-m\ltq
\,|\,
D_{0}=d\Big)
\,\geq\,
P\Big(
\tau^*>
{\kappa^{\ell(1-\ve)}}
\,|\,
D_{0}=d\Big)
\,.
\end{multline*}
In the neutral case, 
by corollary~\ref{corneu},
%by lemma~\ref{monphimor},
%the map $\Phi_H$ used to build
%the coupling of the process $(D_n)_{n\geq 0}$ is non--decreasing.
%In particular,
%the coupling is monotone with respect to the initial condition
the distance process is monotone.
Therefore,
for $d\geq(\ln \ell)^m$,
$$
P\big(
\tau^*>
{\kappa^{\ell(1-\ve)}}
%-m\ltq
\,|\,
D_{0}=d\big)\,\geq\,
P\big(
\tau^*>
{\kappa^{\ell(1-\ve)}}
%-m\ltq
\,|\,
D_{0}=
(\ln \ell)^m
\big)\,.
$$
Reporting in the previous sum, we get
\begin{multline*}
P\big(\tau^*>
{\kappa^{\ell(1-\ve)}}\big)
\,\geq\,\cr
P\big(
\tau^*>
{\kappa^{\ell(1-\ve)}}
%-m\ltq
\,|\,
D_{0}=
(\ln \ell)^m
\big)
\,P\big( 
D_{m\ltq}\geq
(\ln \ell)^m
,\,\cE\big)\,.
\end{multline*}
We first study the last term in the above inequality.
The status of the process at time
$m\ltq$  is a function of the random vectors
$$R_n\,=\,\big(S_n, I_n, J_n,U_{n,1},\dots,U_{n,\ell}\big)\,,\qquad
1\leq n\leq  m\ltq
\,.$$
We make an intermediate conditioning with respect to %the variables
$S_n, I_n, J_n$:
\begin{multline*}
\,P\big( 
D_{m\ltq}\geq
(\ln \ell)^m
,\,\cE\big)\cr\,=\,
E\Big(
\,P\big( 
%\forall i\in\um\quad
%D_{m\ltq}(i)\geq\ln\ell
D_{m\ltq}\geq
(\ln \ell)^m
,\,\cE
\,\big|\,
S_n, I_n, J_n,\,
1\leq n\leq m\ltq
\big)
\Big)
\,.
\end{multline*}
The variables $S_n, I_n, J_n,\,
1\leq n\leq  m\ltq$ being fixed, the state of the process
at time 
$m\ltq$  is a non--decreasing function of the variables
$$U_{n,1},\dots,U_{n,\ell},\,\qquad
1\leq n\leq  m\ltq\,.$$
Indeed, the mutation map $\cM_H(\cdot,u_1,\dots,u_\ell)$ is non--decreasing 
with respect to $u_1\dots,u_\ell$.
Thus the events $\cE$ and
$\{\,D_{m\ltq}\geq
(\ln \ell)^m\,\}$
%$D_{m\ltq}(i)\geq\ln\ell$, $1\leq i\leq m$,
%$$D_{m\ltq}(i)\geq\ln\ell,\quad 1\leq i\leq m$$
are both non--decreasing with respect to these variables.
By the FKG inequality for a product measure (see the end of the appendix),
we have
\begin{multline*}
P\big( 
%\forall i\in\um\quad
%D_{m\ltq}(i)\geq\ln\ell
D_{m\ltq}\geq
(\ln \ell)^m
,\,\cE
\,\big|\,
S_n, I_n, J_n,\,
1\leq n\leq  m\ltq
\big)
\,\geq\,\cr
%\prod_{1\leq i\leq m}
P\big( 
D_{m\ltq}\geq
(\ln \ell)^m
%D_{m\ltq}(i)\geq\ln\ell
\,\big|\,
S_n, I_n, J_n,\,
1\leq n\leq  m\ltq
\big)
\cr\hfill
\times
P\big( 
\cE
\,\big|\,
S_n, I_n, J_n,\,
1\leq n\leq m\ltq
\big)\,.
\end{multline*}
Yet $\cE$ does not depend on the variables
$S_n, I_n, J_n$, therefore
$$P\big( 
\cE
\,\big|\,
S_n, I_n, J_n,\,
1\leq n\leq  m\ltq
\big)\,=\,P(\cE)\,.
$$
Reporting in the conditioning, we obtain
\begin{multline*}
\,P\big( 
D_{m\ltq}\geq
(\ln \ell)^m
,\,\cE\big)\,\geq\,
\cr
E\Big(
P\big( 
D_{m\ltq}\geq
(\ln \ell)^m
\,\big|\,
S_n, I_n, J_n,\,
1\leq n\leq m\ltq
\big)\,
P(\cE)
\Big)
\cr
\,=\,
P\big( 
D_{m\ltq}\geq
(\ln \ell)^m\big)\,
P\big(
\cE\big)\,.
\end{multline*}
%D_{m\ltq}\geq
%We apply next the FKG inequality to the process 
%$(D_n)_{n\geq 0}$ to get
%\begin{multline*}
%\,P\big( 
%D_{m\ltq}\geq
%(\ln \ell)^m\big)\,=\,
%P\big( 
%\forall i\in\um\quad
%D_{m\ltq}(i)\geq\ln\ell
%\big)\cr
%\,\geq\,
%\prod_{1\leq i\leq m}
%P\big( 
%%D_{m\ltq}\geq
%%(\ln \ell)^m
%D_{m\ltq}(i)\geq\ln\ell
%\big)\,.
%\end{multline*}
Next,
\begin{multline*}
P\big( 
D_{m\ltq}\geq
(\ln \ell)^m\big)\,
\,=\,
1-
P\Big( 
\exists i\in\um\quad
D_{m\ltq}(i)<
\ln \ell\Big)
\cr
\,
\,\geq\,
1-
\sum_{1\leq i\leq m}
\kern-3pt
P\big( 
D_{m\ltq}(i)<\ln\ell\big)\,.
\end{multline*}
From proposition~\ref{nemarg},
$$\forall i\in\um\qquad
P\big(
D_{m\ltq}(i)<\ln\ell\big)
\,=\,
P\big(
Y_{N(m\ltq)}<\ln\ell\big)\,,
$$
therefore
$$
P\big( 
D_{m\ltq}\geq
(\ln \ell)^m\big)\,
\geq\,
1-m
P\big(
Y_{N(m\ltq)}<\ln\ell\big)\,,
$$
We estimate now the probability of the event
$\smash{\{\,Y_{N(m\ltq)}<\ln\ell\,\}}$, or rather its complement. 
%We have
%\begin{multline*}
%P\big(
%Y_{N(m\ltq)}<\ln\ell\big)\,\leq\,
%P\big(
%{N(m\ltq)}<\sqrt\ell\big)\cr
%\,+\,
%P\big(
%Y_{N(m\ltq)}<\ln\ell,\,
%{N(m\ltq)}\geq\sqrt\ell\big)
%\,.
%\end{multline*}
The random variable 
$\smash{N(m\ltq)}$
follows the binomial law with parameter $m\ltq$
and $1/m$, therefore
\begin{multline*}
P\big(
{N(m\ltq)}<\sqrt\ell\big)\,\leq\,
P\Big(
\exp\big(-{N(m\ltq)}\big)>\exp\big(-\sqrt\ell\big)\Big)\cr
\,\leq\,
\exp\big(\sqrt\ell\big)\Big(
\frac{1}{em}+1-
\frac{1}{m}\Big)^{m\ltq}
\cr
\,\leq\,
\exp\Big(\sqrt\ell
+
{\ltq}\big(
\frac{1}{e}-
{1}\big)\Big)
\,\leq\,
\exp\Big(\sqrt\ell
-\frac{1}{2}
{\ltq}\Big)
\,.
\end{multline*}
%\begin{align*}
%E(N(m\ltq))\,&=\,
%{\ltq}\,,\cr
%\var
%(N(m\ltq))\,&=\,m\ltq
%\frac{1}{m}
%\Big(1-
%\frac{1}{m}\Big)\,\leq\,
%{\ltq}
%\,.
%\end{align*}
%Applying Chebyshev's inequality, we get
%for $\ell$ large enough,
%\begin{multline*}
%P\big(
%N({m\ltq})\leq\sqrt\ell\big)
%\,\leq\,
%P\bigg(
%\Big|N({m\ltq})-
%{\ltq}\Big|
%\geq
%{\ltq}-
%\sqrt\ell\bigg)
%\cr
%\,\leq\,
%\frac{\displaystyle\ltq}{\displaystyle
%\Big({\ltq}
%-\sqrt\ell\Big)^2}
%\,\leq\,
%\frac{\displaystyle 2}{\displaystyle\ltq}\,.\hfil
%\end{multline*}
Combining the previous estimates
and using proposition~\ref{firstdesc},
we have
\begin{multline*}
P\big(
Y_{N(m\ltq)}\geq\ln\ell\big)\,\geq\,
P\big(
Y_{N(m\ltq)}\geq\ln\ell,\,
N({m\ltq})\geq\sqrt\ell\big)
\cr
\,\geq\,
\sum_{t=\sqrt{\ell}}^{+\infty}
P\big(
Y_{N(m\ltq)}\geq\ln\ell,\,
N({m\ltq})=t\big)
\cr
\,=\,
\sum_{t=\sqrt{\ell}}^{+\infty}
P\big(
Y_{t}\geq\ln\ell\big)\,P\big(
N({m\ltq})=t\big)
\cr
\,\geq\,
\Big(
1-
\exp\Big(-\frac{1}{2}(\ln\ell)^2\Big)
\Big)
\Big(1-
\exp\Big(\sqrt\ell
-\frac{1}{2}
{\ltq}\Big)\Big)
\cr
\,\geq\,
1-
\exp\Big(-\frac{1}{3}(\ln\ell)^2\Big)
%\Big)
\,,
\end{multline*}
the last inequality being valid for $\ell$ large enough.
Putting the previous estimates together, we have
$$
\,P\big( 
D_{m\ltq}\geq
(\ln \ell)^m
,\,\cE\big)\,\geq\,
\Big(
1-
m\exp\Big(-\frac{1}{3}(\ln\ell)^2\Big)\Big)
\,\Big(1-\frac{p}{\kappa}\Big)^{
m\ltq\ln\ell}\,.$$
We study next 
$$
P\big(
\tau^*>
{\kappa^{\ell(1-\ve)}}
%-m\ltq
\,|\,
D_{0}=
(\ln \ell)^m
\big)\,.$$
We give first an estimate showing that a visit to $0$ becomes very unlikely 
if the starting point is far from~$0$.
\begin{lemma}\label{premhit}
%For any $n\geq 0$, any $b\in\ul$, we have
For $b\in\ul$, we have
$$\forall n\geq 0\qquad P\big(
\tau^*\leq n
\,|\,
D_{0}=(b)^m
\big)
%\left(
%\begin{matrix}
%b\\
\,\leq\,
nm
\frac{\cB(0)}{\cB(b)}\,.
$$
\end{lemma}
\begin{proof}
Let $n\geq 0$ and $b\in\ul$.
We write
\begin{multline*}
P\big(
\tau^*\leq n
\,|\,
D_{0}=(b)^m
\big)
%\left(
%\begin{matrix}
%b\\
%\vdots\\
%b
%\end{matrix}
%\right)
\,=\,\cr
P\big(
\exists\,t\leq n\quad\exists\,i\in\um\quad
D_t(i)=0
\,|\,
D_{0}=(b)^m
\big)
\cr
\,\leq\,
\sum_{ 1\leq t\leq n }
\sum_{ 1\leq i\leq m }
P\big(
D_t(i)=0
\,|\,
D_{0}=(b)^m
\big)\,.
\end{multline*}
By proposition~\ref{nemarg}, for any $t\geq 0$,
any
$i\in\um$,
$$P\big(
D_t(i)=0
\,|\,
D_{0}=(b)^m
\big)\,=\,
P\big(
Y_{N(t)}=0
\,|\,
Y_{0}=b
\big)\,.
$$
Using proposition~\ref{ivy} and
lemma~\ref{majinv},
we have
$$P\big(
Y_{N(t)}=0
\,|\,
Y_{0}=b
\big)\,\leq\,
\frac{\cB(0)}{\cB(b)}\,.
$$
Putting together the previous inequalities, we get
$$P\big(
\tau^*\leq n
\,|\,
D_{0}=(b)^m
\big)
%\left(
%\begin{matrix}
%b\\
\,\leq\,
nm
\frac{\cB(0)}{\cB(b)}\,
$$
as requested.
\end{proof}

\noindent
Let $\ve'>0$.
Now
\begin{multline*}
P\Big(
\tau^*>
{\kappa^{\ell(1-\ve)}}
%-m\ltq
\,\big|\,
D_{0}=
(\ln \ell)^m
\Big)
\cr
\,\geq\,
P\Big(
%\forall t\in
%\{\,\ell^2m,\dots,
%{\kappa^{\ell(1-\ve)}}
%-m\ltq\,\}\quad
%D_t\in\cN,\,
\tau^*>m\ell^2 ,\,
D_t\in\cN\text{ for }
m\ell^2\leq t\leq
{\kappa^{\ell(1-\ve)}}
%-m\ltq\,\}\quad
\,\big|\,
D_{0}=
(\ln \ell)^m
\Big)\cr
\,=\,\sum_{d\in\cN}
P\Big(
\begin{matrix}
\tau^*>m\ell^2 ,\, D_{m\ell^2}=d\\
D_t\in\cN\text{ for }
m\ell^2\leq t\leq
{\kappa^{\ell(1-\ve)}}
\end{matrix}
\,\Big|\,
D_{0}=
(\ln \ell)^m
\Big)\cr
\,\geq\,
\sum_{d\geq (\lk(1-\ve'))^m}
P\Big(
D_t\in\cN\text{ for }
m\ell^2\leq t\leq
{\kappa^{\ell(1-\ve)}}
\,\big|\,
\tau^*>m\ell^2 
,\,
D_{m\ell^2}=d
\Big)
\cr
\hfill \times
P\Big(
\tau^*>m\ell^2, \,
D_{m\ell^2}=d\,
\big|\,
D_{0}=
(\ln \ell)^m
\Big)\,.
\end{multline*}
Using the Markov property
and the monotonicity of the process $(D_t)_{t\geq 0}$, we have
for
${d\geq (\lk(1-\ve'))^m}$,
\begin{multline*}
P\Big(
D_t\in\cN\text{ for }
m\ell^2\leq t\leq
{\kappa^{\ell(1-\ve)}}
\,\big|\,
\tau^*>m\ell^2 
,\,
D_{m\ell^2}=d
\Big)
\cr
\,=\,
P\Big(
\forall t\in\{\,0,\dots,
{\kappa^{\ell(1-\ve)}}
-m\ell^2
%\tau^*>
%{\kappa^{\ell(1-\ve)}}
%-m\ltq
\,\}\quad
D_t\in\cN\,\big|\,
D_{0}=d\Big)\cr
\,=\,
P\Big(
\tau^*>
{\kappa^{\ell(1-\ve)}}
-m\ell^2
\,\big|\,
D_{0}=d\Big)
\,\geq\,P\Big(
\tau^*>
{\kappa^{\ell(1-\ve)}}
\,\big|\,
D_{0}=d\Big)
%\,.
%\end{multline*}
%In the neutral case, 
%by lemma~\ref{monphimor},
%the map $\Phi_H$ used to build
%the coupling of the process $(D_n)_{n\geq 0}$ is non--decreasing.
%In particular,
%the coupling is monotone with respect to the initial condition
%and
%for 
%$d\geq (\lk(1-\ve'))^m$,
%$$
%P\big(
%\tau^*>
%{\kappa^{\ell(1-\ve)}}
%\,\big|\,
%D_{0}=d\big)
\cr
\,\geq\,
P\big(
\tau^*>
{\kappa^{\ell(1-\ve)}}
\,\big|\,
D_{0}=
(\lk(1-\ve'))^m
\big)\,.
\end{multline*}
Reporting in the previous sum, we get
\begin{multline*}
P\Big(\tau^*>
{\kappa^{\ell(1-\ve)}}
\,\big|\,
D_{0}=(\ln\ell)^m\Big)
\,\geq\,
P\big(
\tau^*>
{\kappa^{\ell(1-\ve)}}
\,\big|\,
D_{0}=
(\lk(1-\ve'))^m
\big)
\cr\hfill\times
P\Big(
\tau^*>m\ell^2 
,\,
D_{m\ell^2}\geq
(\lk(1-\ve'))^m
\,\big|\,
D_{0}=
(\ln \ell)^m
\Big)\,.
\end{multline*}
We first take care of the last probability. 
%We could use directly the correlation
%inequality at the process level (see the corollary in \cite{Har}). 
%Yet it is enough to write
We write
\begin{multline*}
P\Big(
\tau^*>m\ell^2 
,\,
D_{m\ell^2}\geq
(\lk(1-\ve'))^m
\,\big|\,
D_{0}=
(\ln \ell)^m
\Big)\,\geq\,\cr
P\Big(
D_{m\ell^2}\geq
(\lk(1-\ve'))^m
\,\big|\,
D_{0}=
(\ln \ell)^m
\Big)-
P\Big(
\tau^*\leq m\ell^2 
\,\big|\,
D_{0}=
(\ln \ell)^m
\Big)\,.
\end{multline*}
To control the last term,
we use the inequality of lemma~\ref{premhit} with $n=m\ell ^2$
and $b=\ln\ell$:
$$P\big(
\tau^*\leq m\ell^2
\,|\,
D_{0}=(\ln\ell)^m
\big)
%\left(
%\begin{matrix}
%b\\
\,\leq\,
(m\ell) ^2 
\frac{\cB(0)}{\cB(\ln\ell)}\,.
$$
%We estimate the right--hand side with the help of 
By lemma~\ref{exco}, we have
$$\frac{\cB(0)}{\cB(\ln\ell)}\,
\leq\,\Big(\frac{2\ln\ell}{\ell}\Big)^{\ln\ell}\,
$$
whence
$$P\big(
\tau^*\leq m\ell^2
\,|\,
D_{0}=(\ln l)^m
\big)\,
\leq\,
(m\ell) ^2 
\Big(\frac{2\ln\ell}{\ell}\Big)^{\ln\ell}\,.
$$
For the other term, we use
the monotonicity of the process $(D_t)_{t\geq 0}$
%and the fact that it has positive correlations (by corollary~\ref{corneu})
to get
\begin{multline*}
P\Big(
D_{m\ell^2}\geq
(\lk(1-\ve'))^m
%,\,\tau^*>m\ell^2 
\,\big|\,
D_{0}=
(\ln \ell)^m
\Big)
\cr
\,\geq\,
P\Big(
D_{m\ell^2}\geq
(\lk(1-\ve'))^m
%,\,\tau^*>m\ell^2 
\,\big|\,
D_{0}=
(0)^m
\Big)
%\,\geq\,
%\prod_{1\leq i\leq m}
%P\Big(
%D_{m\ell^2}(i)\geq
%\lk(1-\ve')
%\,\big|\,
%D_{0}=
%(0)^m
\cr
\,=\,
1-
P\Big( 
\exists i\in\um\quad
D_{m\ell^2}(i)<
\lk(1-\ve')
\,\big|\,
D_{0}=
(0)^m
\Big)
\cr
\,
\,\geq\,
1-
\sum_{1\leq i\leq m}
\kern-3pt
P\big( 
D_{m\ell^2}(i)<
\lk(1-\ve')
\,\big|\,
D_{0}=
(0)^m
\Big)
\,.
\end{multline*}
From proposition~\ref{nemarg}, for any $i\in\um$,
$$%\forall i\in\um\qquad
P\Big(
D_{m\ell^2}(i)<
\lk(1-\ve')
\,\big|\,
D_{0}=
(0)^m
\Big)
\,=\,
P\Big(
Y_{N(m\ell^2)}<\lk(1-\ve')
\,\big|\,Y_0=0
\Big)\,,
$$
therefore
\begin{multline*}
P\Big(
D_{m\ell^2}\geq
(\lk(1-\ve'))^m
%,\,\tau^*>m\ell^2 
\,\big|\,
D_{0}=
(\ln \ell)^m
\Big)
\cr
\,\geq\,
1-m
P\Big(
Y_{N(m\ell^2)}<\lk(1-\ve')
\,\big|\,Y_0=0
\Big)\,.
\end{multline*}
%
%$$
%P\big( 
%D_{m\ltq}\geq
%(\ln \ell)^m\big)\,
%\geq\,
%1-m
%P\big(
%Y_{N(m\ltq)}<\ln\ell\big)\,,
%$$
We estimate now the probability of the event
$\smash{\{\,Y_{N(m\ell^2)}<\lk(1-\ve')\,\}}$, or rather its complement. 
%By proposition~\ref{nemarg}, for $i\in\um$,
%$$P\Big(
%D_{m\lq}(i)\geq\lkep
%\,\big|\,
%D_{0}=
%(0)^m\Big)
%\,=\,
%P\big(
%Y_{N(m\lq)}\geq\lkep\,|\,Y_0=0\big)\,.
%$$
The random variable 
${N(m\lq)}$
follows the binomial law with parameter $m\lq$
and $1/m$, therefore
%\begin{align*}
%E(N(m\lq))\,&=\,
%{\lq}\,,\cr
%\var
%(N(m\lq))\,&=\,m\lq
%\frac{1}{m}
%\Big(1-
%\frac{1}{m}\Big)\,\leq\,
%{\lq}
%\,.
%\end{align*}
%Applying Chebyshev's inequality, we get
%for $\ell$ large enough,
%\begin{multline*}
%P\Big(
%N({m\lq})\leq\frac{4\ell}{a\ve'}\Big)
%\,\leq\,
%P\bigg(
%\Big|N({m\lq})-
%{\lq}\Big|
%\geq
%{\lq}-
%\frac{4\ell}{a\ve'}
%\bigg)
%\cr
%\,\leq\,
%\frac{\displaystyle\lq}{\displaystyle
%\Big({\lq}-
%\frac{4\ell}{a\ve'}
%\Big)^2}
%\,\leq\,
%\frac{\displaystyle 2}{\displaystyle\lq}\,.\hfil
%\end{multline*}
\begin{multline*}
P\Big(
{N(m\ell^2)}<
\frac{4\ell}{a\ve'}
\Big)\,\leq\,
P\Big(
\exp\big(-{N(m\ell^2)}\big)>\exp\big(-
\frac{4\ell}{a\ve'}
\big)\Big)\cr
\,\leq\,
\exp\big(
\frac{4\ell}{a\ve'}
\big)\Big(
\frac{1}{em}+1-
\frac{1}{m}\Big)^{m\ell^2}
\cr
\,\leq\,
\exp\Big(
\frac{4\ell}{a\ve'}
+
{\ell^2}\big(
\frac{1}{e}-
{1}\big)\Big)
\,\leq\,
\exp\Big(
\frac{4\ell}{a\ve'}
-\frac{1}{2}
{\ell^2}\Big)
\,.
\end{multline*}
Combining the previous estimates
and using proposition~\ref{secdesc},
we have
\begin{multline*}
P\Big(
Y_{N(m\lq)}\geq
\lkep
\,\big|\,Y_0=0
\Big)\cr
\,\geq\,
P\Big(
Y_{N(m\lq)}\geq
\lkep
,\,
N({m\lq})\geq
\frac{4\ell}{a\ve'}
\,\big|\,Y_0=0
\Big)
\cr
\,\geq\,
\sum_{t={4\ell/(a\ve')}}^{+\infty}
P\Big(
Y_{N(m\lq)}\geq\lkep,\,
N({m\lq})=t
\,\big|\,Y_0=0
\Big)
\cr
\,=\,
\sum_{t={4\ell/(a\ve')}}^{+\infty}
P\Big(
Y_{t}\geq\lkep
\,\big|\,Y_0=0
\Big)\,P\big(
N({m\lq})=t\big)
\cr
\,\geq\,
\Big(
1-
\exp{
\big(
-c(\ve')\ell\big)}
\Big)
\Big(1-
\exp\Big(
\frac{4\ell}{a\ve'}
-\frac{1}{2}
{\ell^2}\Big)
\Big)\cr
\,\geq\,
%\Big(
1-
\exp{
\big(
-\frac{1}{2}c(\ve')\ell\big)}
%\Big)
\,,
\end{multline*}
where $c(\ve')>0$ and
the last inequality is valid for $\ell$ large enough.
Putting together the previous estimates, we obtain
\begin{multline*}
P\Big(
\tau^*>m\ell^2 
,\,
D_{m\ell^2}\geq
(\lk(1-\ve'))^m
\,\big|\,
D_{0}=
(\ln \ell)^m
\Big)\,\geq\,\cr
1-m\exp{
\big(
-\frac{1}{2}c(\ve')\ell\big)}
\,-\,
(m\ell) ^2 
\Big(\frac{2\ln\ell}{\ell}\Big)^{\ln\ell}\,.
\end{multline*}
It remains to study
$P\big(
\tau^*>
{\kappa^{\ell(1-\ve)}}
\,|\,
D_{0}=
(\lk(1-\ve'))^m
\big)\,.$
We use the inequality of lemma~\ref{premhit} with 
$n={\kappa^{\ell(1-\ve)}}$
and
$b=\lk(1-\ve')$:
$$P\big(
\tau^*\leq
{\kappa^{\ell(1-\ve)}}
\,\big|\,
D_{0}=
(\lk(1-\ve'))^m
\big)\,\leq\,
{\kappa^{\ell(1-\ve)}}m
\frac{\cB(0)}{\cB(
\lk(1-\ve')
)}\,.
$$
For $\ve'$ small enough, using the large deviation estimates
of lemma~\ref{gdb}, we see that there exists $c(\ve,\ve')>0$ such that,
for $\ell$ large enough,
$$P\big(
\tau^*\leq
{\kappa^{\ell(1-\ve)}}
\,\big|\,
D_{0}=
(\lk(1-\ve'))^m
\big)\,\leq\,
\exp(-c(\ve,\ve')\ell)\,.$$
Collecting all the previous estimates,
we conclude that, for $\ell$ large enough,
$$\displaylines{
P\Big(\tau^*>
{\kappa^{\ell(1-\ve)}}
\,\big|\,
D_0=(1)^m 
\Big)
%D_{0}=(\ln\ell)^m\Big)
\,\geq\,
\Big(
1-
m\exp\Big(-\frac{1}{3}(\ln\ell)^2\Big)\Big)
\,\Big(1-\frac{p}{\kappa}\Big)^{
m\ltq\ln\ell}\cr
\hfill\times
\Big(1-\exp(-c(\ve,\ve')\ell)\Big)
\bigg(
1-m\exp{
\big(
-\frac{1}{2}c(\ve')\ell\big)}
\,-\,
(m\ell) ^2 
\Big(\frac{2\ln\ell}{\ell}\Big)^{\ln\ell}
\bigg)
\,.
}$$
Moreover, by Markov's inequality,
$$
E\Big(\tau^*\,|\,
D_0=(1)^m 
\Big)
\,\geq\,
{\kappa^{\ell(1-\ve)}}\,
P\Big(\tau^*\geq
{\kappa^{\ell(1-\ve)}}
\,\big|\,
D_0=(1)^m 
\Big)\,.$$
It follows that
$$\liminf_{
\genfrac{}{}{0pt}{1}{\ell,m\to\infty,\,
q\to 0
}
{{\ell q} \to a,\,
\frac{\scriptstyle m}{\scriptstyle \ell}\to\alpha
}
}
\,\frac{1}{\ell}\ln
E\Big(\tau^*\,|\,
D_0=(1)^m 
\Big)
\,\geq\,(1-\ve)\ln\kappa\,.$$
Letting $\ve$ go to $0$ yields the desired lower bound.
\end{proof}
\vfill\eject
\section{Synthesis}\label{secsyn}
As in theorem~\ref{mainth},
we suppose that 
%$\ell$ goes to $\infty$, $m$ goes to $\infty$ and $q$ goes to $0$ 
$$\ell\to +\infty\,,\qquad m\to +\infty\,,\qquad q\to 0\,,$$
%$$\ell\to +\infty\,,\qquad q\to 0\,,$$
in such a way that
%$${\ell q} \to a\in [0,+\infty]\,.$$
$${\ell q} \to a\in ]0,+\infty[\,,
\qquad\frac{m}{\ell}\to\alpha\in [0,+\infty]\,.$$
%We consider the regime where
%$\ell,m$ go to $\infty$ 
%and
We put now together the estimates of 
sections~\ref{bide} and~\ref{disc} in order to
evaluate the formula for the invariant measure
obtained at the end of 
section~\ref{bounds}.
For $\theta=\ell,1$, we rewrite this formula as
\begin{multline*}
\int_{\textstyle\pml}
f\Big( \frac{ o(0) }{m} \Big)
\,d\mu_O^\theta(o)
\,=\,
\frac{
E\big({\tau_0}
\,\big|\,
Z^\theta_0= 1\big)}
{
\displaystyle
E\big(\tau^*\,|\,
O^\theta_0=\oto
\big)+
E\big({\tau_0}
\,\big|\,
Z^\theta_0= 1\big)
}
%\frac{\displaystyle
%E\big(\tau_1\,|\,
%Z^>_0= 0
%\big)+
%\,{1 }
%+
%{ M_H(\theta,0) }
%E\big({\tau_0}
%\,\big|\,
%Z^\theta_0= 1\big)
%}
\cr
\times\,\Big(
\frac{1}
{ M_H(\theta,0) 
E\big({\tau_0}
\,\big|\,
Z^\theta_0= 1\big)
}
+1\Big)
\sum_{i=1}^m 
f\Big( \frac{ i }{m} \Big)
\nu^\theta(i)
\,.
\end{multline*}
%$q$ to $0$ 
%in such a way that $m/\ell$ and $\ell q$ are kept constant.
%We set
By proposition~\ref{numest}, 
$$\lim_{
\genfrac{}{}{0pt}{1}{\ell,m\to\infty}
{q\to 0,\,
{\ell q} \to a}
}
\Big(
\frac{1}
{ M_H(\theta,0) 
E\big({\tau_0}
\,\big|\,
Z^\theta_0= 1\big)
}
+1\Big)
\sum_{i=1}^m 
f\Big( \frac{ i }{m} \Big)
\nu^\theta(i)
\,=\,
f\big(\rho^*(a)\big)\,.
$$
%$q$ to $0$ 
%in such a way that $m/\ell$ and $\ell q$ are kept constant.
%the term on the second line 
%converges towards $f\big(\rho^*(a)\big).$
By corollary~\ref{exta},
$$\lim_{
\genfrac{}{}{0pt}{1}{\ell,m\to\infty}
{q\to 0,\,
{\ell q} \to a}
}
\,\frac{1}{m}\ln
E(\tau_0\,|\,Z^\theta_0=1)\,=\,
\int_0^{\rho^*(a)}\ln \phi(\exa,0,s)\,ds\,.$$
By proposition~\ref{bdd}, for $\alpha\in[0,+\infty[$,
$$\lim_{
\genfrac{}{}{0pt}{1}{\ell,m\to\infty,\,
q\to 0
}
{{\ell q} \to a,\,
\frac{\scriptstyle m}{\scriptstyle \ell}\to\alpha
}
}
\,\frac{1}{\ell}\ln
E\big(\tau^*\,|\,
O^\theta_0=\oto
\big)
\,=\,\ln\kappa\,.$$
For the case $\alpha=+\infty$, 
by corollary~\ref{majext} and
proposition~\ref{visitz},
$$\lim_{
\genfrac{}{}{0pt}{1}{\ell,m\to\infty,\,
q\to 0
}
{{\ell q} \to a,\,
\frac{\scriptstyle m}{\scriptstyle \ell}\to\infty
}
}
\,\frac{1}{\ell}\ln
\Big(\frac{1}{m}
E\big(\tau^*\,|\,
O^\theta_0=\oto
\big)\Big)
\,\leq\,
\ln\kappa\,.$$
These estimates allow to evaluate the 
%first fraction which involves
ratio between the discovery time and the persistence time. 
We define a function 
$\phi:\,]0,+\infty[\to
[0,+\infty]$ by setting $\phi(a)=0$ if $a\geq\ln\sigma$ and
$$
\forall a<\ln\sigma\qquad
\phi(a)\,=\,
%\frac{
%\displaystyle \ln\kappa}
{\displaystyle 
\int_0^{\rho^*(a)}\ln \phi(\exa,0,s)\,ds}\,.$$
We have then, for $\alpha\in[0,+\infty[$ or $\alpha=+\infty$,
%We consider two cases:
%\medskip
%
%\noindent
%$\bullet\quad$ If $\alpha<\phi(a)$ then
$$\lim_{
\genfrac{}{}{0pt}{1}{\ell,m\to\infty,\,
q\to 0
}
{{\ell q} \to a,\,
\frac{\scriptstyle m}{\scriptstyle \ell}\to\alpha
}
}
\frac{
E\big({\tau_0}
\,\big|\,
Z^\theta_0= 1\big)}
{
\displaystyle
E\big(\tau^*\,|\,
O^\theta_0=\oto
\big)
}\,=
\begin{cases}
\quad 0 &\text{if }\alpha\,\phi(a)<\ln\kappa\\
\,\,+\infty &\text{if }\alpha\,\phi(a)>\ln\kappa\\
\end{cases}
$$
Notice that the result is the same for $\theta=\ell$ and $\theta=1$.
%\medskip
%
%\noindent
%$\bullet\quad$ If $\alpha>\phi(a)$ then
%\begin{cases}
%\displaystyle\frac{(1-\sigma(1-e^{-a}))
%\,\ln\kappa
%}
Putting together the bounds on $\nu$ given in section~\ref{bounds} and
the previous considerations, we conclude that
$$\lim_{
\genfrac{}{}{0pt}{1}{\ell,m\to\infty,\,
q\to 0
}
{{\ell q} \to a,\,
\frac{\scriptstyle m}{\scriptstyle \ell}\to\alpha
}
}
\int_{[0,1]} f\,d\nu\,=\,
\begin{cases}
\quad 0 &\text{if }\alpha\,\phi(a)<\ln\kappa\\
f\big(\rho^*(a)\big)
&\text{if }\alpha\,\phi(a)>\ln\kappa\\
\end{cases}
$$
This is valid for any 
continuous non--decreasing function 
$f:[0,1]\to\R$ 
such that $f(0)=0$. 
To obtain the statement of theorem~\ref{stronger}, it remains
to
compute the integral.
For $a<\ln\sigma$,
\begin{multline*}
\phi(a)\,=\,\int_0^{\rho^*(a)}\ln \phi(\exa,0,s)\,ds
\cr
\,=\,
\int_0^{\rho^*(a)}\ln 
\frac{\displaystyle \sigma\exa (1-s)}
{\displaystyle \sigma(1-\exa)s+(1-s)}
\,ds
\cr
\,=\,
\frac
{ \displaystyle \sigma(1-e^{-a})
\ln\frac{\displaystyle\sigma(1-e^{-a})}{\displaystyle\sigma-1}
%\big(\ln({\sigma(1-e^{-a})})-\ln({\sigma-1})\big)
+\ln(\sigma e^{-a})}
{
\displaystyle (1-\sigma(1-e^{-a}))
 }
\end{multline*}
and we are done.
\appendix
\vfill\eject
\addtocontents{toc}{\protect\enlargethispage{\baselineskip}}
\section{Appendix on Markov chains}
In this appendix, we recall classical definitions and results from the
theory of Markov chains with finite state space.
The goal is to clarify the objects involved in the definition of the model, 
and to state the fundamental general results used in the proofs.
This material can be found in any reference book on Markov chains,
for instance 
\cite{Breiman}, 
\cite{Fe1},
%\cite{Fe2},
\cite{SH}. 
The definitions and results on monotonicity, coupling and the FKG
inequality are exposed in
the books of
Liggett \cite{LIG}
and
Grimmett~\cite{GRI}.
\medskip

\noindent
{\bf Construction of continuous time Markov processes.}
The most convenient way to define a continuous time process is to give
its infinitesimal generator. 
The infinitesimal generator of a Markov process $(X_t)_{t\geq 0}$ with
values in a finite state space $\cE$ is the linear operator $L$ acting
on the functions from $\cE$ to
%$\smash{\left({\cal A}^\ell\right)^m}$
$\mathbb R$ defined as follows.
For any function $\phi:\cE\to{\mathbb R}$, any $x\in \cE$, 
$$L\phi(x)\,=\,
\lim_{t\to 0}\,\frac{1}{t}\Big(E\big(\phi(X_t)|X_0=x\big) -\phi(x)\Big)\,.
$$
It turns out that the law of the process $(X_t)_{t\geq 0}$
is entirely determined by the generator~$L$. Therefore all the probabilistic
results on the process
$(X_t)_{t\geq 0}$ can in principle be derived working only with
its infinitesimal generator.

In the case where the state space of the process is finite, the situation is quite
simple and it is possible to provide direct constructions of a process 
having a specific infinitesimal generator. These constructions are not unique,
but they provide more insight into the dynamics.
Suppose that the generator $L$ is given by
$$\forall x\in \cE\qquad L\phi(x)\,=\,
\sum_{y\in \cE}
c(x,y)
\big(\phi(y)-\phi(x)\big)\,.$$
The evolution of a process
$(X_t)_{t\geq 0}$ 
having $L$ as infinitesimal generator can loosely be described as follows.
Suppose that $X_t=x$. Let
$$c(x)\,=\,\sum_{y\neq x}c(x,y)\,.$$
Let $\tau$ be a random variable whose law is exponential with parameter $c(x)$:
$$\forall s\geq 0\qquad
P(\tau\geq s)\,=\, \exp(-c(x)s)\,.$$
The process waits at $x$ until time $t+\tau$.
At time
$t+\tau$, it jumps to a state $y\neq x$ chosen according to the following law:
$$P\big(X_{t+\tau}=y\big)\,=\,
\frac{c(x,y)}{c(x)}\,.$$
The same scheme is then applied starting from $y$. In this construction, the
waiting times $\tau$ and the jumps are all independent.
\medskip

\noindent
{\bf Construction of discrete time Markov chains.}
To build a discrete time Markov chain, we need only to define
its transition mechanism. When the state space $\cE$ is finite, this amounts
to giving its transition matrix
$$\big(p(x,y), \, x,y\in \cE\big)\,.$$
The only requirement on $p$ is that it is a stochastic matrix, i.e., it
satisfies
$$\displaylines{
\forall x,y\in \cE\qquad 0\leq p(x,y)\leq 1\,,\cr
\forall x\in \cE\qquad \sum_{y\in \cE}p(x,y)\,=\,1\,.
}$$
%\medskip
In the sequel, we consider
a discrete time Markov chain 
$(X_t)_{t\geq 0}$ 
with
values in a finite state space $\cE$ and
with transition matrix
$(p(x,y))_{x,y\in \cE}$.
\medskip

%$\big(p(x,y), \, x,y\in \cE\big)$.
\noindent
{\bf Invariant probability measure.}
If the Markov chain is
irreducible and aperiodic,
%Such a Markov chain 
then it admits a unique
invariant probability 
measure 
$\mu$, i.e., the set of equations
$$\mu(y)\,=\,\sum_{x\in\cE}\mu(x)\,p(x,y)\,,
\qquad
y\in\cE\,,$$
admits a unique solution.
The Markov chain 
$(X_t)_{t\geq 0}$ 
is said to be reversible with respect to
a probability measure $\nu$ if it satisfies the 
detailed balanced conditions:
$$\forall x,y\in\cE\qquad
\nu(x)\,p(x,y)\,=\,
\nu(y)\,p(y,x)\,.$$
If the Markov chain
$(X_t)_{t\geq 0}$ is reversible
with respect to a probability measure $\nu$, then $\nu$
is an invariant probability measure for
$(X_t)_{t\geq 0}$. In case
$(X_t)_{t\geq 0}$ is in addition irreducible and aperiodic,
then $\nu$ is the unique invariant probability measure of the chain. 
\begin{lemma}\label{majinv}
Suppose that $\mu$ is an invariant probability measure for the Markov
chain $(X_t)_{t\geq 0}$. We have then
$$\forall x,y\in\cE\quad
\forall t\geq 0\qquad
{\mu(x)}P\big(X_t=y\,|\,X_0=x\big)
\,\leq\,{\mu(y)}\,.$$
%\,\leq\,\frac{\mu(y)}{\mu(x)}\,.$$
\end{lemma}
\begin{proof}
The proof is done by induction on $t$. For $t=0$, we have
$$P\big(X_0=y\,|\,X_0=x\big)\,=\,0\quad
\text{if}\quad y\neq x\,,$$
and the result holds. Suppose it has been proved 
until time $t\in\mathbb N$.
We have then, for $x,y\in\cE$,
\begin{multline*}
{\mu(x)}\,
P\big(X_{t+1}=y\,|\,X_0=x\big)\,=\,
\sum_{z\in\cE}
{\mu(x)}\,
P\big(X_{t+1}=y,\,X_t=z\,|\,X_0=x\big)
\cr
\,=\,
\sum_{z\in\cE}
{\mu(x)}\,
P(X_t=z\,|\,X_0=x\big)
\,
P\big(X_{t+1}=y\,|\,X_t=z)
\cr
\,\leq\,
\sum_{z\in\cE}
{\mu(z)}\,
p(z,y)
\,=\,\mu(y)\,
\end{multline*}
and the claim is proved at time $t+1$.
\end{proof}
%\begin{theorem}[Ergodic theorem for Markov chains]\label{ergodic}

\noindent
We state next the
ergodic theorem for Markov chains.
We consider only the case where the state space $\cE$ is finite.
\begin{theorem}\label{ergodic}
Suppose that the Markov chain 
$(X_t)_{t\geq 0}$ is irreducible aperiodic.
Let
$\mu$ be its invariant probability measure.
For any initial distribution $\mu_0$,
for any function $f:\cE\to\R$, we have, with probability one,
$$%\forall x,y\in\cE\qquad
%\lim_{t\to\infty}P(X_t=y\,|\,X_0=x)\,=\,\mu(y)\,.$$
\lim_{t\to\infty} \,
\frac{1}{t}
\int_0^tf(X_s)\,ds\,=\,\int_{\cE}f(x)\,d\mu(x)\,.
$$
\end{theorem}
%Together with Cesaro's theorem, the ergodic theorem for Markov chains
%implies the following classical ergodic result:
%$$\forall x\in\cE\qquad
%\lim_{t\to\infty} \,
%\frac{1}{t}
%\int_0^tf(X_s)\,ds\,.$$
\medskip

\noindent
{\bf Lumping.} The basic lumping result
for Markov chains can be found in section~6.3 of 
the book of Kemeny and Snell \cite{KS}.
Let $(E_1,\dots,E_r)$ be a partition of $\cE$.
Let $f:\cE\to\{\,1,\dots,r\,\}$ be the function defined by
$$\forall x\in\cE\qquad
f(x)\,=\,
\begin{cases}
\quad 1 &\text{if}\quad x\in E_1\\
\quad\vdots  &\vdots\\
\quad r &\text{if }\quad x\in E_r\\
\end{cases}
\,.$$
The Markov chain
$(X_t)_{t\geq 0}$ is said to be lumpable with respect to the partition
$(E_1,\dots,E_r)$ if, for every initial distribution $\mu_0$ of $X_0$, the
process
$\smash{\big(f(X_t)\big)_{t\geq 0}}$ is a Markov chain on
$\{\,1,\dots,r\,\}$ whose transition probabilities do not depend on $\mu_0$.
\begin{theorem}[Lumping theorem]\label{lumpt}
A necessary and sufficient condition for the Markov chain
$(X_t)_{t\geq 0}$ to be lumpable with respect to the partition
$(E_1,\dots,E_r)$ is that, 
%for any
$$\forall i,j\in\{\,1,\dots,r\,\}\,\quad
\forall x,y\in E_i\,\qquad
\sum_{z\in E_j}p(x,z)\,=\,
\sum_{z\in E_j}p(y,z)\,.$$
Suppose that this condition holds. 
For $i,j\in\{\,1,\dots,r\,\}$, let us denote by $p_E(i,j)$ the common value
of the above sums. The process
$\smash{\big(f(X_t)\big)_{t\geq 0}}$ is then a Markov chain 
with transition matrix $(p_E(i,j))_{1\leq i,j\leq r}$.
\end{theorem}
\medskip

\noindent
{\bf Monotonicity.}
We recall some standard definitions concerning monotonicity 
for stochastic processes.
A classical reference
is Liggett's book
\cite{LIG},
especially for applications to particle systems. 
We consider
a discrete time Markov chain 
$(X_t)_{t\geq 0}$ 
with
values in a space $\cE$.
%values in a finite state space $\cE$.
%and with transition matrix $(p(x,y))_{x,y\in \cE}$.
We suppose that the state space $\cE$ is finite and
that it is equipped with a partial order $\leq$.
A function $f:\cE\to\R$ is non--decreasing if
$$\forall x,y\in\cE\qquad
x\leq y\quad\Rightarrow\quad f(x)\leq f(y)\,.$$
%Let $\mu,\nu$ be two probability measures on $\cE$.
%We say that $\nu$ dominates $\mu$,
%which we denote by $\mu\leq\nu$, if
%for any non--decreasing function $f$, we have
%$\mu(f)\leq \nu(f)$.
The Markov chain
$(X_t)_{t\geq 0}$ is said to be monotone if, 
for any non--decreasing function $f$, the function
$$x\in\cE\mapsto E\big(f(X_t)\,|\,X_0=x\big)$$
is non--decreasing.
\medskip

\noindent
{\bf Coupling.}
A natural way to prove monotonicity is to construct an adequate coupling.
%\begin{definition}
%\label{coupl}
A coupling is a family of processes
$(X_t^x)_{t\geq 0}$
indexed by 
$x\in\cE$, which are all defined on the same probability space, and such that, 
for $x\in\cE$, the process
$(X_t^x)_{t\geq 0}$ is the Markov chain starting from $X_0=x$.
The coupling is said to be monotone if
$$\forall x,y\in\cE\qquad
x\leq y\quad\Rightarrow\quad \forall t\geq 1\qquad X_t^x\leq X_t^y\,.$$
%Indeed,
%if there exists a coupling 
%$(X_t^x)_{t\geq 0}$, $x\in\cE$, of the Markov chains with different initial
%conditions such that
%\smallskip
%
%\noindent
%$\bullet $ for $x\in\cE$, the process
%$(X_t^x)_{t\geq 0}$ is the Markov chain starting from $X_0=x$,
%
%\noindent
%$\bullet $ for $x,y\in\cE$, if $x\leq y$, then
%$X_t^x\leq X_t^y$ for $t\geq 1$,
%\smallskip
%
%\noindent
If there exists a monotone coupling, 
then the 
Markov chain
%$(X_t)_{t\geq 0}$ 
is monotone.
%\end{definition}
\medskip

\noindent
{\bf FKG inequality.}
We consider the product space
$[0,1]^n$ equipped with the product order. 
Let $\mu$ be a probability measure on $[0,1]$ and let
us denote by $\mu^{\otimes n}$ the product probability measure on
$[0,1]^n$ whose marginals are equal to~$\mu$.
%A probability measure $\mu$ on $\cE$ is 
%said to have positive correlations if,
The Harris inequality, or
the FKG inequality 
in this context, 
says that,
for any non--decreasing functions $f,g:[0,1]^n\to\R$, we have
$$\int_{[0,1]^n} fg\,d\mu^{\otimes n}\,\geq\,
\int_{[0,1]^n} f\,d\mu^{\otimes n}\,
\int_{[0,1]^n} g\,d\mu^{\otimes n}\,.$$
%A standard example is the case of product measures.
%$[0,1]^n$ has positive correlations 
The case of Bernoulli product measures is exposed 
in section~$2.2$ of
%$[0,1]^n$ has positive correlations (see for instance section~$2.2$ of
Grimmett's book~\cite{GRI}.
%We quote next a simple criterion ensuring that a Markov process preserves
%positive correlations.
%We say 
%that every jump of 
%$(X_t)_{t\geq 0}$ is up or down if 
%$$
%\forall t\geq 0\quad
%\forall x,y\in\cE\qquad 
%P(X_{t+1}=y\,|\,X_t=x)>0\quad\Rightarrow\quad x\leq y\text{ or }x\geq y\,.$$
%We will use the following result of Harris \cite{Har}.
%\begin{theorem}\label{correqu}
%%Let $\mu_0$ be a probability measure on $\cE$ having positive correlations.
%Suppose that the
%Markov chain
%$(X_t)_{t\geq 0}$ is monotone and that every jump of 
%$(X_t)_{t\geq 0}$ is up or down.
%If the initial distribution of the 
%Markov chain
%$(X_t)_{t\geq 0}$ 
%has positive correlations, then, for any $t\geq 0$, 
%the law of $X_t$
%has also positive correlations.
%\end{theorem}

\vfill\eject
\noindent

%\nocite{*}
\bibliographystyle{plain}
%%\chapterspace{-3}
\bibliography{spr}
\vfill\eject
\printindex
\vfill\eject
 \thispagestyle{empty}
\tableofcontents
\end{document}